\PassOptionsToPackage{pdftex,
pdfversion=1.7,
pdfencoding=auto,
pdfnewwindow=true,
pdfusetitle=true,
bookmarks=true,
bookmarksnumbered=true,
bookmarksopen=true,
pdfpagemode=UseThumbs,
bookmarksopenlevel=1,
pdfpagelabels=false,
breaklinks=true
}{hyperref}
\PassOptionsToPackage{usenames,dvipsnames,table}{xcolor}
\documentclass[aps,english,10pt,superscriptaddress,onecolumn,twoside,longbibliography,pra,floatfix,fleqn,nofootinbib]{revtex4-2}

\usepackage[OT1]{fontenc}
\usepackage[utf8]{inputenx}
\usepackage{amsfonts}
\usepackage{amssymb}
\usepackage{amsthm}
\usepackage[intlimits,fleqn]{amsmath}
\usepackage{bm}
\usepackage{graphicx}
\usepackage[normalem]{ulem} 
\usepackage{paralist}
\usepackage{microtype}
\microtypecontext{spacing=nonfrench}
\microtypesetup{
expansion={true,nocompatibility},
protrusion={true,nocompatibility},
activate={true,nocompatibility},
tracking=true,
kerning=true,
spacing={true}
}

\usepackage{float}
\usepackage{wrapfig}
\usepackage{array}
\usepackage{ragged2e}
\usepackage{tabularx}
\usepackage{booktabs}
\usepackage[intlimits,fleqn]{mathtools}

\newcolumntype{R}{>{\raggedleft\arraybackslash}X}
\newcolumntype{C}{>{\centering\arraybackslash}X}
\newcolumntype{L}{>{\raggedright\arraybackslash}X}
\newcolumntype{J}{>{\justifying\arraybackslash}X}
\usepackage{adjustbox}
\usepackage{multirow}
\newcolumntype{T}[2]{
 >{\adjustbox{angle=#1,lap=\width-(#2)}\bgroup}
 l
 <{\egroup}
}
\usepackage{verbatim}

\usepackage[usenames,dvipsnames,table]{xcolor}
\definecolor{purple}{RGB}{128,0,128}
\definecolor{ultramarine}{RGB}{63, 0, 255}
\definecolor{medblue}{RGB}{0, 0, 100}
\definecolor{panblue}{RGB}{0,24,150}
\definecolor{carmine}{RGB}{150, 0, 24}
\definecolor{gray}{RGB}{150, 150, 150}
\definecolor{googleblue}{RGB}{34, 0, 204}
\definecolor{darkgreen}{RGB}{0, 80, 0}
\usepackage[pdftex,pdfversion=1.7,pdfencoding=auto,pdfnewwindow=true,pdfusetitle=true,bookmarks=true,bookmarksnumbered=true,bookmarksopen=true,pdfpagemode=UseThumbs,bookmarksopenlevel=1,pdfpagelabels=false]{hyperref}
\hypersetup{colorlinks,
linkcolor=carmine,
citecolor=darkgreen,
urlcolor=googleblue,
anchorcolor=OliveGreen}

\newcommand*{\mred}[1]{{\color{RawSienna}{\mathbf{#1}}}}
\newcommand*{\mgreen}[1]{{\color{OliveGreen}{\mathbf{#1}}}}
\newcommand*{\tred}[1]{{\color{carmine}{\textbf{#1}}}}
\newcommand*{\tblue}[1]{{\color{medblue}{\textbf{#1}}}}
\newcommand*{\tpurp}[1]{{\color{Plum}{{#1}}}}

\usepackage[capitalise]{cleveref}
\Crefname{eqs}{Eqs.}{Eqs.}
\creflabelformat{eqs}{(#2#1#3)}
\crefrangelabelformat{equation}{(#3#1#4-#5#2#6)}
\Crefmultiformat{equation}{Eqs.~(#2#1#3}{,#2#1#3)}{,#2#1#3}{,#2#1#3)}
\crefrangelabelformat{eqs}{(#3#1#4-#5#2#6)}
\Crefmultiformat{eqs}{Eqs.~(#2#1#3}{,#2#1#3)}{,#2#1#3}{,#2#1#3)}
\Crefname{example}{Example}{Examples}
\Crefname{section}{Sec.}{Secs.}

\newtheorem{theorem}{Theorem}
\newtheorem{lemma}[theorem]{Lemma}

\newtheorem{corollary}[theorem]{Corollary}
\newtheorem{definition}[theorem]{Definition}
\theoremstyle{definition}

\newcounter{example}[section]
\newenvironment{example}[1][]{\refstepcounter{example}\par\medskip
 \noindent \textbf{Example~\theexample}\hspace{1em}\rmfamily#1}{\par\medskip\par}
\Crefname{example}{Example}{Examples}
\creflabelformat{example}{#2#1#3}
\crefrangelabelformat{example}{#3#1#4-#5#2#6}
\Crefmultiformat{example}{Examples~#2#1#3}{, #2#1#3}{, #2#1#3 }{and #2#1#3}
\renewcommand{\theexample}{\arabic{example}}

\newcommand{\p}[2][]{{P_{#1}}\parenths{#2}}
\newcommand{\pfunc}[1]{P_{#1}}
\newcommand{\An}[2][]{{\mathsf{An}_{#1}}\parenths{#2}}
\newcommand{\Pa}[2][]{{\mathsf{Pa}_{#1}}\parenths{#2}}
\newcommand{\Ch}[2][]{{\mathsf{Ch}_{#1}}\parenths{#2}}
\newcommand{\SmallNamedFunction}[3][]{\operatorname{\mathsf{#2}}_{#1}\parenths{#3}}
\newcommand{\subgraph}[2][]{\SmallNamedFunction[#1]{SubDAG}{#2}}
\newcommand{\ansubgraph}[2][]{\SmallNamedFunction[#1]{AnSubDAG}{#2}}
\newcommand{\nodes}[1]{\SmallNamedFunction{Nodes}{#1}}
\newcommand{\obsnodes}[1]{\SmallNamedFunction{ObservedNodes}{#1}}
\newcommand{\latnodes}[1]{\SmallNamedFunction{LatentNodes}{#1}}
\newcommand{\inflations}[1]{\SmallNamedFunction{Inflations}{#1}}
\newcommand{\DAG}[1]{\SmallNamedFunction{DAG}{#1}}
\newcommand{\edges}[1]{\SmallNamedFunction{Edges}{#1}}
\newcommand{\aindep}{\ensuremath{\perp_d}}
\newcommand{\indep}{\perp\!\!\!\!\perp}
\newcommand{\eql}{\mathord{=}}
\DeclarePairedDelimiter{\parens}{\lparen}{\rparen}
\DeclarePairedDelimiter{\parenths}{\lparen}{\rparen}
\DeclarePairedDelimiter{\braces}{\lbrace}{\rbrace}
\DeclarePairedDelimiter{\bracks}{\lbrack}{\rbrack}
\DeclarePairedDelimiter{\expec}{\mathbb{E}[}{]}
\newcommand{\brackets}[1]{\braces*{#1}}

\usepackage[normalem]{ulem}
\newcommand{\st}[1]{\ifmmode\text{\sout{\ensuremath{#1}}}\else\sout{#1}\fi}

\begin{document}
\onecolumngrid

\title{The Inflation Technique for Causal Inference with Latent Variables}

\author{Elie Wolfe}
\email{ewolfe@perimeterinstitute.ca}
\affiliation{Perimeter Institute for Theoretical Physics, Waterloo, Ontario, Canada, N2L 2Y5}

\author{Robert W. Spekkens}
\email{rspekkens@perimeterinstitute.ca}
\affiliation{Perimeter Institute for Theoretical Physics, Waterloo, Ontario, Canada, N2L 2Y5}

\author{Tobias Fritz}
\email{tfritz@perimeterinstitute.ca}
\affiliation{Perimeter Institute for Theoretical Physics, Waterloo, Ontario, Canada, N2L 2Y5}

\date{\today}

\begin{abstract}

The problem of causal inference is to determine if a given probability distribution on observed variables is compatible with some causal structure. The difficult case is when the causal structure includes latent variables. We here introduce the \emph{inflation technique} for tackling this problem. An inflation of a causal structure is a new causal structure that can contain multiple copies of each of the original variables, but where the ancestry of each copy mirrors that of the original. 
To every distribution of the observed variables that is compatible with the original causal structure, we assign a family of marginal distributions on certain subsets of the copies that are compatible with the inflated causal structure. It follows that compatibility constraints for the inflation can be translated into compatibility constraints for the original causal structure. 
Even if the constraints at the level of inflation are weak, such as observable statistical independences implied by disjoint causal ancestry, the translated constraints can be strong. We apply this method to derive new inequalities whose violation by a distribution witnesses that distribution's incompatibility with the causal structure (of which Bell inequalities and Pearl's instrumental inequality are prominent examples). We describe an algorithm for deriving all such inequalities for the original causal structure that follow from ancestral independences in the inflation. For three observed binary variables with pairwise common causes, it yields inequalities that are stronger in at least some aspects than those obtainable by existing methods. We also describe an algorithm that derives a weaker set of inequalities but is more efficient. Finally, we discuss which inflations are such that the inequalities one obtains from them remain valid even for quantum (and post-quantum) generalizations of the notion of a causal model.

\end{abstract}

\maketitle
\onecolumngrid\clearpage
\tableofcontents

\section{Introduction}

Given a joint probability distribution of some observed variables, the problem of \tblue{causal inference} is to determine which hypotheses about the causal mechanism can explain the given distribution. Here, a causal mechanism may comprise both causal relations among the observed variables, as well as causal relations among these and a number of unobserved variables, and among unobserved variables only. {Causal inference has applications in all areas of science that use statistical data and for which causal relations are important. Examples include determining the effectiveness of medical treatments, sussing out biological pathways, making data-based social policy decisions, and possibly even in developing strong machine learning algorithms~\cite{pearl2009causality,spirtes2011causation,studeny2005probabilistic,
koller2009probabilistic,Pearl2018machinelearning}.} A closely related type of problem is to determine, for a given set of causal relations, the set of all distributions on observed variables that can be generated from them. 
A special case of both problems is the following decision problem: given a probability distribution and a hypothesis about the causal relations, determine whether the two are compatible: could the given distribution have been generated by the hypothesized causal relations? This is the problem that we focus on.
We develop necessary conditions for a given distribution to be compatible with a given hypothesis about the causal relations.

In the simplest setting, the causal hypothesis consists of a directed acyclic graph (DAG) {\em all} of whose nodes correspond to observed variables. In this case, obtaining a verdict on the compatibility of a given distribution with the causal hypothesis is simple: the compatibility holds if and only if the distribution is Markov with respect to the DAG, which is to say that the distribution features all of the conditional independence relations that are implied by $d$-separation relations among variables in the DAG. The DAGs that are compatible with the given distribution can be determined algorithmically~\cite{pearl2009causality}.\footnote{As illustrated by the vast amount of literature on the subject, the problem can still be difficult in practice, for example due to a large number of variables in certain applications or due to finite statistics.}

A significantly more difficult case is when one considers a causal hypothesis which consists of a DAG 
some of whose nodes correspond to \tblue{latent} (i.e., unobserved) variables, so that the set of observed variables corresponds to a strict subset of the nodes of the DAG. This case occurs, e.g.,~in situations where one needs to deal with the possible presence of unobserved confounders, and thus is particularly relevant for experimental design in applications. With latent variables, the condition that all of the conditional independence relations among the observed variables that are implied by $d$-separation relations in the DAG is still a necessary condition for compatibility of a given such distribution with the DAG, but in general it is no longer sufficient, and this is what makes the problem difficult. 


Whenever the observed variables in a DAG have finite cardinality\footnote{The cardinality of a variable is the number of possible values it can take.}, one may also restrict the latent variables in the causal hypothesis to be of finite cardinality as well, without loss of generality~\cite{rosset2016finite}. As such, the mathematical problem which one must solve to infer the distributions that are compatible with the hypothesis is a quantifier elimination problem for some finite number of variables, as follows: The probability distributions of the observed variables can all be expressed as functions of the parameters specifying the conditional probabilities of each node given its parents, many of which involve latent variables. If one can eliminate these parameters, then one obtains constraints that refer exclusively to the probability distribution of the observed variables. This is a {\em nonlinear} quantifier elimination problem. The Tarski-Seidenberg theorem provides an \emph{in principle} algorithm for an exact solution, but unfortunately the computational complexity of such quantifier elimination techniques is far too large to be practical, except in particularly simple scenarios~\cite{Geiger-Meek,LeeSpekkens}.\footnote{Techniques for finding approximate solutions to nonlinear quantifier elimination may help~\cite{ChavesPolynomial}.} Most uses of such techniques have been in the service of deriving compatibility conditions that are necessary but not sufficient, for both observational~\cite{Geiger-Meek99, Garcia2, Garcia, SSG} and interventionist data~\cite{In1, intervension1, intervension2}.

%

Historically, the insufficiency of the conditional independence relations for causal inference in the presence of latent variables was first
noted by Bell in the context of the hidden variable problem in quantum physics~\cite{bell1964einstein}. Bell considered an experiment for which considerations from relativity theory implied a very particular causal structure, and he derived an inequality
that any distribution compatible with this structure, and compatible with certain constraints imposed by quantum
theory, must satisfy. Bell also showed that this inequality was violated by distributions generated from entangled
quantum states with particular choices of incompatible measurements. Later work, by Clauser, Horne, Shimony and
Holt (CHSH) derived inequalities without assuming any facts about quantum correlations~\cite{CHSHOriginal}; this derivation can retrospectively be understood as the first derivation of a constraint arising from the causal structure of the Bell scenario alone~\cite{WoodSpekkens}. The CHSH inequality was the
first example of a compatibility condition that appealed to the strength of the correlations rather than simply the
conditional independence relations inherent therein.
Since then, many generalizations of the CHSH inequality have been derived for the same sort of causal structure~\cite{Brunner2013Bell}. The idea that such work is best understood as a contribution to the field of causal inference has only recently been put forward~\cite{WoodSpekkens,fritz2012bell,pusey2014gdag,BeyondBellII}, as has the idea that techniques developed by researchers in the foundations of quantum theory may be usefully adapted to causal inference\footnote{The current article being another example of the phenomenon \cite{BeyondBellII,ChavesNoSignalling,chaves2014informationinference,weilenmann2016entropic,kela2016covariance,ChavesPolynomial,TavakoliStarNetworks,RossetNetworks,TavakoliNoncyclicNetworks}.}.

Independently of Bell's work, Pearl later derived the \tblue{instrumental inequality}~\cite{pearl1995instrumental}, which provides a necessary condition for the compatibility of a distribution with a causal structure known as the \emph{instrumental scenario}. This causal structure comes up when considering, for instance, certain kinds of noncompliance in drug trials. More recently, \citet{steudel2010ancestors} derived an inequality which must hold whenever a distribution on $n$ variables is compatible with a causal structure where no set of more than $c$ variables has a common ancestor, for arbitrary $n,c \in \mathbb{N}$. More recent work has focused specifically on the simplest nontrivial case, with $n=3$ and $c=2$, a causal structure that has been called the Triangle scenario~\cite{fritz2012bell,chaves2014novel} (\cref{fig:TriMainDAG}).

Recently, Henson, Lal and Pusey~\cite{pusey2014gdag} have investigated those causal structures for which merely confirming that a given distribution on observed variables satisfies all of the conditional independence relations implied by $d$-separation relations does not guarantee that this distribution is compatible with the causal structure. They coined the term \emph{interesting} for causal structures that have this property. They presented a catalogue of all potentially interesting causal structures having six or fewer nodes in~\cite[App.~E]{pusey2014gdag}, of which all but three were shown to be indeed interesting. Evans has also sought to generate such a catalogue~\cite{evans2012graphical}. The Bell scenario, the Instrumental scenario, and the Triangle scenario all appear in the catalogue, together with many others. Furthermore,they provided numerical evidence and an intuitive argument in favour of the hypothesis that the fraction of causal structures that are interesting increases as the total number of nodes increases. This highlights the need for moving beyond a case-by-case consideration of individual causal structures and for developing techniques for deriving constraints beyond conditional independence relations that can be applied to any interesting causal structure. 
Shannon-type entropic inequalities are an example of such constraints~\cite{steudel2010ancestors,fritz2012bell,fritz2013marginal,chaves2014novel,chaves2014informationinference}. They can be derived for a given causal structure with relative ease, via exclusively linear quantifier elimination, since conditional independence relations are linear equations at the level of entropies. They also have the advantage that they apply for any finite cardinality of the observed variables. Recent work has also looked at non-Shannon type inequalities, potentially further strengthening the entropic constraints~\cite{weilenmann2016entropic,pianaar2016interesting}. However, entropic techniques are still wanting, since the resulting inequalities are often rather weak. For example, they are not sensitive enough to witness some known incompatibilities, in particular for distributions that only arise in quantum but not classical models with a given causal structure~\cite{fritz2012bell,weilenmann2016entropic}\footnote{It should be noted that non-standard entropic inequalities can be obtained through a fine-graining of the causal scenario, namely by \emph{conditioning} on the distinct finite possible outcomes of root variables (``settings''), and these types of inequalities \emph{have} proven somewhat sensitive to quantum-classical separations~\cite{braunstein1988entropic,SchumacherInequality,chaves2014novel}. Such inequalities are still limited, however, in that they are only applicable to those causal structures which feature observed root nodes. The potential utility of entropic analysis where fine-graining is generalized to \emph{non}-root observed nodes is currently being explored by E.W. and Rafael Chaves. Jacques Pienaar has also alluded to similar considerations as a possible avenue for further research~\cite{pianaar2016interesting}.}.

In order to improve this state of affairs, we here introduce a new technique for deriving necessary conditions for the compatibility of a distribution of observed variables with a given causal structure, which we term the {\em\tblue{inflation technique}}. This technique is frequently capable of witnessing incompatibility when many other causal inference techniques fail. For example, in \cref{example:noWdist} of \cref{subsec:witnessingincompat} we prove that the tripartite ``W-type'' distribution is incompatible with the Triangle scenario, despite the incompatibility being invisible to other causal inference tools such as conditional independence relations, Shannon-type~\cite{fritz2013marginal,chaves2014novel,chaves2014informationinference} or non-Shanon-type entropic inequalities~\cite{weilenmann2016entropic}, or covariance matrices~\cite{kela2016covariance}.

The inflation technique works roughly as follows. For a given causal structure under consideration, one can construct many new causal structures, termed {\em inflations} of this causal structure. An inflation duplicates one or more of the nodes of the original causal structure, while mirroring the form of the subgraph describing each node's ancestry. Furthermore, the causal parameters that one adds to the inflated causal structure mirror those of the original causal structure. We show that if marginal distributions on certain subsets of the observed variables in the original causal structure are compatible with the original causal structure, then the same marginal distributions on certain copies of those subsets in the inflated causal structure are compatible with the inflated causal structure (\cref{mainlemma}). Similarly, we show that any necessary condition for compatibility of such distributions with the inflated causal structure translates into a necessary condition for compatibility with the original causal structure (\cref{maincorollary}). Thus, applying standard techniques for deriving causal compatibility inequalities to the inflated causal structure typically results in new causal compatibility inequalities for the original causal structure. The reader interested in seeing an example of how our technique works may want to take a sneak peak at \cref{subsec:witnessingincompat}.

Concretely, we consider causal compatibility inequalities for the inflated causal structure that are obtained as follows. One begins by identifying inequalities for the \tblue{marginal problem}, which is the problem of determining when a given family of marginal distributions on some subsets of variables can arise as marginals of a global joint distribution. {One then looks for sets of variables within the inflated causal structure which admit of nontrivial d-separation relations . (We mainly consider sets of variables with disjoint ancestries.) For each such set, one writes down the appropriate factorization of their joint distribution. These factorization conditions are finally substituted} into the marginal problem inequalities to obtain causal compatibility inequalities for the inflated causal structure. Although these constraints are extremely weak, the inflation technique turns them into powerful necessary conditions for compatibility with the original causal structure.

We show how to identify all relevant factorization conditions from the structure of the inflated causal structure, and also how to obtain all marginal problem inequalities by enumerating all facets of the associated \tblue{marginal polytope} (\cref{step:marginalsproblem}). Translating the resulting causal compatibility inequalities on the inflated causal structure back to the original causal structure, we obtain causal compatibility conditions in the form of nonlinear (polynomial) inequalities. As a concrete example of our technique, we present all the causal compatibility inequalities that can be derived in this manner from a particular inflation of the Triangle scenario (\cref{sec:CCineqs}). In general, we also show how to efficiently obtain a partial set of marginal problem inequalities by enumerating transversals of a certain hypergraph (\cref{sec:TSEM}). 

Besides the entropic techniques discussed above, our method is the first systematic tool for causal inference with latent variables that goes beyond observed conditional independence relations while not assuming any bounds on the cardinality of each latent variable. While our method can be used to systematically generate necessary conditions for compatibility with a given causal structure, we do not know whether the set of inequalities thus generated are also sufficient. 

We present our technique primarily as a tool for standard causal inference, but we also briefly discuss applications to {\em quantum} causal models~\cite{leifer2013conditionalstates,pusey2014gdag,BeyondBellII,Chaves2015infoquantum,ried2015quantum,costa2016quantum,allen2016quantum} and causal models within generalized probabilistic theories~\cite{pusey2014gdag} (\cref{sec:classicallity}). In particular, we discuss when our inequalities are necessary conditions for a distribution of observed variables to be compatible with a given causal structure within any generalized probabilistic theory~\cite{hardy2001quantum,barrett2007information} rather than simply within classical probability theory.

\section{Basic Definitions of Causal Models and Compatibility}\label{sec:definitions}

A \tblue{causal model} consists of a pair of objects: a \tblue{causal structure} and a family of \tblue{causal parameters}. We define each in turn. 
First, recall that a directed acyclic graph (DAG) $G$ consists of a finite set of nodes $\nodes{G}$ and a set of directed edges $\edges{G}\subseteq\nodes{G}\times\nodes{G}$, meaning that an edge is an ordered pair of nodes, such that this directed graph is \emph{acylic}, which means that there is no way to start and end at the same node by traversing edges forward.
In the context of a causal model, each node $X\in\nodes{G}$ will be equipped with a random variable that we denote by the same letter $X$. A directed edge $X\to Y$ corresponds to the possibility of a direct causal influence from the variable $X$ to the variable $Y$. In this way, the edges represent causal relations.

Our terminology for the causal relations between the nodes in a DAG is the standard one. The parents of a node $X$ in $G$ are defined as those nodes from which an outgoing edge terminates at $X$, i.e.~$\Pa[G]{X} = \{\:Y\:|\:Y\to X\:\}$. When the graph $G$ is clear from the context, we omit the subscript. Similarly, the children of a node $X$ are defined as those nodes at which edges originating at $X$ terminate, i.e.~$\Ch[G]{X} = \{\:Y\:|\: X\to Y\:\}$. If $\bm{X}$ is a set of nodes, then we put $\Pa[G]{\bm{X}} := \bigcup_{X\in\bm{X}} \Pa[G]{X}$ and $\Ch[G]{\bm{X}} := \bigcup_{X\in\bm{X}} \Ch[G]{X}$. The \tblue{ancestors} of a set of nodes $\bm{X}$, denoted $\An[G]{\bm{X}}$, are defined as those nodes which have a directed path to some node in $\bm{X}$, including the nodes in $\bm{X}$ themselves\footnote{The inclusion of a node itself within the set of its ancestors is contrary to the colloquial use of the term ``ancestors''. 
One uses this definition so that any correlation between two variables can always be attributed to a common ``ancestor''. This includes, for instance, the case where one variable is a parent of the other.
}. 
Equivalently, $\An{\bm{X}} := \bigcup_{n\in\mathbb{N}} \mathsf{Pa}^n(\bm{X})$, where $\mathsf{Pa}^n(\bm{X})$ is defined inductively via $\mathsf{Pa}^0(\bm{X}) := \bm{X}$ and $\mathsf{Pa}^{n+1}(\bm{X}) := \mathsf{Pa}(\mathsf{Pa}^n(\bm{X}))$. 

A \tblue{causal structure} is a DAG that incorporates a distinction between two types of nodes: the set of observed nodes, 
and the set of latent nodes \footnote{\citet[Def. 2.3.2]{pearl2009causality} uses the term \emph{latent structure} when referring to a DAG supplemented by a specification of latent nodes, whereas here that specification is implicit in our term \emph{causal structure}.}. 
Following~\cite{pusey2014gdag}, we will depict the observed nodes by triangles and the latent nodes by circles, as in~\cref{fig:TriMainDAG}\footnote{Note that this convention differs from that of~\cite{leifer2013conditionalstates}, where triangles represent classical variables and circles represent quantum systems.}. Henceforth, we will use $G$ to refer to the causal structure rather than just the DAG, so that $G$ includes a specification of which variables are observed, denoted $\obsnodes{G}$, and which are latent, denoted $\latnodes{G}$.
Frequently, we will also imagine the causal structure to include a specification of the cardinalities of the observed variables. While these are finite in all of our examples, the inflation technique may apply in the case of continuous variables as well. 
Although we will not do so in this work, the inflation technique can also be applied in the presence of other types of constraints, e.g.~when all variables are assumed to be Gaussian. 

The second component of a causal model is a family of \tblue{causal parameters}.
The causal parameters specify, for each node $X$, the conditional probability distribution over the values of the random variable $X$, given the values of the variables in $\Pa{X}$. In the case of root nodes, we have $\Pa{X} = \emptyset$, and the conditional distribution is an unconditioned distribution.
We write $\pfunc{Y|X}$ for the conditional distribution of a variable $Y$ given a variable $X$, while the particular conditional probability of the variable $Y$ taking the value $y$ given that the variable $X$ takes the values $x$ is denoted\footnote{Although our notation suggests that all variables are either discrete or described by densities, we do not make this assumption. All of our equations can be translated straightforwardly into proper measure-theoretic notation.} $\p[Y|X]{y|x}$. Therefore, a family of causal parameters has the form
\begin{align}
 \{ \pfunc{X|\Pa[G]{X}} : X \in \nodes{G} \}.
\end{align}
Finally, a \tblue{causal model} $M$ consists of a causal structure together with a family of causal parameters,
\[
	M = ( G, \{ \pfunc{X|\Pa[G]{X}} : X \in \nodes{G} \}).
\]
A causal model specifies a joint distribution of all variables in the causal structure via
\begin{align}\label{Markov}
P_{\nodes{G}} = \prod_{X\in \nodes{G}} \pfunc{X|\Pa[G]{X}},
\end{align}
where $\prod$ denotes the usual product of functions, so that e.g.~$(P_{Y|X} \times P_Y)(x,y) = P_{Y|X}(y|x) P_X(x)$. A distribution $P_{\nodes{G}}$ arises in this way if and only if it satisfies the Markov conditions associated to $G$~\cite[Sec.~1.2]{pearl2009causality}.

The joint distribution of the observed variables is obtained from the joint distribution of all variables by marginalization over the latent variables,
\begin{align}\label{MarkovObserved}
P_{\obsnodes{G}} = \sum_{\{U :U \in\latnodes{G}\}} P_{\nodes{G}},
\end{align}
where $\sum_U$ denotes marginalization over the (latent) variable $U$, so that $(\sum_U P_{UV})(v):= \sum_u P_{UV}(uv)$.

\begin{definition}\label{def:compatible}
A given distribution $P_{\obsnodes{G}}$ is \tblue{compatible} with a given causal structure $G$ if there is some choice of the causal parameters that yields $P_{\obsnodes{G}}$ via \cref{Markov,MarkovObserved}. A given \emph{family} of distributions on a family of \emph{subsets} of observed variables is compatible with a given causal structure if and only if there exists some $P_{\obsnodes{G}}$ such that both \begin{compactenum}
\item $P_{\obsnodes{G}}$ is compatible with the causal structure, and 
\item $P_{\obsnodes{G}}$ yields the given family as marginals.
\end{compactenum}
\end{definition}

\section{The Inflation Technique for Causal Inference}

\subsection{Inflations of a Causal Model}

We now introduce the notion of \tblue{an inflation of a causal model}. If a causal model specifies a causal structure $G$, then an inflation of this model specifies a new causal structure, $G'$, which we refer to as an inflation of $G$. 
For a given causal structure $G$, there are many causal structures $G'$ constituting an inflation of $G$. We denote the set of such causal structures $\inflations{G}$. 
The particular choice of $G'\in \inflations{G}$ then determines how to map a causal model $M$ on $G$ into a causal model $M'$ on $G'$, 
since the family of causal parameters of $M'$ will be determined by a function $M' = \SmallNamedFunction[G\to G']{Inflation}{M}$ that we define below. We begin by defining when a causal structure $G'$ is an inflation of $G$, building on some preliminary definitions. 

For any subset of nodes $\bm{X}\subseteq\nodes{G}$, we denote the \tblue{induced subgraph} on $\bm{X}$ by $\subgraph[G]{\bm{X}}$. It consists of the nodes $\bm{X}$ and those edges of $G$ which have both endpoints in $\bm{X}$. Of special importance to us is the 
\tblue{ancestral subgraph} $\ansubgraph[G]{\bm{X}}$, which is the subgraph induced by the ancestry of $\bm{X}$, $\ansubgraph[G]{\bm{X}}\coloneqq\subgraph[G]{\An[G]{\bm{X}}}$. 

In an inflated causal structure $G'$, every node is also labelled by a node of $G$. That is, every node of the inflated causal structure $G'$ is a copy of some node of the original causal structure $G$, and the copies of a node $X$ of $G$ in $G'$ are denoted $X_1,\ldots, X_k$. 
The subscript that indexes the copies is termed the \tblue{copy-index}. 
A copy is classified as observed or latent according to the classification of the original. Similarly, any constraints on cardinality or other types of constraints such as Gaussianity are also inherited from the original. 
When two objects (e.g.~nodes, sets of nodes, causal structures, etc\ldots) are the same up to copy-indices, then we use $\sim$ to indicate this, as in $X_i\sim X_j\sim X$. In particular, $\bm{X}\sim\bm{X}'$ for sets of nodes $\bm{X}\subseteq\nodes{G}$ and $\bm{X}'\subseteq\nodes{G'}$ if and only if $\bm{X}'$ contains exactly one copy of every node in $\bm{X}$. Similarly, $\subgraph[G']{\bm{X}'}\sim\subgraph[G]{\bm{X}}$ means that in addition to $\bm{X}\sim\bm{X}'$, an edge is present between two nodes in $\bm{X}'$ if and only if it is present between the two associated nodes in $\bm{X}$.

In order to be an inflation, $G'$ must locally mirror the causal structure of $G$:
\begin{definition}
	The causal structure $G'$ is said to be an \tblue{inflation} of $G$, that is, $G' \in \inflations{G}$, if and only if for every $V_i\in\obsnodes{G'}$, the ancestral subgraph of $V_i$ in $G'$ is equivalent, under removal of the copy-index, to the ancestral subgraph of $V$ in $G$,
\begin{align}\label{eq:definflationDAG}
G' \in\inflations{G} \quad\text{ iff }\quad \forall V_i\in \obsnodes{G'}:\; \ansubgraph[G']{V_i}\sim\ansubgraph[G]{V}.
\end{align}
Equivalently, the condition can be restated wholly in terms of local causal relationships, i.e.
\begin{align}\label{eq:definflationDAGalt}
G' \in\inflations{G} \quad\text{ iff }\quad \forall X_i\in \nodes{G'}:\; \Pa[G']{X_i}\sim\Pa[G]{X}.
\end{align}
\end{definition}

In particular, this means that an inflation is a fibration of graphs~\cite{fibgraphs}, although there are fibrations that are not inflations.

To illustrate the notion of inflation, we consider the causal structure of \cref{fig:TriMainDAG}, which is called the {\em Triangle scenario} (for obvious reasons) and which has been studied recently by a number of authors [\citealp{pusey2014gdag}~(Fig.~E\#8), \citealp{WoodSpekkens}~(Fig.~18b), \citealp{fritz2012bell}~(Fig.~3), \citealp{chaves2014novel}~(Fig.~6a), \citealp{Chaves2015infoquantum}~(Fig.~1a), \citealp{BilocalCorrelations}~(Fig.~8), \citealp{steudel2010ancestors}~(Fig.~1b), \citealp{chaves2014informationinference}~(Fig.~4b)].
Different inflations of the Triangle scenario are depicted in \cref{fig:TriFullDouble,fig:Tri222,fig:simpleinflation,fig:simplestinflation,fig:TriDagSubA2B1C1}, which will be referred to as the {\em Web}, {\em Spiral}, {\em Capped}, and {\em Cut} inflation, respectively.

\begin{figure}[h]
\centering
\begin{minipage}[t]{0.23\linewidth}
\centering
\includegraphics[scale=1]{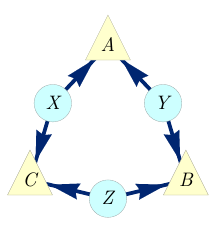}
\caption{The Triangle scenario.}\label{fig:TriMainDAG}
\end{minipage}
\hfill
\begin{minipage}[t]{0.43\linewidth}
\centering
\resizebox{\textwidth}{!}{\includegraphics[scale=1]{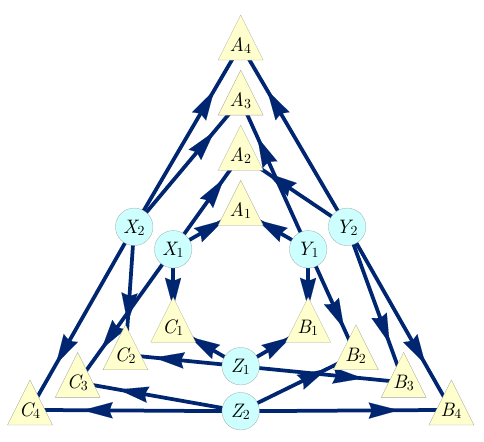}}
\caption{The Web inflation of the 
Triangle scenario where each latent node has been duplicated and each observed node has been quadrupled. The four copies of each observed node correspond to the four possible choices of parentage given the pair of copies of each latent parent of the observed node.}\label{fig:TriFullDouble}
\end{minipage}
\hfill
\begin{minipage}[t]{0.3\linewidth}
\centering
\includegraphics[scale=1]{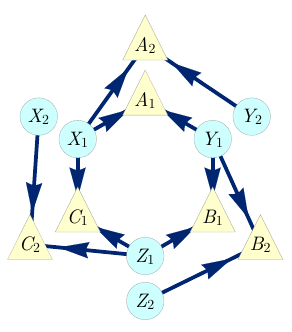}
\caption{The Spiral inflation of the Triangle scenario. Notably, this causal structure is the ancestral subgraph of the set $\{ A_1 A_2 B_1 B_2 C_1 C_2\}$ in the Web inflation (\cref{fig:TriFullDouble}).}
 \label{fig:Tri222}
\end{minipage}
\end{figure}

\begin{figure}[hb]
\centering
\begin{minipage}[t]{0.3\linewidth}
\centering
\includegraphics[scale=1]{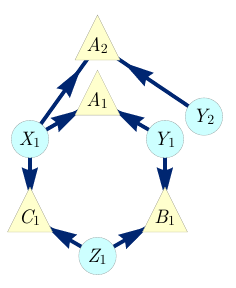}
\caption{The Capped inflation of the Triangle scenario; notably also the ancestral subgraph of the set $\{ A_1 A_2 B_1 C_1\}$ in the Spiral inflation (\cref{fig:Tri222}).}
\label{fig:simpleinflation}
\end{minipage}\hfill
\begin{minipage}[t]{0.275\linewidth}
\centering
\includegraphics[scale=1]{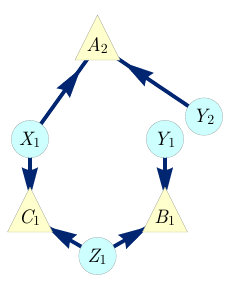}
\caption{The Cut inflation of the Triangle scenario; notably also the ancestral subgraph of the set $\{ A_2 B_1 C_1\}$ in the Capped inflation (\cref{fig:simpleinflation}). Unlike the other examples, this inflation does not contain the Triangle scenario as a subgraph. 
}
\label{fig:simplestinflation}
\end{minipage}
\hfill
\begin{minipage}[t]{0.325\linewidth}
\centering
\includegraphics[scale=1]{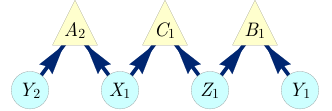}
\caption{A different depiction of the Cut inflation of \cref{fig:simplestinflation}. }\label{fig:TriDagSubA2B1C1}
\end{minipage}
\end{figure}

We now define the function $\mathsf{Inflation}_{G\to G'}$, that is, we specify how causal parameters are defined for a given inflated causal structure in terms of causal parameters on the original causal structure.

\begin{definition}
\label[definition]{def:inflat}
Consider causal models $M$ and $M'$ where $\DAG{M}=G$ and $\DAG{M'}=G'$, where $G'$ is an inflation of $G$. Then $M'$ is said to be the {\em \tblue{$G\to G'$ inflation of $M$}}, that is, $M' = \SmallNamedFunction[G\to G']{Inflation}{M}$, if and only if for every node $X_i$ in $G'$, the manner in which $X_i$ depends causally on its parents within $G'$ is the same as the manner in which $X$ depends causally on its parents within $G$. Noting that $X_i \sim X$ and that $\Pa[G']{X_i} \sim \Pa[G]{X}$ by~\cref{eq:definflationDAGalt}, one can formalize this condition as:
\begin{align}\label{eq:funcdependences}
 \forall X_i \in \nodes{G'}:\; \pfunc{X_i| \Pa[G']{X_i}}=\pfunc{X|\Pa[G]{X}}.
\end{align}
\end{definition}

For a given triple $G$, $G'$, and $M$, this definition specifies a unique inflation model $M'$, resulting in a well-defined function ${\operatorname{\mathsf{Inflation}}_{G\to G'}}$.

To sum up, the inflation of a causal model is a new causal model where (i) each variable in the original causal structure may have counterparts in the inflated causal structure with ancestral subgraphs mirroring those of the originals, and (ii) the manner in which a variable depends causally on its parents in the inflated causal structure is given by the manner in which its counterpart in the original causal structure depends causally on its parents. The operation of modifying a DAG and equipping the modified version with conditional probability distributions that mirror those of the original also appears in the \emph{do calculus} and \emph{twin networks} of~\citet{pearl2009causality}, and moreover bears some resemblance to the \emph{adhesivity} technique used in deriving non-Shannon-type entropic inequalities (see also \cref{sec:NonShannon}).

We are now in a position to describe the key property of the inflation of a causal model, the one that makes it useful for causal inference. With notation as in \cref{def:inflat}, let
$P_{\bm{X}}$ and $P_{\bm{X}'}$ denote marginal distributions on some $\bm{X}\subseteq\nodes{G}$ and $\bm{X}'\subseteq\nodes{G'}$, respectively. Then
\begin{align}\label{eq:coincidingdistrodef}
\quad\text{if }\quad \bm{X}'\sim \bm{X} \;\;\text{and}\;\; \ansubgraph[G']{\bm{X}'}\sim\ansubgraph[G]{\bm{X}}, \quad\text{then}\quad P_{\bm{X}'}=P_{\bm{X}}.
\end{align}
This follows from the fact that the distributions on $\bm{X}'$ and $\bm{X}$ depend only on their ancestral subgraphs and the parameters defined thereon, which by the definition of inflation are the same for $\bm{X}'$ and for $\bm{X}$.
It is useful to have a name for those sets of observed nodes in $G'$ which satisfy the antecedent of~\cref{eq:coincidingdistrodef}, that is, for which one can find a copy-index-equivalent set in the original causal structure $G$ with a copy-index-equivalent ancestral subgraph. We call such subsets of the observed nodes of $G'$ \tblue{injectable sets},
\begin{align}\begin{split}\label{eq:definjectable}
&\bm{V}'\in\SmallNamedFunction{InjectableSets}{G'} \\
&\quad\text{ iff }\quad \exists \bm{V}\subseteq \obsnodes{G} \;\; :\;\; \bm{V}'\sim\bm{V} \;\;\text{and}\;\; \ansubgraph[G']{\bm{V}'}\sim\ansubgraph[G]{\bm{V}}.
\end{split}\end{align}

Similarly, those sets of observed nodes in the original causal structure $G$ which satisfy the antecedent of~\cref{eq:coincidingdistrodef}, that is, for which one can find a corresponding set in the inflated causal structure $G'$ with a copy-index-equivalent ancestral subgraph, we describe as \tblue{images of the injectable sets} under the dropping of copy-indices,
\begin{align}\begin{split}\label{eq:defimageinjectable}
& \bm{V}\in\SmallNamedFunction{ImagesInjectableSets}{G} \\
& \quad\text{ iff }\quad \exists \bm{V}' \subseteq \obsnodes{G'} \;\; :\;\; \bm{V}'\sim\bm{V} \;\;\text{and}\;\; \ansubgraph[G']{\bm{V}'}\sim\ansubgraph[G]{\bm{V}}.
\end{split}\end{align}
Clearly, $\bm{V}\in\SmallNamedFunction{ImagesInjectableSets}{G}$ iff $\exists \bm{V}' \subseteq \SmallNamedFunction{InjectableSets}{G'}$ such that $\bm{V}\sim \bm{V}'$.

For example in the Spiral inflation of the Triangle scenario depicted in~\cref{fig:Tri222}, the set $\brackets{A_1 B_1 C_1}$ is injectable because its ancestral subgraph is equivalent up to copy-indices to the ancestral subgraph of $\brackets{A B C}$ in the original causal structure, and the set $\brackets{A_2 C_1}$ is injectable because its ancestral subgraph is equivalent to that of $\brackets{ A C}$ in the original causal structure.

A set of nodes in the inflated causal structure can only be injectable if it contains at most one copy of any node from the original causal structure. More strongly, it can only be injectable if its ancestral subgraph contains at most one copy of any observed or latent node from the original causal structure. 
Thus, in \cref{fig:Tri222}, $\brackets{A_1 A_2 C_1}$ is not injectable because it contains two copies of $A$, and $\brackets{A_2 B_1 C_1}$ is not injectable because its ancestral subgraph contains two copies of $Y$. 

We can now express \cref{eq:coincidingdistrodef} in the language of injectable sets,
\begin{align}\label{keyinference}
P_{\bm{V}'}=P_{\bm{V}}\quad\text{if } \;\; \bm{V}' \sim \bm{V}\;\; \text{and} \;\; \bm{V}'\in\SmallNamedFunction{InjectableSets}{G'}.
\end{align}

In the example of \cref{fig:Tri222}, injectability of the sets $\brackets{A_1 B_1 C_1}$ and $\brackets{A_2 C_1}$ thus implies that the marginals on each of these are equal to the marginals on their counterparts, $\brackets{A B C}$ and $\brackets{A C}$, in the original causal model, so that $P_{A_1 B_1 C_1} = P_{A B C}$ and $P_{A_2 C_1} = P_{A C}$.

\subsection{Witnessing Incompatibility}
\label{subsec:witnessingincompat}

Finally, we can explain why inflation is relevant for deciding whether a distribution is compatible with a causal structure. {For a distribution $P_{\obsnodes{G}}$ to be compatible with $G$, there must be a causal model $M$ that yields it. Per \cref{def:compatible}, given a $P_{\obsnodes{G}}$ compatible with $G$, the family of marginals of $P_{\obsnodes{G}}$ on the images of the injectable sets of observed variables in $G$, $\{ P_{\bm{V}} : \bm{V} \in \SmallNamedFunction{ImagesInjectableSets}{G}\}$, are \emph{also} said to be compatible with $G$. Looking at the inflation model $M' = \SmallNamedFunction[G\to G']{Inflation}{M}$, \cref{keyinference} implies that the family of distributions on the injectable sets given by $\{ P_{\bm{V}'} : \bm{V}' \in \SmallNamedFunction{InjectableSets}{G'}\}$ --- where $P_{\bm{V}'} = P_{\bm{V}}$ for $\bm{V}'\sim\bm{V}$ --- is compatible with $G'$.

{The same considerations apply for any family of distributions such that each set of variables 
 in the family corresponds to an injectable set (i.e., when the family of distributions is associated with an \emph{incomplete} collection of injectable sets.) Formally, }

\begin{lemma} \label[lemma]{mainlemma}
Let the causal structure $G'$ be an inflation of $G$. 
{Let $\mathbb{S}' \subseteq \SmallNamedFunction{InjectableSets}{G'}$ be a collection of injectable sets, and let $\mathbb{S} \subseteq \SmallNamedFunction{ImagesInjectableSets}{G}$ be the images of this collection under the dropping of copy-indices. If 
a distribution $P_{\obsnodes{G}}$ is compatible with $G$, then the family of distributions $\{ P_{\bm{V}} : \bm{V} \in \mathbb{S} \}$ is compatible with $G$ per \cref{def:compatible}. Furthermore 
 the corresponding family of distributions $\{ P_{\bm{V}'} : \bm{V}' \in \mathbb{S}' \}$, defined via $P_{\bm{V}'}= P_{\bm{V}}$ for $\bm{V}' \sim \bm{V}$, must be compatible with $G'$.}
\end{lemma}

We have thereby related a question about compatibility with the original causal structure to one about compatibility with the inflated causal structure. If one can show that the new compatibility question on $G'$ is answered in the negative, then it follows that the original compatibility question on $G$ is answered in the negative as well. Some simple examples serve to illustrate the idea.

\begin{example}[\tred{Incompatibility of perfect three-way correlation with the Triangle scenario}]
\label{example:noGHZ}


Consider the following causal inference problem. We are given a joint distribution of three binary variables, $P_{A B C}$, where the marginal on each variable is uniform and the three are perfectly correlated,
\begin{align}\label{eq:ghzdistribution1}
P_{A B C} =\frac{[000]+[111]}{2},\quad\text{i.e.,}\quad P_{A B C}(a b c)=\begin{cases}\tfrac{1}{2}&\text{if }\; a = b = c, \\ 0&\text{otherwise},\end{cases}
\end{align}
and we would like to determine whether it is compatible with the Triangle scenario (\cref{fig:TriMainDAG}). The notation $[abc]$ in \cref{eq:ghzdistribution1} is shorthand for the deterministic distribution where $A$, $B$, and $C$ take the values $a, b$, and $c$ respectively; in terms of the Kronecker delta, $[abc]:= \delta_{A,a} \delta_{B,b} \delta_{C,c}$.

Since there are no conditional independence relations among the observed variables in the Triangle scenario, there is no opportunity for ruling out the distribution on the grounds that it fails to satisfy the required conditional independences. 

To solve the causal inference problem, we consider the Cut inflation (\cref{fig:simplestinflation}). The injectable sets include $\brackets{A_2 C_1}$ and $\brackets{B_1 C_1}$. Their images in the original causal structure are $\brackets{AC}$ and $\brackets{BC}$, respectively.

We will show that the distribution of \cref{eq:ghzdistribution1} is not compatible with the Triangle scenario by demonstrating that the contrary assumption of compatibility implies a contradiction. If the distribution of \cref{eq:ghzdistribution1} were compatible with the Triangle scenario, then so too would its pair of marginals on $\brackets{AC}$ and $\brackets{BC}$, which are given by:
\begin{align*}
P_{A C} = P_{B C} = \frac{[00]+[11]}{2}.
\end{align*}
By \cref{mainlemma}, this compatibility assumption would entail that the marginals
\begin{align}\label{ghzmarginals}
P_{A_2 C_1} = P_{B_1 C_1} = \frac{[00]+[11]}{2}
\end{align}
are compatible with the Cut inflation of the Triangle scenario. We now show that the latter compatibility cannot hold, thereby obtaining our contradiction. It suffices to note that (i) the only joint distribution that exhibits perfect correlation between $A_2$ and $C_1$ and between $B_1$ and $C_1$ also exhibits perfect correlation between $A_2$ and $B_1$, and (ii) $A_2$ and $B_1$ have no common ancestor in the Cut inflation and hence must be marginally independent in any distribution that is compatible with it. 

We have therefore certified that the distribution $P_{A B C}$ of~\cref{eq:ghzdistribution1} is not compatible with the Triangle scenario, recovering a result originally proven by \citet{steudel2010ancestors}.
\end{example}

\begin{example}[\tred{Incompatibility of the W-type distribution with the Triangle scenario}]
\label{example:noWdist}


Consider another causal inference problem on the Triangle scenario, namely, that of determining whether the distribution 
\begin{align}\label{eq:wdistribution1}
P_{A B C}=\frac{[100]+[010]+[001]}{3},\quad\text{i.e.,}\quad P_{A B C}(a b c)=\begin{cases}\tfrac{1}{3}&\text{if }\; a + b + c = 1, \\ 0&\text{otherwise}.\end{cases}
\end{align}
is compatible with it. We call this the W-type distribution\footnote{The name stems from the fact that this distribution is reminiscent of 
the famous quantum state appearing in~\cite{3Qubits2Ways}, called the \emph{W state}.}. To settle this compatibility question, we consider the Spiral inflation of the Triangle scenario (\cref{fig:Tri222}).
The injectable sets in this case include $\{A_1 B_1 C_1\}$, $\{A_2 C_1\}$, $\{B_2 A_1\}$, $\{C_2 B_1\}$, $\{A_2\}$, $\{B_2\}$ and $\{C_2\}$. 

Therefore, we turn our attention to determining whether the marginals of the W-type distribution on the images of these injectable sets are compatible with the Triangle scenario. These marginals are:
\begin{align}
P_{A B C}&= \frac{[100]+[010]+[001]}{3}, \label{V4}\\
P_{A C}= P_{B A} = P_{C B} & = \frac{[10]+[01]+[00]}{3}, \label{V1}\\
P_{A}= P_B = P_C & = \frac{2}{3}[0] + \frac{1}{3}[1]. \label{V5}
\end{align}
By \cref{mainlemma}, this compatibility holds only if the associated marginals for the injectable sets, namely, 
\begin{align}
P_{A_1 B_1 C_1}&= \frac{[100]+[010]+[001]}{3}, \label{W4}\\
P_{A_2 C_1} = P_{B_2 A_1} = P_{C_2 B_1} & = \frac{[10]+[01]+[00]}{3}, \label{W1}\\
P_{A_2} = P_{B_2} = P_{C_2} & = \frac{2}{3}[0] + \frac{1}{3}[1], \label{W5}
\end{align}
are compatible with the Spiral inflation (\cref{fig:Tri222}). \cref{W1}
implies that $C_1 \eql 0$ whenever $A_2 \eql 1$. It similarly implies that $A_1 \eql 0$ whenever $B_2 \eql 1$, and that $B_1 \eql 0$ whenever $C_2 \eql 1$, 
\begin{align} 
\begin{split}\label{Ws}
&A_2 \eql 1 \:\implies\: C_1 \eql 0,\\
&B_2 \eql 1 \:\implies\: A_1 \eql 0,\\
&C_2 \eql 1 \:\implies\: B_1 \eql 0.
\end{split}
\end{align}
The Spiral inflation is such that $A_2$, $B_2$ and $C_2$ have no common ancestor and consequently are marginally independent in any distribution compatible with it. Together with the fact that each value of these variables has a nonzero probability of occurrence (by \cref{W5}), this implies that
\begin{align} \label{WW2}
\text{Sometimes} \quad &A_2 \eql 1\,\text{ and }\, B_2 \eql 1\,\text{ and }\, C_2 = 1.
\end{align} 
Finally, \cref{Ws} together with \cref{WW2} entails
\begin{align} \label{WW3}
\text{Sometimes} \quad &A_1 \eql 0\,\text{ and }\, B_1 \eql 0\,\text{ and }\, C_1 \eql 0.
\end{align}
This, however, contradicts~\cref{W4}. Consequently, the family of marginals described in \cref{W4,W1,W5} is \emph{not} compatible with the causal structure of~\cref{fig:Tri222}. By \cref{mainlemma}, this implies that the family of marginals described in \cref{V4,V1,V5}---and therefore the W-type distribution of which they are marginals---is not compatible with the Triangle scenario.

To our knowledge, this is a new result. In fact, the incompatibility of the W-type distribution with the Triangle scenario cannot be derived via any of the 
 existing causal inference techniques. In particular:
\begin{enumerate}
\item Checking conditional independence relations is not relevant here, as there are \emph{no} conditional independence relations between any observed variables in the Triangle scenario. 
\item The relevant Shannon-type entropic inequalities for the Triangle scenario have been classified, and they do not witness the incompatibility~\cite{fritz2013marginal,chaves2014novel,chaves2014informationinference}. 
\item Moreover, \emph{no} entropic inequality can witness the W-type distribution as unrealizable. \citet{weilenmann2016entropic} have constructed an inner approximation to the entropic cone of the Triangle causal structure, and the entropies of the W-distribution form a point in this cone. In other words, a distribution with the same entropic profile as the W-type distribution \emph{can} arise from the Triangle scenario.
\item The newly-developed method of covariance matrix causal inference due to \citet{kela2016covariance}, which gives tighter constraints than entropic inequalities for the Triangle scenario, also cannot detect the incompatibility.
\end{enumerate}
Therefore, in this case at least, the inflation technique appears to be more powerful. 

We have arrived at our incompatibility verdict by combining inflation with reasoning reminiscent of Hardy's version of Bell's theorem~\cite{L.Hardy:PRL:1665,Mansfield2012}. \cref{sec:TSEM} will present a generalization of this kind of argument and its applications to causal inference. 
\end{example}

\begin{example}[\tred{Incompatibility of PR-box correlations with the Bell scenario}]
\label{example:noPR}

Bell's theorem~\cite{bell1964einstein,Brunner2013Bell,bell1966lhvm,CHSHOriginal} concerns the question of whether the distribution obtained in an experiment involving a pair of systems that are measured at space-like separation is compatible with a causal structure of the form of \cref{fig:NewBellDAG1}. Here, the observed variables are $\brackets{A,B,X,Y}$, and $\Lambda$ is a latent variable acting as a common cause of $A$ and $B$. We shall term this causal structure the \emph{Bell scenario}. While the causal inference formulation of Bell's theorem is not the traditional one, several recent articles have introduced and advocated this perspective~[\citealp{WoodSpekkens}~(Fig.~19), \citealp{pusey2014gdag}~(Fig.~E\#2), \citealp{BeyondBellII}~(Fig.~1), \citealp{chaves2014novel}~(Fig.~1), \citealp{wolfe2015nonconvexity}~(Fig.~2b), \citealp{steeg2011relaxation}~(Fig.~2)]. 

\begin{figure}[ht]
\centering
\begin{minipage}[t]{0.45\linewidth}
\centering
\includegraphics[scale=1]{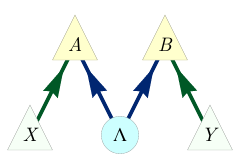}
\caption{The Bell scenario causal structure. The local outcomes, $A$ and $B$, of a pair of measurements are assumed to each be a function of some latent common cause and their independent local experimental settings, $X$ and $Y$.}\label{fig:NewBellDAG1}
\end{minipage}
\hfill
\begin{minipage}[t]{0.45\linewidth}
\centering
\includegraphics[scale=1]{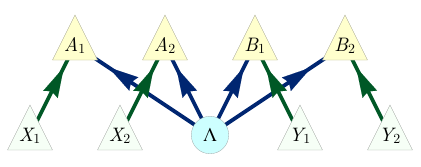}
\caption{An inflation of the Bell scenario causal structure, where both local settings and outcome variables have been duplicated.}\label{fig:BellDagCopy1}
\end{minipage}
\end{figure}

We consider the distribution ${P_{A B X Y} = P_{A B | X Y} P_{X} P_{Y}}$, where $P_{X}$ and $P_{Y}$ are arbitrary full-support distributions on $\{0,1\}$\footnote{In the literature on the Bell scenario, the variables $X$ and $Y$ are termed ``settings''. Generally, we may think of observed root variables as settings, coloring them light green in the figures. They are natural candidates for variables to condition on.},
 and
\begin{align}\begin{split}\label{eq:PRbox}
&P_{A B | X Y}=\begin{cases}
\frac{1}{2}\parenths{[00]+[11]}&\text{if }x\eql 0,y\eql 0\\
\frac{1}{2}\parenths{[00]+[11]}&\text{if }x\eql 1,y\eql 0\\
\frac{1}{2}\parenths{[00]+[11]}&\text{if }x\eql 0,y\eql 1\\
\frac{1}{2}\parenths{[01]+[10]}&\text{if }x\eql 1,y\eql 1\end{cases},\qquad
\text{i.e.,}\;\;P_{A B | X Y}\parens{a b |x y}=\begin{cases}\tfrac{1}{2}&\text{if }\; a \oplus b = x \cdot y , \\ 0&\text{otherwise}.\end{cases}
\end{split}\end{align}
This conditional distribution was discovered by Tsirelson~\cite{Tsirelson1980} and later independently by Popescu and Rohrlich~\cite{PROriginal,PRUnit}. It has become known in the field of quantum foundations as the \emph{PR-box} after the latter authors.\footnote{The PR-box is of interest because it represents a manner in which experimental observations could deviate from the predictions of quantum theory while still being consistent with relativity.}

The Bell scenario implies nontrivial conditional independences\footnote{Recall that variables $X$ and $Y$ are conditionally independent given $Z$ if $P_{XY|Z}(xy|z) = P_{X|Z}(x|z) P_{Y|Z}(y|z)$ for all $z$ with $P_{Z}(z)>0$. Such a conditional independence is denoted by $X\indep Y \:|\: Z$.} among the observed variables, namely, $X \indep Y$, $A \indep Y| X$, and
 $B \indep X|Y$, as well as those that can be generated from these by the semi-graphoid axioms \cite{WoodSpekkens}.
It is straightforward to check that these conditional independence relations are respected by the $P_{ABXY}$ resulting from~\cref{eq:PRbox}. It is well-known that this distribution is nonetheless incompatible with the Bell scenario, since it violates the CHSH inequality.
Here we present a proof of incompatibility in the style of Hardy’s proof of Bell’s theorem~\cite{L.Hardy:PRL:1665} in terms of the inflation technique, using the inflation of the Bell scenario depicted in \cref{fig:BellDagCopy1}.

We begin by noting that $\{A_1 B_1 X_1 Y_1\}$, $\{A_2 B_1 X_2 Y_1\}$, $\{A_1 B_2 X_1 Y_2\}$, $\{A_2 B_2 X_2 Y_2\}$, $\{X_1\}$, $\{X_2\}$, $\{Y_1\}$, and $\{Y_2\}$ are all injectable sets. 
By \cref{mainlemma}, it follows that any causal model that recovers $P_{ABXY}$ inflates to a model that results in marginals
\begin{align}
P_{A_1 B_1 X_1 Y_1}=P_{A_2 B_1 X_2 Y_1}=P_{A_1 B_2 X_1 Y_2}=P_{A_2 B_2 X_2 Y_2}&=P_{A B X Y},\label{PR1}\\
P_{X_1}=P_{X_2}=P_X, \qquad P_{Y_1}=P_{Y_2}&=P_Y.\label{PR5}
\end{align}
Using the definition of conditional probability, we infer that
\begin{align}
P_{A_1 B_1 |X_1 Y_1}=P_{A_2 B_1 |X_2 Y_1}=P_{A_1 B_2 |X_1 Y_2}=P_{A_2 B_2 |X_2 Y_2}=P_{A B |X Y}\label{PRb}.
\end{align}
Because $\{X_1\}$, $\{X_2\}$, $\{Y_1\}$, and $\{Y_2\}$ have no common ancestor in the inflated causal structure, these variables must be marginally independent in any distribution compatible with it, so that $P_{X_1 X_2 Y_1 Y_2} = P_{X_1} P_{X_2} P_{Y_1} P_{Y_2}$. Given the assumption that the distributions $P_{X}$ and $P_{Y}$ have full support, it follows from~\cref{PR5} that
\begin{align}\label{PRs}
\text{Sometimes} \quad &X_1 = 0\,\text{ and }\, X_2 =1\,\text{ and }\, Y_1 = 0\,\text{ and }\, Y_2 = 1.
\end{align} 
On the other hand, from~\cref{PRb} together with the definition of PR-box,~\cref{eq:PRbox}, we conclude that 
\begin{align} 
\begin{split}
\label{PRsi}
&X_1 \eql 0,\: Y_1 \eql 0 \:\implies\: A_1 = B_1,\\[-1ex]
&X_1 \eql 0,\: Y_2 \eql 1 \:\implies\: A_1 = B_2,\\[-1ex]
&X_2 \eql 1,\: Y_1 \eql 0 \:\implies\: A_2 = B_1,\\[-1ex]
&X_2 \eql 1,\: Y_2 \eql 1 \:\implies\: A_2\ne B_2.
\end{split}
\end{align}
Combining this with~\cref{PRs}, we obtain
\begin{align}
\text{Sometimes} \quad &A_1 = B_1\,\text{ and }\, A_1 = B_2\,\text{ and }\, A_2 = B_1\,\text{ and }\, A_2\neq B_2.
\end{align} 
No values of $A_1$, $A_2$, $B_1$, and $B_2$ can jointly satisfy these conditions. So we have reached a contradiction, showing that our original assumption of compatibility of $P_{ABXY}$ with the Bell scenario must have been false.

The structure of this argument parallels that of standard proofs of the incompatibility of the PR-box with the Bell scenario. Standard proofs focus on a set of variables $\{A_0 A_1 B_0 B_1\}$ where $A_x$ is the value of $A$ when $X=x$ and $B_y$ is the value of $B$ when $Y=y$. Note that the distribution $\displaystyle\sum_{\Lambda} P_{A_0|\Lambda}P_{A_1|\Lambda}P_{B_0|\Lambda}P_{B_1|\Lambda}P_{\Lambda}$ 
 is a joint distribution of these four variables for which the marginals on pairs $\{ A_0 B_0\}$, $\{ A_0 B_1\}$, $\{ A_1 B_0\}$ and $\{ A_1 B_1\}$ are those that can arise in the Bell scenario. 
 The existence of such a joint distribution rules out the possibility of having $A_1 = B_1$, $A_1 = B_2$, $A_2 = B_1$ but $ A_2\neq B_2$, and therefore shows that the PR-box distribution is incompatible with the Bell scenario~\cite{LSW,roberts_thesis}. In light of our use of \cref{PRs}, the reasoning based on the inflation of Fig.~\ref{fig:BellDagCopy1} is really the same argument in disguise.

\cref{sec:Bellscenarios} shows that the inflation of the Bell scenario depicted in~\cref{fig:BellDagCopy1} is sufficient to witness the incompatibility of any distribution that is incompatible with the Bell scenario.
\end{example}

\subsection{Deriving Causal Compatibility Inequalities}
\label{Sec:DerivingInequalities}

The inflation technique can be used not only to witness the incompatibility of a given distribution with a given causal structure, but also to derive necessary conditions that a distribution must satisfy to be compatible with the given causal structure. These conditions can always be expressed as inequalities, and we will refer to them as {\em causal compatibility inequalities}\footnote{Note that we can include equality constraints for causal compatibility within the framework of causal compatibility inequalities alone; it suffices to note that an equality constraint can always be expressed as a pair of inequalities, i.e. satisfying $x=y$ is equivalent to satisfying both $x\leq y$ and $x\geq y$. The requirement that a distribution must be Markov (or Nested Markov) relative to a DAG is usually formulated as a set of equality constraints.}. Formally, we have:

\begin{definition}
Let $G$ be a causal structure and let $\mathbb{S}$ be a family of subsets of the observed variables of $G$, $\mathbb{S} \subseteq 2^{\obsnodes{G}}$. Let $I_{\mathbb{S}}$ denote an inequality that operates on the corresponding family of distributions, $\{ P_{\bm{V}}: \bm{V} \in \mathbb{S}\}$. Then $I_{\mathbb{S}}$ is a \tblue{\em causal compatibility inequality for the causal structure $G$} whenever it is satisfied by every family of distributions $\{ P_{\bm{V}}: \bm{V} \in \mathbb{S}\}$ that is compatible with $G$.
\end{definition}
While violation of a causal compatibility inequality witnesses the incompatibility with the causal structure, satisfaction of the inequality does not guarantee compatibility. This is the sense in which it merely provides a {\em necessary} condition for compatibility. 

The inflation technique is useful for deriving causal compatibility inequalities because of the following consequence of \cref{mainlemma}:

\begin{corollary} \label[corollary]{maincorollary}
Suppose that $G'$ is an inflation of $G$. Let $\mathbb{S}' \subseteq \SmallNamedFunction{InjectableSets}{G'}$ be a family of injectable sets and $\mathbb{S} \subseteq \SmallNamedFunction{ImagesInjectableSets}{G}$ the images of members of $\mathbb{S}'$ under the dropping of copy-indices.
Let $I_{\mathbb{S}'}$ be a causal compatibility inequality for $G'$ operating on families $\{ P_{\bm{V}'} : \bm{V}' \in \mathbb{S}'\}$. Define an inequality $I_{\mathbb{S}}$ as follows: in the functional form of $I_{\mathbb{S}'}$, replace every occurrence of a term $P_{\bm{V}'}$ by $P_{\bm{V}}$ for the unique $\bm{V}\in S$ with $\bm{V} \sim \bm{V}'$. Then $I_{\mathbb{S}}$ is a causal compatibility inequality for $G$ operating on families $\{ P_{\bm{V}} : \bm{V}\in \mathbb{S}\}$.
\end{corollary}

\begin{proof}
Suppose that the family $\{ P_{\bm{V}} : \bm{V} \in \mathbb{S}\}$ is compatible with $G$. By \cref{mainlemma}, it follows that the family $ \{ P_{\bm{V}'} : \bm{V}' \in \mathbb{S}'\}$ where $P_{\bm{V}'}:= P_{\bm{V}}$ for $\bm{V}' \sim \bm{V}$ is compatible with $G'$. Since $I_{\mathbb{S}'}$ is a causal compatibility inequality for $G'$, it follows that $\{ P_{\bm{V}'} : \bm{V}' \in \mathbb{S}'\}$ satisfies $I_{\mathbb{S}'}$. But by the definition of $I_{\mathbb{S}}$, its evaluation on $\{ P_{\bm{V}} : \bm{V} \in \mathbb{S}\}$ is equal to $I_{\mathbb{S}'}$ evaluated on $\{ P_{\bm{V}'} : \bm{V}' \in \mathbb{S}'\}$. It therefore follows that $\{ P_{\bm{V}} : \bm{V} \in \mathbb{S}\}$ satisfies $I_{\mathbb{S}}$. Since $\{ P_{\bm{V}} : \bm{V}\in \mathbb{S}\}$ was an arbitrary family compatible with $G$, we conclude that $I_{\mathbb{S}}$ is a causal compatibility inequality for $G$.
\end{proof}

We now present some simple examples of causal compatibility inequalities for the Triangle scenario that one can derive from the inflation technique via \cref{maincorollary}. Some terminology and notation will facilitate their description. We refer to a pair of nodes which do not share any common ancestor as being \tblue{ancestrally independent}. This is equivalent to being $d$-separated by the empty set~\cite{pearl2009causality,spirtes2011causation,studeny2005probabilistic,koller2009probabilistic}. Given that the conventional notation for $X$ and $Y$ being $d$-separated by $Z$ in a DAG is $X\aindep Y|Z$, we denote $X$ and $Y$ being ancestrally independent within $G$ as $X\aindep Y$. Generalizing to sets, $\bm{X}\aindep \bm{Y}$ indicates that no node in $\bm{X}$ shares a common ancestor with any node in $\bm{Y}$ within the causal structure $G$, 
\begin{align}
\bm{X}\aindep \bm{Y} \quad \text{iff} \quad \An[G]{\bm{X}}\cap\An[G]{\bm{Y}}=\emptyset.
\end{align}
Ancestral independence is closed under union; that is, $\bm{X}\aindep \bm{Y}$ and $\bm{X}\aindep \bm{Z}$ implies $\bm{X}\aindep (\bm{Y{\cup}Z)}$. Consequently, pairwise ancestral independence implies joint factorizability; i.e. $\forall_{i\neq j} \bm{X}_i\aindep \bm{X}_j$ implies that $P_{\cup_i \bm{X}_i}=\prod_i P_{\bm{X}_i}$.


\begin{example}[\tred{A causal compatibility inequality in terms of \tpurp{\emph{correlators}}}]
\label{example:polytriangle}

As in \cref{example:noGHZ} of the previous subsection, consider the Cut inflation of the Triangle scenario (\cref{fig:simpleinflation}), where all observed variables are binary. For technical convenience, we assume that they take values in the set $\{-1,+1\}$, rather than taking values in $\{0,1\}$ as was presumed in the last subsection.

The injectable sets that we make use of are $\brackets{A_2 C_1}$, $\brackets{B_1 C_1}$, $\{ A_2\}$, and $\brackets{B_1}$. From \cref{maincorollary}, any causal compatibility inequality for the inflated causal structure that operates on the marginal distributions of $\brackets{A_2 C_1}$, $\brackets{B_1 C_1}$, $\{ A_2\}$, and $\brackets{B_1}$ will yield a causal compatibility inequality for the original causal structure that operates on the marginal distributions on $\brackets{A C}$, $\brackets{B C}$, $\brackets{A}$, and $\brackets{B}$. We begin by noting that for {\em any} distribution on three binary variables $\{A_2 B_1 C_1\}$, that is, {\em regardless} of the causal structure in which they are embedded, the marginals on $\brackets{A_2 C_1}$, $\brackets{B_1 C_1}$ and $\brackets{A_2 B_1}$ satisfy the following inequality for expectation values~\cite{pitowsky_boole_1994,Pitowsky1989,kellerer_marginal_1964,leggett_garg_1985,araujo_cycle_2013},
\begin{equation}
	\label{eq:polymonogamyraw}
	\expec{ A_2 C_1} + \expec{ B_1 C_1 } \leq 1 + \expec{ A_2 B_1 }.
\end{equation}
This is an example of a constraint on pairwise correlators that arises from the presumption that they are consistent with a joint distribution. (The problem of deriving such constraints is the {\em marginal constraint problem}, discussed in detail in \cref{sec:ineqs}.)

But in the Cut inflation of the Triangle scenario~(\cref{fig:simpleinflation}), $A_2$ and $B_1$ have no common ancestor and consequently any distribution compatible with this inflated causal structure must make $A_2$ and $B_1$ marginally independent. In terms of correlators, this can be expressed as 
\begin{align}\label{corrfact}
A_2 \aindep B_1 \implies A_2 \indep B_1 \implies \expec{ A_2 B_1 } = \expec{ A_2} \expec{ B_1 }.
\end{align}
Substituting this into~\cref{eq:polymonogamyraw}, we have
\begin{equation}
	\expec{ A_2 C_1} + \expec{ B_1 C_1 } \leq 1 + \expec{ A_2 } \expec{ B_1}.
\end{equation}
This is an example of a simple but nontrivial causal compatibility inequality for the causal structure of~\cref{fig:simpleinflation}. Finally, by \cref{maincorollary}, we infer that 
\begin{equation}
	\label{eq:polymonogamy}
	\expec{ A C} + \expec{ B C} \leq 1 + \expec{ A} \expec{ B}
\end{equation}
is a causal compatibility inequality for the Triangle scenario. This inequality expresses the fact that as long as $A$ and $B$ are not completely biased, there is a tradeoff between the strength of $AC$ correlations and the strength of $BC$ correlations. 

Given the symmetry of the Triangle scenario under permutations and sign flips of $A$, $B$ and $C$, it is clear that the image of inequality~\eqref{eq:polymonogamy} under any such symmetry is also a valid causal compatibility inequality. Together, these inequalities constitute a type of monogamy\footnote{We are here using the term ``monogamy'' in the same sort of manner in which it is used in the context of entanglement theory~\cite{horo4}.} of correlations in the Triangle scenario with binary variables: if any two observed variables with unbiased marginals are perfectly correlated, then they are both independent of the third.

Moreover, since inequality \eqref{eq:polymonogamyraw} is valid even for continuous variables with values in the interval $[-1,+1]$, it follows that the polynomial inequality \eqref{eq:polymonogamy} is valid in this case as well.

Note that inequality \eqref{eq:polymonogamyraw} serves as a robust witness certifying the incompatibility of 3-way perfect correlation (described in \cref{eq:ghzdistribution1}) with the Triangle scenario. Inequality \eqref{eq:polymonogamyraw} is robust in the sense that it demonstrates the incompatibility of distributions \emph{close} to 3-way perfect correlation.

One might be curious as to \emph{how} close to perfect correlation one can get while still being compatible with the Triangle scenario. To partially answer this question, we used \cref{eq:polymonogamyraw} to \emph{rule out} many distributions close to perfect correlation and we also pursued explicit model-construction to \emph{rule in} various distributions sufficiently \emph{far} from perfect correlation. Explicitly, we found that distributions of the form
\begin{align}\label{eq:ghznoisy}
P_{A B C} =\alpha\frac{[000]+[111]}{2}+(1-\alpha)\frac{[\text{else}]}{6},\quad\text{i.e.,}\quad P_{A B C}(a b c)=\begin{cases}\tfrac{\alpha}{2}&\text{if }\; a = b = c, \\ \tfrac{1-\alpha}{6}&\text{otherwise},\end{cases}
\end{align}
 where $[\text{else}]$ denotes any point distribution $[abc]$ other than $[000]$ or $[111]$, are incompatible for the range $\tfrac{5}{8}=0.625< \alpha \leq 1$ as a consequence of \cref{eq:polymonogamyraw}. On the other hand, we found a family of explicit models allowing us to certify the \emph{compatibility} of distributions for $0\leq \alpha \leq \tfrac{1}{2}$. 

{The presence of this gap between our inner and outer constructions could reflect either the inadequacy of our limited model constructions or the inadequacy of relatively small inflations of the Triangle causal structure to generate suitably sensitive inequalities. We defer closing the gap to future work\footnote{Using the Web inflation of the Triangle as depicted in \cref{fig:TriFullDouble} we were able to slightly improve the range of certifiably incompatible $\alpha$, namely we find that $P_{A B C}$ is incompatible with the Triangle scenario for all $\tfrac{3\sqrt{3}}{2}-2 \approx 0.598< \alpha$. The relevant causal compatibility inequality justifying the improved bound is $6 \expec{\_ \_} + \expec{\_ \_}^2-4\expec{\_}^2 \leq 3$, where $\expec{\_ \_}\coloneqq \frac{\expec{A B}+\expec{B C}+\expec{A C}}{3}$ and $\expec{\_ }\coloneqq \frac{\expec{A}+\expec{B}+\expec{C}}{3}$.}.}

\end{example}

\begin{example}[\tred{A causal compatibility inequality in terms of \tpurp{\emph{entropic quantities}}}]
\label{ex:entropic}

One way to derive constraints that are independent of the cardinality of the observed variables is to express these in terms of the mutual information between observed variables rather than in terms of correlators. The inflation technique can also be applied to achieve this.
To see how this works in the case of the Triangle scenario, consider again the Cut inflation~(\cref{fig:simpleinflation}). 

One can follow the same logic as in the preceding example, but starting from a different constraint on marginals. For any distribution on three variables $\{A_2 B_1 C_1\}$ of arbitrary cardinality (again, regardless of the causal structure in which they are embedded), the marginals on $\brackets{A_2 C_1}$, $\brackets{B_1 C_1}$ and $\brackets{A_2 B_1}$ satisfy the inequality~\cite[Eq.~(29)]{fritz2013marginal}
\label{example:entropic}\begin{align}\label{eq:MIraw}
	I(A_2 : C_1) + I(C_1 : B_1) \leq H(C_1) + I(A_2 : B_1),	
\end{align}
where $H(X)$ denotes the Shannon entropy of the distribution of $X$, and $I(X: Y)$ denotes the mutual information between $X$ and $Y$ with respect to the marginal joint distribution on the pair of variables $X$ and $Y$. The fact that $A_2$ and $B_1$ have no common ancestor in the inflated causal structure implies that in any distribution that is compatible with it, $A_2$ and $B_1$ are marginally independent. This is expressed entropically as the vanishing of their mutual information, 
\begin{align}\label{entropicfact}
A_2 \aindep B_1 \implies A_2 \indep B_1 \implies I(A_2 : B_1) =0.
\end{align}
Substituting the latter equality into~\cref{eq:MIraw}, we have
\begin{align}
	I(A_2 : C_1) + I(C_1 : B_1) \leq H(C_1).
\end{align}
This is another example of a nontrivial causal compatibility inequality for the causal structure of~\cref{fig:simpleinflation}. By \cref{maincorollary}, it follows that 
\begin{align}\label{eq:monogomyofcorrelations}
	I(A : C) + I(C : B) \leq H(C)
\end{align}
is also a causal compatibility inequality for the Triangle scenario. This inequality was originally derived in~\cite{fritz2012bell}. Our rederivation in terms of inflation coincides with the proof found by~\citet{pusey2014gdag}.
\end{example}

{Standard algorithms already exist for deriving entropic casual compatibility inequalities given a causal structure~\cite{fritz2013marginal,chaves2014novel,chaves2014informationinference}. We do not expect the methodology of causal inflation to offer any computation advantage in the task of deriving entropic inequalities. The advantage of the inflation approach is that it provides a narrative for explaining an entropic inequality without reference to unobserved variables. As elaborated in \cref{sec:classicallity}, this consequently has applications to quantum information theory. A further advantage is the potential of the inflation approach to give rise to non-Shannon type inequalities, starting from Shannon type inequalities; see \cref{sec:NonShannon} for further discussion.}

\begin{example}[\tred{A causal compatibility inequality in terms of \tpurp{\emph{joint distributions}}}]
\label{example:probineq}

Consider the Spiral inflation of the Triangle scenario (\cref{fig:Tri222}) with the injectable sets $\{A_1 B_1 C_1\}$, $\{A_1 B_2\}$, $\{B_1 C_2\}$, $\{ A_1, C_2\}$, $\{A_2\}$, $\{B_2\}$, and $\{C_2\}$. We derive a causal compatibility inequality under the assumption that the observed variables are binary, adopting the convention that they take values in $\{0,1\}$.

We begin by noting that the following is a constraint that holds for any joint distribution of $\{A_1 B_1 C_1 A_2 B_2 C_2\}$, regardless of the causal structure, 
\begin{align}\label{eq:FritzF3raw}
	P_{A_2 B_2 C_2}(111) \leq P_{A_1 B_2 C_2}(111) + P_{B_1 C_2 A_2}(111) + P_{A_2 C_1 B_2}(111) + P_{A_1 B_1 C_1}(000).
\end{align}
To prove this claim, it suffices to check that the inequality holds for each of the $2^6$ deterministic assignments of outcomes to $\{A_1 B_1 C_1 A_2 B_2 C_2\}$, from which the general case follows by convex linearity. A more intuitive proof will be provided in~\cref{sec:TSEM}.

Next, we note that certain sets of variables have no common ancestors with other sets of variables in the inflated causal structure, which implies the marginal independence of these sets. Such independences are expressed in the language of joint distributions as factorizations,
\begin{align}\begin{split}\label{eq:tri222fac}
A_1 B_2 \aindep C_2 \:\implies\:	P_{A_1 B_2 C_2} &= P_{A_1 B_2} P_{C_2}, \\
B_1 C_2 \aindep A_2 \:\implies\:	P_{B_1 C_2 A_2} &= P_{B_1 C_2} P_{A_2}, \\
A_2 C_1 \aindep B_2 \:\implies\:	P_{A_2 C_1 B_2} &= P_{A_2 C_1} P_{B_2}, \\
A_2 \aindep B_2 \aindep C_2 \:\implies\:	P_{A_2 B_2 C_2} &= P_{A_2} P_{B_2} P_{C_2} .
\end{split}\end{align}
Substituting these factorizations into~\cref{eq:FritzF3raw}, we obtain the polynomial inequality
\begin{equation}
	P_{A_2}(1) P_{B_2}(1) P_{C_2}(1) \leq P_{A_1 B_2}(11) P_{C_2}(1) + P_{B_1 C_2}(11) P_{A_2}(1) + P_{A_2 C_1}(11) P_{B_2}(1) + P_{A_1 B_1 C_1}(000).
\end{equation}
This, therefore, is a causal compatibility inequality for the inflated causal structure. Finally, by \cref{maincorollary}, we infer that 
\begin{equation}\label{eq:FritzF3}
	P_{A}(1) P_{B}(1) P_{C}(1) \leq P_{AB}(11) P_C(1) + P_{BC}(11) P_A(1) + P_{AC}(11) P_B(1) + P_{ABC}(000)
\end{equation}
is a causal compatibility inequality for the Triangle scenario. 

What is distinctive about this inequality is that---through the presence of the term $P_{ABC}(000)$---it takes into account genuine three-way correlations, while the inequalities we derived earlier only depend on the two-variable marginals. This inequality is strong enough to demonstrate the incompatibility of the W-type distribution of \cref{eq:wdistribution1} with the Triangle scenario: for this distribution, the right-hand side of the inequality vanishes while the left-hand side does not.
\end{example}

Of the known techniques for witnessing the incompatibility of a distribution with a causal structure or deriving necessary conditions for compatibility, the most straightforward one is to consider the constraints implied by ancestral independences among the observed variables of the causal structure. 
The constraints derived in the last two sections have all made use of this basic technique, but at the level of the inflated causal structure rather than the original causal structure. The constraints that one thereby infers for the original causal structure reflect facts about it that cannot be expressed in terms of ancestral independences among its observed variables. The inflation technique exposes these facts in the ancestral independences among observed variables of the inflated causal structure.

In the rest of this article, we shall continue to rely only on the ancestral independences among observed variables within the inflated causal structure to derive examples of compatibility constraints on the original causal structure. Nonetheless, it seems plausible that the inflation technique can also amplify the power of {\em other} techniques that do not merely consider ancestral independences among the observed variables. We consider some prospects in \cref{sec:otherprospects}.

\section{Systematically Witnessing Incompatibility and Deriving Inequalities}
\label{sec:ineqs}

This section considers the problem of how to generalize the above examples of causal inference via the inflation technique to a systematic procedure. We start by introducing the crucial concept of an \emph{expressible set},
 which figures implicitly in our earlier examples. By reformulating \cref{example:noGHZ}, we sketch our general method and explain why solving a \emph{marginal problem} is an essential subroutine of our method. Subsequently, \cref{step:findpreinjectable} explains how to systematically identify, for a given inflated causal structure, all of the 
 sets that are expressible by virtue of ancestral independences.
 \cref{step:marginalsproblem} describes how to solve any sort of marginal problem. This may involve determining all the facets of the \emph{marginal polytope}, which is computationally costly (\cref{sec:projalgorithms}). It is therefore useful to also consider relaxations of the marginal problem that are more tractable by deriving valid linear inequalities which may or may not bound the marginal polytope tightly. We describe one such approach based on possibilistic Hardy-type paradoxes and the hypergraph transversal problem in \cref{sec:TSEM}.

As far as causal compatibility inequalities are concerned, we limit ourselves to those expressed in terms of probabilities\footnote{Or, for binary variables, equivalently in terms of correlators, as in the first example of \cref{Sec:DerivingInequalities}.}, as these are generally the most powerful. However, essentially the same techniques can be used to derive inequalities expressed in terms of entropies~\cite{fritz2013marginal}, as demonstrated in \cref{ex:entropic}. 

In the examples from the previous section, the initial inequality---a constraint upon marginals that is independent of the causal structure---involves sets of observed variables that are \emph{not} all injectable sets. However, the Markov conditions on the inflated causal structures nevertheless allowed us to \emph{express} the distribution on these sets in terms of the known distributions on the injectable sets. For instance, in \cref{example:polytriangle}, the set $\{ A_2 B_1\}$ is not injectable, but it can be partitioned into the singleton sets $\{ A_2 \}$ and $\{ B_1\}$ which are ancestrally independent, so that one has $P_{A_2 B_1} = P_{A_2} P_{B_1} = P_A P_B$ in every inflated causal model. This motivates us to define the notion of an \tblue{expressible set} of variables in an inflated causal structure as one for which the joint distribution can be expressed as a function of distributions over injectable sets by making repeated use of the conditional independences implied by $d$-separation relations as well as marginalization. More formally,

\begin{definition}\label{def:expressible}
Consider an inflation $G'$ of a causal structure $G$. Sufficient conditions for a set of variables $\bm{V}' \subset \obsnodes{G'}$ to be \tblue{expressible} include $\bm{V}'\in\SmallNamedFunction{InjectableSets}{G'}$, or if $\bm{V}'$ can be obtained from a collection of injectable sets by recursively applying the following rules:
\begin{enumerate}
\item For $\bm{X}',\bm{Y}',\bm{Z}'\subseteq\obsnodes{G'}$, if $\bm{X}'\aindep \bm{Y}'\:|\:\bm{Z}'$ and $\bm{X}' \cup \bm{Z}'$ and $\bm{Y}' \cup \bm{Z}'$ are expressible, then $\bm{X}'\cup \bm{Y}'\cup \bm{Z}'$ is also expressible. {This follows by constructing $\p[\bm{X}'\bm{Y}'\bm{Z}']{\bm{x}\bm{y}\bm{z}}=\begin{cases}\frac{\p[\bm{X}'\bm{Z}']{\bm{x}\bm{z}}\p[\bm{Y}'\bm{Z}']{\bm{y}\bm{z}}}{\p[\bm{Z}']{\bm{z}}} & \quad\text{if }\p[\bm{Z}']{\bm{z}}>0, \\
 0 & \quad\text{if }\p[\bm{Z}']{\bm{z}}= 0.
 \end{cases}$}
\item If $\bm{V}'\subseteq\obsnodes{G'}$ is expressible, then so is every subset of $\bm{V}'$. {This follows by marginalization.}
\end{enumerate}
\end{definition}
An expressible set is \emph{maximal} if it is not a proper subset of another expressible set.

Expressible sets are important since in an inflated model, the distribution of the variables making up an expressible set can be computed explicitly from the known distributions on the injectable sets, by repeatedly using the conditional independences implied by $d$-separation and taking marginals. \cref{example:Pienaar} provides a good example.

With the exception of \cref{sec:interestingproof}, in the remainder of this article we will limit ourselves to working with expressible sets of a particularly simple kind and leave the investigation of more general expressible sets to future work.

\begin{definition}
A set of nodes $\bm{V}'\subseteq\obsnodes{G'}$ is \tblue{ai-expressible} if it can be written as a union of injectable sets that are ancestrally independent,
\begin{align}\label{eq:defpreinj}
\begin{split}
\bm{V}'\in & \SmallNamedFunction{AI-ExpressibleSets}{G'} \\
	& \quad\text{ iff }\quad \exists \{ \bm{X}'_i \in \SmallNamedFunction{InjectableSets}{G'} \} \quad \text{s.t.}\quad \bm{V}'=\bigcup_i \bm{X}'_i \quad\text{and} \quad \forall_{i\ne j}: \bm{X}'_i \aindep \bm{X}'_j \text{ in }G'.
\end{split}
\end{align}
An ai-expressible set is \emph{maximal} if it is not a proper subset of another ai-expressible set.
\end{definition}

Because ancestral independence in $G'$ implies statistical independence for any compatible distribution, it follows that if $\bm{V}'$ is an ai-expressible set with ancestrally independent and injectable components $\bm{V}'_1,\ldots,\bm{V}'_n$, then we have the factorization
\begin{align}\label{eq:preinjfactor}
P_{\bm{V}'} = P_{\bm{V}'_1} \cdots P_{\bm{V}'_n}
\end{align}
for any distribution compatible with $G'$. The situation, therefore, is this: for any constraint that one can derive for the marginals on the ai-expressible sets based on the existence of a joint distribution---and hence without reference to the causal structure---one can infer a constraint that {\em does} refer to the causal structure by substituting within the derived constraint a factorization of the form of~\cref{eq:preinjfactor}. This results in a causal compatibility inequality on $G'$ of a very weak form that only takes into account the independences between observed variables.

As a build-up to our exposition of a systematic application of the inflation technique, we now revisit \cref{example:noGHZ}. As before, to demonstrate the incompatibility of the distribution of \cref{eq:ghzdistribution1} with the Triangle scenario, we assume compatibility and derive a contradiction. Given the distribution of \cref{eq:ghzdistribution1}, \cref{mainlemma} implies that the marginal distributions on the injectable sets of the Cut inflation of the Triangle scenario are 
\begin{equation}
P_{A_2 C_1} = P_{B_1 C_1} = \frac{1}{2} [00] +\frac{1}{2} [11], 
\label{marginals1}
\end{equation}
and
\begin{equation}
\qquad P_{A_2} = P_{B_1}=\frac{1}{2} [0] +\frac{1}{2} [1].
\label{marginals1prime}
\end{equation}
From the fact that $A_2$ and $B_1$ are ancestrally independent in the Cut inflation, we also infer that the distribution on the ai-expressible set $\{A_2 B_1\}$ must be
\begin{equation}
P_{A_2 B_1} = P_{A_2}P_{B_1} = \left(\frac{1}{2} [0] +\frac{1}{2} [1]\right)\times\left(\frac{1}{2} [0] +\frac{1}{2} [1]\right)=\frac{1}{4} [00]+\frac{1}{4} [01]+\frac{1}{4} [10]+\frac{1}{4} [11].
\label{marginals2}
\end{equation}
But there is no three-variable distribution $P_{A_2 B_1 C_1}$ that would have as its two-variable marginals the distributions of \cref{marginals1,marginals2}. For as we noted in our prior discussion of this example, the perfect correlation between $A_2$ and $C_1$ exhibited by $P_{A_2 C_1}$ and the perfect correlation between $B_1$ and $C_1$ exhibited by $P_{B_1 C_1}$ would entail perfect correlation between $A_2$ and $B_1$ as well, which is at odds with \eqref{marginals2}. 
We have therefore derived a contradiction and consequently can infer the incompatibility of the distribution of \cref{eq:ghzdistribution1} with the Triangle scenario.

Generalizing to an arbitrary causal structure, therefore, the procedure is as follows:
\begin{enumerate}
\item Based on the inflation under consideration, identify the ai-expressible sets and how they each partition into ancestrally independent injectable sets.
\item From the given distribution on the original causal structure, infer the family of distributions on the ai-expressible sets of the inflated causal structure as follows: the distribution on any injectable set is equal to the corresponding distribution on its image in the original causal structure; the distribution on any ai-expressible set is the product of the distributions on the injectable sets into which it is partitioned.
\item Determine whether the family of distributions obtained in step 2 are the marginals of a single joint distribution. If not, then the original distribution is incompatible with the original causal structure.
\end{enumerate}

{We have just described how to test a specified joint distribution for compatibility with a given causal structure by means of considering an inflation of that causal structure. Passing the inflation-based test is a necessary but not sufficient requirement for the specified joint distribution to be compatible with a given causal structure.} {The procedure is to focus on a particular family of marginals (on the images of injectable sets) of the given joint distribution, then from products of these, obtain the distribution on each of the ai-expressible sets. Finally, one asks simply whether the family of distributions on the ai-expressible sets are consistent in the sense of all being marginals of a single joint distribution. By analogous logic, the following technique allows one to systematically derive causal compatibility inequalities: find the constraints that any family of distributions on the ai-expressible sets must satisfy if these are to be consistent in the sense of all being marginals of a single joint distribution. Next, express each distribution of this family as a product of distributions on the injectable sets, according to \cref{eq:preinjfactor}, and rewrite the constraints in terms of the family of distributions on the injectable sets. These constraints constitute causal compatibility inequalities for the inflated causal structure. Finally, one can rewrite the constraints in terms of the family of distributions on the images of the injectable sets, using \cref{maincorollary}, to obtain causal compatibility inequalities for the original causal structure.}


In summary, we have used the contrapositive of \cref{mainlemma} in order to show:

\begin{theorem}
Let $G'$ be an inflation of $G$. Let a distribution $P_{\obsnodes{G}}$ be given. Consider the family of 
distributions
 $\{ P_{\bm{V}'} : \bm{V}' \in \SmallNamedFunction{AI-ExpressibleSets}{G'} \}$. Following ~\cref{eq:preinjfactor}, each distribution in that set factorizes according to ${P_{\bm{V}'} = \prod_{i{=}1}^n P_{\bm{V}'_i}}$, where the variable subsets $\bm{V}'_1 \cdots \bm{V}'_n$ associated with the factorization are precisely the injectable components of the ai-expressible set $\bm{V}'$. Additionally, for every injectable set $\bm{V}'_i$, let $P_{\bm{V}_i'}= P_{\bm{V}_i}$ where $P_{\bm{V}_i}$ is the marginal on $\bm{V}_i$ of $P_{\obsnodes{G}}$, and where $\bm{V}_i' \sim  \bm{V}_i$. 
 If the family of 
distributions $\{ P_{\bm{V}'} : \bm{V}' \in \SmallNamedFunction{AI-ExpressibleSets}{G'} \}$ does not arise as the family of marginals of some joint distribution, then the original distribution $P_{\obsnodes{G}}$ is not compatible with $G$.
\end{theorem}

The ai-expressible sets play a crucial role in linking the original causal structure with the inflated causal structure. They are precisely those sets of variables whose joint distributions in the inflation model are fully specified by the causal model on the original causal structure, as they can be computed using \cref{eq:preinjfactor} and \cref{mainlemma}. So we begin with the problem of identifying the ai-expressible sets systematically.

\subsection{Identifying the AI-Expressible Sets}
\label{step:findpreinjectable}

To identify the ai-expressible sets of an inflated causal structure $G'$, we must first identify the injectable sets. This problem can be reduced to identifying the injectable pairs of nodes, because if all of the pairs in a set of nodes are injectable, then so too is the set itself. This can be proven as follows. 
Let $\varphi : G' \to G$ be the projection map from $G'$ to the original causal structure $G$, corresponding to removing copy-indices. Then $\varphi$ has the characteristic feature that it preserves and reflects edges: if $A \to B$ in $G'$, then also $\varphi(A) \to \varphi(B)$ in $G$, and vice versa; this follows from the assumption that $G'$ is an inflation of $G$. 
A set $\bm{V} \subseteq \obsnodes{G'}$ is injectable if and only if the
restriction of $\varphi$ to $\An{\bm{V}}$ is an injective map. 
But now injectivity of a map means precisely that no two different
elements of the domain get mapped to the same element of the codomain.
So if $\bm{V}$ is injectable, then so is each of its two-element subsets;
conversely, if $\bm{V}$ is not injectable, then $\varphi$ maps two nodes among the
ancestors of $\bm{V}$ to the same node, which means that there are two nodes in the
ancestry that differ only by copy-index. Each of these two nodes must be
an ancestor of at least some node in $\bm{V}$; if one chooses two such
descendants, then one gets a two-element subset of $\bm{V}$ such that $\varphi$ is not
injective on the ancestry of that subset, and therefore this two-element
set of observed nodes is not injectable.

To enumerate the injectable sets, it is therefore useful to encode certain features of the inflated causal structure in an undirected graph which we call the \tblue{injection graph}. The nodes of the injection graph are the observed nodes of the inflated causal structure, and a pair of nodes $A_i$ and $B_j$ share an edge if the pair $\{ A_i B_j\}$ is injectable. For example, \cref{fig:injection222} shows the injection graph of the Spiral inflation of the Triangle scenario (\cref{fig:Tri222}).
The property noted above states that the injectable sets are precisely the cliques\footnote{A \emph{clique} is a set of nodes in an undirected graph any two of which share an edge.} of the injection graph.
While for many other applications only the maximal cliques are of interest, our application of the inflation technique requires knowledge of all nonempty cliques. 

Given a list of the injectable sets, the ai-expressible sets can be read off from the \tblue{ai-expressability graph}.
The nodes of the ai-expressibility graph are taken to be the injectable sets in $G'$, and two nodes share an edge if the associated injectable sets are ancestrally independent.
\cref{fig:preinjectiongraph222} depicts an example. 
The ai-expressible sets correspond to the cliques of the ai-expressibility graph: the union of all the injectable sets that make up the nodes of a clique is an ai-expressible set, while the individual nodes already give us the partition into injectable sets relevant for the factorization relation of \cref{eq:preinjfactor}. 
For our purposes, it is sufficient to enumerate the maximal ai-expressible sets, so that one only needs to consider the maximal cliques of the ai-expressibility graph.

\begin{figure}[t]
\centering
\begin{minipage}[t]{0.3\linewidth}
\centering
\includegraphics[scale=1]{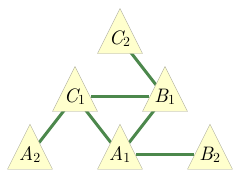}
\caption{The injection graph corresponding to the Spiral inflation of the Triangle scenario (\cref{fig:Tri222}), wherein the cliques are the injectable sets.
}\label{fig:injection222}
\end{minipage}
\hfill
\begin{minipage}[t]{0.33\linewidth}
\centering
\includegraphics[scale=1]{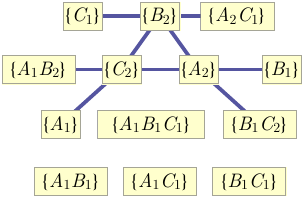}
\caption{The ai-expressibility graph corresponding to the Spiral inflation of the Triangle scenario (\cref{fig:Tri222}), wherein two injectable sets are adjacent iff they are ancestrally independent. A set of nodes is ai-expressible iff it arises as a union of sets that form a clique in this graph.}\label{fig:preinjectiongraph222}
\end{minipage}
\hfill
\begin{minipage}[t]{0.3\linewidth}
\centering
\includegraphics[scale=0.25]{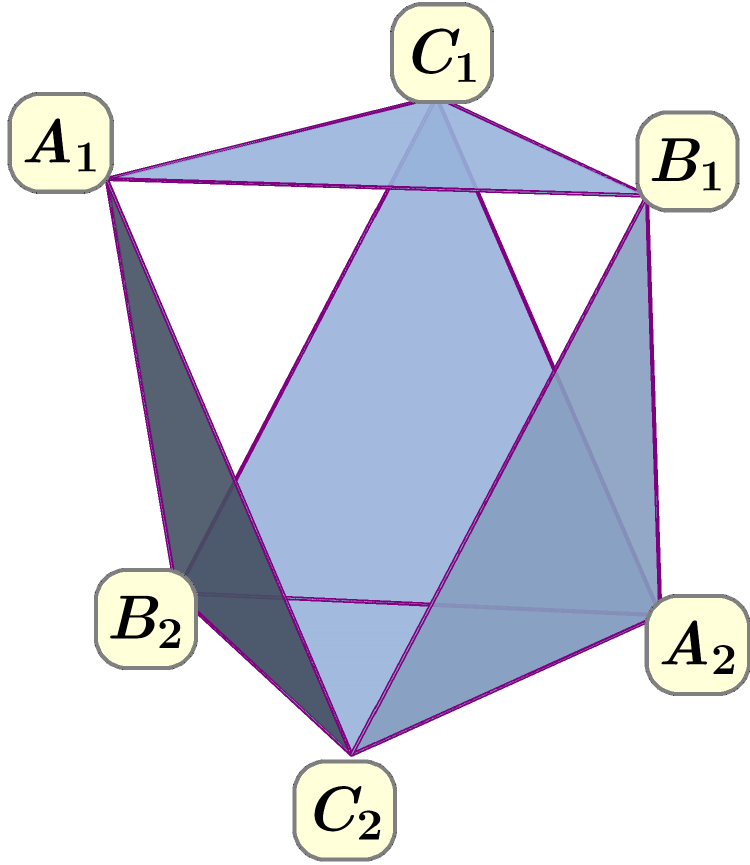}
\caption{The simplicial complex of ai-expressible sets for the Spiral inflation of the Triangle scenario (\cref{fig:Tri222}). The 5 facets correspond to the maximal ai-expressible sets, namely $\{A_1 B_1 C_1\}$, $\{A_1 B_2 C_2\}$, $\{A_2 B_1 C_2\}$, $\{A_2 B_2 C_1\}$ and $\{A_2 B_2 C_2\}$.}\label{fig:simplicialcomplex222}
\end{minipage}
\end{figure}

From \cref{fig:injection222,fig:preinjectiongraph222}, we easily infer the injectable sets and the maximal ai-expressible sets, as well as the partition of the maximal ai-expressible sets into ancestrally independent subsets. For the Spiral example, this results in:
\begin{align}\label{eq:basicsetup222}
\qquad
{\underbrace{\begin{matrix}
\, \\
\brackets{A_1},\:\brackets{B_1},\:\brackets{C_1},\\
\brackets{A_2},\:\brackets{B_2},\:\brackets{C_2},\\
\brackets{A_1 B_1},\:\brackets{A_1 C_1},\:\brackets{B_1 C_1},\\
\brackets{A_1 B_2}.\:\brackets{A_2 C_1},\:\brackets{B_1 C_2},\\
\brackets{A_1 B_1 C_1}
\end{matrix}}_{\substack{\text{The injectable sets}}}}
\qquad\qquad
{\underbrace{\begin{matrix}
\brackets{A_1 B_1 C_1} \\
\brackets{A_1 B_2 C_2} \\
\brackets{B_1 C_2 A_2} \\
\brackets{C_1 A_2 B_2} \\
\brackets{A_2 B_2 C_2}
\end{matrix}}_{\substack{\text{The maximal}\\\text{ai-expressible sets}}}}
\qquad\qquad
{\underbrace{\begin{matrix}
\\
\{A_1 B_2\} \aindep \{ C_2\} \\
\{B_1 C_2\} \aindep \{ A_2\} \\
\{C_1 A_2\} \aindep \{ B_2\} \\
\{A_2\} \aindep \{B_2\} \aindep \{ C_2\}
\end{matrix}}_{\substack{\text{The relevant}\\\text{ancestral independences}}}}
\end{align}

Having identified the ai-expressible sets and how they partition into injectable sets, we now infer the factorization relations implied by ancestral independences, which is \cref{eq:tri222fac} in the Spiral example. Next, we discuss the other ingredient of our systematic procedure: the marginal problem.

\subsection{The Marginal Problem and its Solution}
\label{step:marginalsproblem}

The third step in our procedure is determining whether the given distributions on ai-expressible sets can arise as marginals of one joint distribution on all observed nodes of the inflated causal structure. In general, the problem of determining whether a given family of distributions can arise as marginals of some joint distribution is known as the {\em marginal problem}\footnote{For further references and an outline of the long history of the marginal problem, see~\cite{fritz2013marginal}. An alternative account using the language of presheaves can also be found in~\cite{abramsky_contextuality_2011}.}. In order to derive causal compatibility inequalities, one must solve the closely related problem of determining necessary and sufficient \emph{constraints} that a family of marginal distributions must satisfy in order for the marginal problem to have a solution. For better clarity, we distinguish these two variants of the marginal problem as the \tblue{marginal satisfiability problem} and the \tblue{marginal constraint problem}. The generic \tblue{marginal problem} will be used as an umbrella term referring to both types.

To specify either sort of marginal problem, one must specify the full set of variables to be considered, denoted $\bm{V}$, together with a family of subsets of $\mathbb{S}$, denoted $(\bm{V}_1,\ldots,\bm{V}_n)$ and called \tblue{contexts}. The family of contexts can be visualized through the simplicial complex that it generates, as illustrated in~\cref{fig:simplicialcomplex222}. A \tblue{marginal scenario} consists of a specification of contexts together with a specification of the cardinality of each variable. Every joint distribution $P_{\bm{V}}$ defines a family of marginal distributions $(P_{\bm{V}_1},\ldots,P_{\bm{V}_n})$ through marginalization, $P_{\bm{V}_i} := \sum_{\bm{V} \setminus \bm{V}_i} P_{\bm{V}}$. The marginal problem concerns the converse inference. In the marginal satisfiability problem, a concrete family of distributions $(P_{\bm{V}_1},\ldots,P_{\bm{V}_n})$ is given, and one wants to decide whether there exists a joint distribution $\hat{P}_{\bm{V}}$ such that $P_{\bm{V}_i} = \sum_{\bm{V}\setminus\bm{V}_i} \hat{P}_{\bm{V}}$ for all $i$. In the marginal constraint problem, one seeks to find conditions on the family of distributions $(P_{\bm{V}_1},\ldots,P_{\bm{V}_n})$, considered as parameters, for when a joint distribution $\hat{P}_{\bm{V}}$ exists which reproduces these as marginals, $P_{\bm{V}_i} = \sum_{\bm{V}\setminus\bm{V}_i} \hat{P}_{\bm{V}}$ for all $i$.
 
In order for $\hat{P}_{\bm{V}}$ to exist, distributions on different contexts must be consistent on the intersection of contexts, that is, marginalizing $P_{\bm{V}_i}$ to those variables in the intersection $\bm{V}_i\cap\bm{V}_j$ must result in the same distribution as marginalizing $P_{\bm{V}_j}$ to that intersection. 
 In many cases, this is not sufficient\footnote{Depending on how the contexts intersect with one another, this \emph{may} be sufficient. A precise characterization for when this occurs has been found by~\citet{vorobev_extension_1960}. See also \citet[Thm. 2]{chaves2016limitedmarginals} for an application of this characterization enabling computationally significant shortcuts in solving the marginal constraint problem.}; indeed, we have already seen examples of additional constraints, namely, the inequalities \eqref{eq:polymonogamyraw}, \eqref{eq:MIraw} and \eqref{eq:FritzF3raw} from \cref{Sec:DerivingInequalities}.
So what are the necessary and sufficient conditions? To answer this question, it helps to realize two things:
\begin{itemize}
	\item The set of all valid (positive, normalized) distributions $P_{\bm{V}}$ is precisely the convex hull of the deterministic assignments of values to $\bm{V}$ (the deterministic distributions), and 
	\item The map $P_{\bm{V}}\mapsto (P_{\bm{V}_1},\ldots,P_{\bm{V}_n})$, describing marginalization to each of the contexts in $(\bm{V}_1,\ldots,\bm{V}_n)$, is linear.
\end{itemize}
Hence the \emph{image} of the set of possibilities for the distribution $P_{\bm{V}}$ under the map $P_{\bm{V}}\mapsto (P_{\bm{V}_1},\ldots,P_{\bm{V}_n})$ is exactly the convex hull of the deterministic assignments of values to $(\bm{V}_1,\ldots,\bm{V}_n)$ which are consistent where these contexts overlap. Since there are only finitely many such deterministic assignments, this convex hull is a polytope; it is called the \tblue{marginal polytope}~\cite{kahle_marginal_2010}. Together with the above equations on coinciding submarginals, the facet inequalities of this polytope solve the marginal constraint problem. The marginal satisfiability problem asks about membership in the polytope; by the above, this becomes a linear program with the joint probabilities $P_{\bm{V}}$ as the unknowns.

To express this more concretely, we write the marginal satisfiability problem in the form of a generic linear program.

Let the \tblue{joint distribution vector} $\bm{v}$ be the vector associated with the joint probability distribution $P_{\bm{V}}$, that is, the vector whose components are the probabilities $P_{\bm{V}}(v)$. Let the \tblue{marginal distribution vector} $\bm{b}$ be the vector that is the concatenation over $i$ of the vectors associated with the distributions $P_{\bm{V}_i}$.
Finally, let the \tblue{marginal description matrix} $\bm{M}$ be the matrix representation of the linear map corresponding to marginalization on each of the contexts, that is, 
$P_{\bm{V}}\to(P_{\bm{V}_1},\ldots,P_{\bm{V}_n})$ where 
$P_{\bm{V}_i} = \sum_{\bm{V}\setminus\bm{V}_i} P_{\bm{V}}$.
The components of $\bm{M}$ all take the value zero or one. 

In this notation, the marginal satisfiability problem consists of determining whether, for a given vector $\bm{b}$, the following constraints are feasible:
\begin{align}\label{eq:marginalproblemgeneric}
 \exists\, \bm{v} :{\bm{v} \geq \bm{0}} ,\; \bm{M}\bm{v}=\bm{b},
\end{align}
where the component-wise inequality $\bm{v}\geq\bm{0}$ enforces the constraint that $P_{\bm{V}}$ is a nonnegative probability distribution. 
This is clearly a linear program.

In the example of \cref{fig:simplicialcomplex222} with binary variables, $\bm{M}$ is a $48\times 64$ matrix, so that $\bm{M}\bm{v}=\bm{b}$ represents 48 equations and $\bm{v} \geq \bm{0}$ represents 64 inequalities; explicit representations of $\bm{M}$, $\bm{v}$, and $\bm{b}$ for the simpler example of the Cut inflation can be found in \cref{sec:explicitmatrix}.
A single linear program can then assess whether there is a solution in $\bm{v}$ for a given marginal distribution vector $\bm{b}$. If this is not the case, then the marginal satisfiability problem has a negative answer. 

Since linear programming is quite easy, probing specific distributions for compatibility for a given inflated causal structure is computationally inexpensive. For instance, using the Web inflation of the Triangle scenario (\cref{fig:TriFullDouble}), which contains a large number of observed variables, our numerical computations have reproduced the result of~\cite[Theorem~2.16]{fritz2012bell}, that a certain distribution considered therein is incompatible with the Triangle scenario\footnote{This distribution {\em is}, however, quantum-compatible with the Triangle scenario (\cref{sec:classicallity}).}.

In the case of the marginal constraint problem, the vector $\bm{b}$ is not given, but one rather wants to find conditions on $\bm{b}$ that hold if and only if \cref{eq:marginalproblemgeneric} has a solution. As per the above, this is a problem of \tblue{facet enumeration}\footnote{In \cref{sec:projalgorithms}, we provide an overview of techniques for facet enumeration.} for the marginal polytope. Equivalently, it is the problem of \tblue{linear quantifier elimination}\footnote{Linear quantifier elimination has already been used in causal inference for deriving entropic causal compatibility inequalities \cite{chaves2014novel,chaves2014informationinference}. In that task, however, the unknowns being eliminated are entropies on sets of variables of which one or more is latent. By contrast, the unknowns being eliminated above are all probabilities on sets of variables all of which are observed---but on the inflated causal structure rather than the original causal structure.} for the system of \cref{eq:marginalproblemgeneric}: one tries to find a system of linear equations and inequalities in $\bm{b}$ such that some $\bm{b}$ satisfies the system if and only if \cref{eq:marginalproblemgeneric} has a solution. There is a unique minimal system achieving this, and it consists of the constraints of consistency on the intersections of contexts (mentioned above), together with the facet inequalities of the marginal polytope. Taken together, these form a system of linear equations and inequalities that is equivalent to \cref{eq:marginalproblemgeneric}, but does not contain any quantifiers. In our application, the equations expressing consistency on the intersections of contexts are guaranteed to hold automatically, so that only the facet inequalities are of interest to us.

In terms of \cref{eq:marginalproblemgeneric}, a valid inequality for the marginal distribution vector $\bm{b}$---such as a facet inequality of the marginal polytope---can always be expressed as $\bm{y}^T\bm{b}\geq 0$ for some vector $\bm{y}$.
Validity of an inequality $\bm{y}^T\bm{b}\geq 0$ means precisely that $\bm{y}^T\bm{M}\geq\bm{0}$, since the columns of $\bm{M}$ are the vertices of the marginal polytope. The marginal satisfiability problem for a given vector $\bm{b}_0$ has no solution if and only if there is a vector $\bm{y}$ that yields a valid inequality but for which $\bm{y}^T\bm{b}_0 < 0$. Necessity follows by noting that if \cref{eq:marginalproblemgeneric} does have a solution $\bm{v}$ for a given vector $\bm{b}_0$, then the fact that $\bm{y}^T\bm{M}\geq \bm{0}$ and the fact that $\bm{v}\geq \bm{0}$ implies that $\bm{y}^T\bm{b}_0 = \bm{y}^T\bm{M}\bm{v} \geq 0$. Sufficiency follows from Farkas' lemma. Most linear programming tools are capable of returning a \emph{Farkas infeasibility certificate} \cite{infeasibilitycertificates} whenever a linear program has no solution. In our case, if the marginal problem is infeasible for a vector $\bm{b}_0$, then the certificate is a vector $\bm{y}$ that yields a valid inequality but for which $\bm{y}^T\bm{b}_0 < 0$.\footnote{Farkas infeasibility certificates are available in \href[pdfnewwindow]{http://docs.mosek.com/8.0/pythonapi/optimizer-task-gety.html}{\textit{Mosek}}, \href[pdfnewwindow]{https://www.gurobi.com/documentation/6.5/refman/farkasdual.html}{\textit{Gurobi}}, and \href[pdfnewwindow]{http://www-01.ibm.com/support/docview.wss?uid=swg21400058}{\textit{CPLEX}}, as well as by accessing dual variables in \href[pdfnewwindow]{http://cvxr.com/cvx/doc/basics.html\#dual-variables}{\textit{cvxr}}/\href[pdfnewwindow]{http://cvxopt.org/userguide/coneprog.html\#linear-cone-programs}{\textit{cvxopt}}.}.

Upon substituting the factorization relations of \cref{eq:preinjfactor} and deleting copy indices, any valid inequality for the marginal problem
turns into a causal compatibility inequality. This applies both to facet inequalities of the marginal polytope, and to Farkas infeasibility certificates. In the latter case, one obtains an explicit causal compatibility inequality which witnesses the given distribution as incompatible with the given causal structure. In other words, if a given distribution is witnessed as incompatible with a causal structure using the technique we have described, then with little additional numerical effort, one can also obtain a causal compatibility inequality that exhibits the incompatibility. This may have applications to problems where the facet enumeration is computationally intractable.

Summarizing, we have shown how to leverage the marginal satisfiability problem to witness causal incompatibility of particular distributions, and how to leverage the marginal constraint problem to derive causal compatibility inequalities.

\subsection{A List of Causal Compatibility Inequalities for the Triangle scenario}
\label{sec:CCineqs}

As an example of the above method, we have considered the Triangle scenario with binary observed variables and derived all causal compatibility inequalities which follow by means of using ancestral independences in the Spiral inflation (\cref{fig:Tri222}). We found that there are 4884 inequalities corresponding to the facets of the relevant marginal polytope, which results in 4884 polynomial causal compatibility inequalities for the Triangle scenario.

However, most inequalities in this set have turned out to be redundant, where an inequality is considered redundant if there is no distribution that violates this inequality but none of the others. We thus have looked for a subset of inequalities that is \emph{irredundant} (does not contain any redundant inequality) but nevertheless \emph{complete} (defines the same set of distributions as the full set). While a finite system of linear inequalities always has a unique irredundant complete subset, this need not be the case for finite systems of polynomial inequalities; we therefore speak of ``a'' complete irredundant set instead of ``the'' complete irredundant set.

We exploited linear programming techniques to quickly identify a 1433-inequalities complete subset of our original 4884 inequalities; concretely, the copy isomorphisms of \cref{sec:coincidingdetails} yield an additional list of linear equations satisfied by all inflation models, and from every set of inequalities that differ by merely a linear combination of these equations we choose one representative. To further prune away redundant inequalities, we successively employed nonlinear constrained maximization on each inequality's left-hand-side, to determine numerically if it could be violated pursuant to all the other inequalities as constraints. An inequality is found to be redundant if the solution to the constrained maximization does not exceed the inequality's right-hand side. Such an inequality was immediately dropped from the set before testing the next candidate for redundancy\footnote{It is advantageous to group the inequalities into symmetry classes \emph{prior} to pruning away redundant inequalities, so that entire classes of inequalities can be discarded when finding that a single representative is redundant \emph{to the other classes}.}. This post-processing led us to identify 60 irredundant inequalities which defines the same set of satisfying distributions as the original 4884. Of the remaining 60, we recognized 8 as uninteresting positivity inequalities, $\p[ABC]{abc}\geq 0$, so that our irredundant complete system consists of 52 polynomial inequalities.

To present those inequalities in an efficient manner, we further grouped them into four symmetry classes. In \cref{eq:ghzrejectclass,eq:wrejectclass,eq:symclass3,eq:symclass4} we present one representative from each class; the multiplicity of inequalities contained in each symmetry class is marked in parentheses. 
The symmetry group for any causal structure with finite-cardinality observed variables is generated by those permutations of the observed variables which can be extended to automorphisms of the (original) DAG, as well as any permutation among the discrete values assigned to an individual observed variable (i.e., bijections on the sample space of that variable). 
In the case of the Triangle scenario with binary observed variables, the symmetry group therefore has 48 elements, comprised of the 6 permutations of the three observed variables, the three local binary-value relabellings, and all their compositions $(48=6\times 2\times 2\times 2)$. 

We choose to express our inequalities in terms of correlators (where the two possible values of each variables to be $\{-1,+1\}$), rather than in terms of joint probabilities, because such a presentation is more compact:
\begin{align}\label{eq:ghzrejectclass}
0\leq 1& - \expec{A C} - \expec{B C} + \expec{A} \expec{B} &(\times 12)
\\\label{eq:wrejectclass}
\begin{split}0\leq 3& - \expec{A} - \expec{B} - \expec{C} +2 \expec{A B} +2 \expec{A C} +2 \expec{B C} \\&\quad+ \expec{A B C}+ \expec{A} \expec{B} + \expec{A} \expec{C} + \expec{B} \expec{C} \\&\quad- \expec{A} \expec{B C} - \expec{B}
 \expec{A C} - \expec{C} \expec{A B} + \expec{A} \expec{B} \expec{C} \end{split}&(\times 8)
 \\\label{eq:symclass3}\begin{split}
0\leq 4& +2 \expec{C} -2 \expec{A B} -3 \expec{A C} -2 \expec{B C} - \expec{A B C} + \expec{A} \expec{B} \expec{C} 
\\&\quad +2 \expec{A} \expec{B} + \expec{A} \expec{C} - \expec{A} \expec{B C} - \expec{C} \expec{A B} \end{split}&(\times 24)
 \\\label{eq:symclass4}\begin{split}
0\leq 4& -2 \expec{A B} -2 \expec{A C} -2 \expec{B C} - \expec{A B C} 
 \\&\quad +2 \expec{A} \expec{B} +2 \expec{A} \expec{C} +2 \expec{B} \expec{C} 
 \\&\quad - \expec{A} \expec{B C} - \expec{B} \expec{A C} - \expec{C} \expec{A B} \end{split}&(\times 8)
\end{align}
All the inequalities \labelcref{eq:ghzrejectclass,eq:wrejectclass,eq:symclass3,eq:symclass4} have no slack in the sense that they can be saturated by distributions compatible with the Triangle scenario. Indeed, all the inequalities are saturated by the \emph{deterministic} distribution $\expec{A}\eql\expec{B}\eql\expec{C}\eql 1$, except for \cref{eq:wrejectclass} which is saturated by the deterministic distribution $\expec{A}\eql\expec{B}\eql-\expec{C}\eql 1$. Generally speaking, any polynomial inequality generated by a facet of the marginal polytope (i.e. corresponding to some linear inequality in the variables of the inflated causal structure) will be saturated by some deterministic distribution.

A machine-readable and closed-under-symmetries version of this list of inequalities may be found in \cref{sec:38ineqs}.

\subsection{Causal Compatibility Inequalities via Hardy-type Inferences from Logical Tautologies}\label{sec:TSEM}

Enumerating all the facets of the marginal polytope is computationally feasible only for small examples. But our method transforms \emph{every} inequality that bounds the marginal polytope into a causal compatibility inequality. We now present a general approach for deriving a special type of such inequalities very quickly.

In the literature on Bell inequalities, it has been noticed that incompatibility with the Bell causal structure can sometimes be witnessed by merely looking at which joint outcomes have zero probability and which ones have nonzero probability. In other words, instead of considering the \emph{probability} of an outcome, the inconsistency of some marginal distributions can be evident from considering only the \emph{possibility} or \emph{impossibility} of each outcome. This insight is originally due to~\citet{L.Hardy:PRL:1665}, and versions of Bell's theorem that are based on the violation of such \tblue{possibilistic constraints} are known as \tblue{Hardy-type paradoxes}~\cite{Garuccio95,CabelloHardyInequality,Braun08,Mancinska14,LSW}; a partial classification of these can be found in~\cite{Mansfield2012}. The method that we describe in the second half of this section can be used to compute a complete classification of possibilistic constraints for \emph{any} marginal problem.

Possibilistic constraints follow from a consideration of {\em logical relations} that can hold among deterministic assignments to the observed variables. Such logical constraints can also be leveraged to derive probabilistic constraints instead of possibilistic ones, as shown in~\cite{Pitowsky1989,Ghirardi08}. This results in a partial solution to any given (probabilistic) marginal problem. Essentially, we solve a possibilistic marginal problem \cite{Mansfield2012}, then upgrade the possibilistic constraints into probabilistic inequalities, resulting in a set of probabilistic inequalities whose satisfaction is a necessary but insufficient condition for satisfying the corresponding probabilistic marginal problem. We now demonstrate how to systematically derive all inequalities of this type.

We have already provided a simple example of a Hardy-type argument in \cref{example:noWdist}, in the logic used to demonstrate that the family of distributions of \cref{W4,W1,W5} cannot arise as the marginals of a single joint distribution. For our present purposes, it is useful to recast that argument into a new but manifestly equivalent form. First, for the family of distributions in question, we have
\begin{align} 
\begin{split}\label{WWs}
&A_2 \eql 1 \:\implies\: C_1 \eql 0,\\
&B_2 \eql 1 \:\implies\: A_1 \eql 0,\\
&C_2 \eql 1 \:\implies\: B_1 \eql 0,\\
\text{Never} &\quad A_1 \eql 0\,\text{ and }\, B_1 \eql 0\,\text{ and }\, C_1 \eql 0.
\end{split}
\end{align}
From the last constraint one infers that at least one of $A_1$, $B_1$ and $C_1$ must be 1, which from the three other constraints implies that at least one of $A_2$, $B_2$ and $C_2$ must be 0, so that it is not the case that all of $A_2$, $B_2$ and $C_2$ are 1. Thus~\cref{WWs} implies
\begin{align} \label{consequent2}
\text{Never} \quad &A_2 \eql 1\,\text{ and }\, B_2 \eql 1\,\text{ and }\, C_2 \eql 1.
\end{align}
However, the Spiral inflation (\cref{fig:Tri222}) is such that $A_2$, $B_2$, and $C_2$ have no common ancestor and consequently the distribution on the ai-expressible set $\{A_2 B_2 C_2\}$ is the product of the distributions on $A_2$, $B_2$ and $C_2$. Since each of the latter has full support (\cref{W5}), it follows that the distribution on $\{A_2 B_2 C_2\}$ also has full support, which contradicts \cref{consequent2}.

We are here interested in recasting the argument in a manner amenable to systematic generalization. This is done as follows. We work in a marginal scenario where the contexts are $\{A_2 B_2 C_2\}$, $\{A_2 C_1\}$, $\{B_2 A_1\}$, $\{C_2 B_1\}$, and $\{A_1 B_1 C_1\}$, and all variables are binary. The first step of the argument is to note that\footnote{Here, $\land$, $\lor$ and $\lnot$ denote conjunction, disjunction and negation respectively.}
\begin{align}\begin{split}\label{tautology1}
&\lnot [A_2 \eql 1, C_1 \eql 1] \bigwedge \lnot [B_2 \eql 1, A_1 \eql 1] \bigwedge \lnot [C_2 \eql 1, B_1 \eql 1] \bigwedge \lnot [A_1 \eql 0, B_1 \eql 0, C_1 \eql 0]\\
 &\qquad\implies
\lnot [A_2 \eql 1, B_2 \eql 1, C_2 \eql 1].
\end{split}\end{align}
is a logical tautology for binary variables. It can be understood as a constraint on marginal {\em deterministic assignments}, which can be thought of as a logical counterpart of a linear inequality bounding the marginal polytope. The second and final step of the argument notes that the given marginal distributions are such that the antecedent is always true, while the consequent is sometimes false.

To see how to translate this into a constraint on marginal {\em distributions}, we rewrite \cref{tautology1} in its contrapositive form,
\begin{align}\begin{split}\label{tautology2}
&[\mgreen{A_2 \eql 1} , \mgreen{B_2 \eql 1} , \mgreen{C_2 \eql 1}] \implies [\mgreen{A_2 \eql 1}, C_1 \eql 1] \lor [\mgreen{B_2 \eql 1}, A_1 \eql 1] \lor [\mgreen{C_2 \eql 1}, B_1 \eql 1] \lor [A_1 \eql 0, B_1 \eql 0, C_1 \eql 0].
\end{split}\end{align}
Next, we note that if a logical tautology can be expressed as
\begin{align}\label{eq:inference}
 E_0 \implies E_1 \lor \ldots \lor E_n,
\end{align}
then by applying the union bound---which asserts that the probability of at least one of a set of events occurring is no greater than the sum of the probabilities of each event occurring---one obtains
\begin{align}\label{eq:possinference}
\p{E_0}\leq \sum\limits_{j=1}^n{\p{E_j}}.
\end{align}
Applying this to \cref{tautology2} in particular yields
\begin{align}\label{eq:F3rawweak}
P_{A_2 B_2 C_2}\parens{\mgreen{1} \mgreen{1} \mgreen{1}} \leq P_{A_1 B_2}\parens{1 \mgreen{1 }}+P_{ B_1 C_2}\parens{ 1 \mgreen{1}}+P_{A_2 C_1}\parens{\mgreen{1 } 1}+P_{A_1 B_1 C_1}\parens{0 0 0},
\end{align}
which is a constraint on the marginal {\em distributions}.
 
This inequality allows one to demonstrate the incompatibility of the family of distributions of \cref{W4,W1,W5} with the Spiral inflation just as easily as one can with the tautology of \cref{tautology1}. The fact that $A_2$, $B_2$ and $C_2$ are ancestrally independent in the Spiral inflation implies that $P_{A_2 B_2 C_2} = P_{A_2} P_{B_2} P_{C_2}$. 
 It then suffices to note that for the given family of distributions, the probability on the left-hand side of \cref{eq:F3rawweak} is nonzero (which corresponds to the consequent of \cref{tautology1} being sometimes false) while every probability on the right-hand side is zero (which corresponds to the antecedent of \cref{tautology1} being always true). But, of course, the inequality can witness many other incompatibilities in addition to this one.

As another example, consider the marginal problem where the variables are $A$, $B$ and $C$, with each being binary, and the contexts are the pairs $\{AB\}$, $\{AC\}$, and $\{BC\}$. 
The following tautology provides a constraint on marginal deterministic assignments:\footnote{This is a tautology since $E \land F \implies E \land F \land (G \lor \lnot G) = (E \land F \land G) \lor (E\land F \land \lnot G) \implies (E \land G) \lor (F \land \lnot G)$.}
\begin{align}\label{GHZtautology}
 \bracks{\mgreen{A \eql 0}, \mgreen{C \eql 0}} \implies \bracks{\mgreen{A \eql 0}, B \eql 0} \lor \bracks{B \eql 1, \mgreen{C \eql 0}}.
\end{align}
Applying the union bound, one obtains a constraint on marginal distributions,\footnote{This inequality is equivalent to \cref{eq:polymonogamyraw}.}
\[
P_{AC}(\mgreen{0 0}) \leq P_{AB}(\mgreen{0} 0) + P_{BC}(1 \mgreen{0}).
\]

In this section, we seek to determine, for any marginal scenario, the set of \emph{all} inequalities that can be derived in this manner. We do so by \tblue{enumerating} the full set of tautologies of the form of \cref{tautology1,GHZtautology}. This boils down to solving the possibilistic version of the marginal constraint problem.

We now describe the general procedure. As before, we express a constraint on marginal deterministic assignments as a logical implication, having a valuation (assignment of outcomes) on one context as the \tblue{antecedent} and a disjunction over valuations on contexts as the \tblue{consequent}. In the following, we explain how to generate \emph{all} such implications which are tight in the sense that the consequent is minimal, i.e., involves as few terms as possible in the disjunction. 

First, we fix the antecedent by choosing some context and a joint valuation of its variables. In order to generate all constraints on marginal deterministic assignments, one will have to perform this procedure for \emph{every} context as the antecedent and \emph{every} choice of valuation thereon. For the sake of concreteness, we take the above Spiral inflation example with $\bracks{\mgreen{A_2 \eql 1}, \mgreen{B_2 \eql 1}, \mgreen{C_2 \eql 1}}$ as the antecedent. 
Each logical implication we consider is required to have the property that any variable that appears in both the antecedent and the consequent must be given the same value in both. 

To formally determine all valid consequents, it is useful to introduce two hypergraphs associated to the problem. Recall the definition of the incidence matrix of a hypergraph: if vertex $i$ is contained in edge $j$ of the hypergraph, the component in the $i$th row and $j$th column of the matrix is 1; otherwise it is 0. 

The first hypergraph we consider is the one whose incidence matrix is the marginal description matrix $\bm{M}$ for the marginal problem being considered, as introduced near \cref{eq:marginalproblemgeneric}. 
Each vertex in this hypergraph corresponds to a valuation on some particular context. Each hyperedge corresponds to a possible joint valuation of \emph{all} the variables. A hyperedge contains a vertex if the valuation represented by the hyperedge is an extension of the valuation represented by the vertex. For example, the hyperedge $\bracks{A_1 \eql 0, \mgreen{A_2 \eql 1}, B_1 \eql 0, \mgreen{B_2 \eql 1}, C_1 \eql 1, \mgreen{C_2 \eql 1}}$ contains the vertex $\bracks{A_1 \eql 0, \mgreen{B_2 \eql 1}, \mgreen{C_2 \eql 1}}$. In our example following \cref{fig:simplicialcomplex222}, this initial hypergraph has $5\cdot 2^3 = 40$ vertices and $2^6 = 64$ hyperedges. 

The second hypergraph is a subhypergraph of the first one. We delete from the first hypergraph all vertices and hyperedges which contradict the outcomes supposed by the antecedent. In our example, because the vertex $\bracks{\mgreen{A_2 \eql 1}, \mred{B_2 \eql 0}, C_1 \eql 1}$ contradicts the antecedent $\bracks{\mgreen{A_2 \eql 1}, \mgreen{B_2 \eql 1}, \mgreen{C_2 \eql 0}}$, we delete it. We also delete the vertex corresponding to the antecedent itself. In our example, this second hypergraph has $2^3 + 3\cdot 2^1 = 14$ vertices and $2^3 = 8$ hyperedges.

All valid (minimal) consequents are (minimal) \tblue{transversals} of this latter hypergraph. A transversal is a set of vertices which has the property that it intersects every hyperedge in at least one vertex. In order to get implications which are as tight as possible, it is sufficient to enumerate only the minimal transversals. Doing so is a well-studied problem in computer science with various natural reformulations and for which manifold algorithms have been developed~\cite{eiter_dualization_2008}.

In our example, it is not hard to check that the consequent of
\begin{align}\begin{split}\label{eq:F3implicationform}
	\bracks{\mgreen{A_2 \eql 1}, \mgreen{B_2 \eql 1}, \mgreen{C_2 \eql 1}} \quad\Longrightarrow\quad & \bracks{A_1 \eql 1, \mgreen{B_2 \eql 1}, \mgreen{C_2 \eql 1}} \lor \bracks{\mgreen{A_2 \eql 1}, B_1 \eql 1, \mgreen{C_2 \eql 1}} \\
	 \lor\: & \bracks{\mgreen{A_2 \eql 1}, \mgreen{B_2 \eql 1}, C_1 \eql 1} \lor \bracks{A_1 \eql 0, B_1 \eql 0, C_1 \eql 0}
\end{split}\end{align}
is such a minimal transversal: every assignment of values to all variables which extends the assignment on the left-hand side satisfies at least one of the terms on the right, but this ceases to hold as soon as one removes any one term on the right. 

We convert these implications into inequalities in the usual way via the union bound (i.e., replacing ``$\Rightarrow$'' by ``$\leq$'' at the level of probabilities and the disjunction by summation). Thus \cref{eq:F3implicationform} translates into the constraint on marginal distributions
\begin{align}\label{eq:F3rawprobform}
 P_{A_2 B_2 C_2}\parens{\mgreen{1} \mgreen{1} \mgreen{1}} \leq P_{A_1 B_2 C_2}\parens{1 \mgreen{1 1}}+P_{A_2 B_1 C_2}\parens{\mgreen{1} 1 \mgreen{1}}+P_{A_2 B_2 C_1}\parens{\mgreen{1 1} 1} + P_{A_1 B_1 C_1}\parens{0 0 0}.
\end{align}
This inequality constitutes a strengthening of \cref{eq:F3rawweak} that we had used as \cref{eq:FritzF3raw} as the starting point for deriving a causal compatibility inequality for the Triangle scenario, \cref{eq:FritzF3}. 

Inequalities that one derives from hypergraph transversals are generally weaker than those that result from a complete solution of the marginal problem. Nevertheless, many Bell inequalities are of this form, the CHSH inequality among them \cite{Ghirardi08}. So it seems that this method is still sufficiently powerful to generate plenty of interesting inequalities. At the same time, the method is significantly less computationally costly than the full-fledged facet enumeration, even if one does it for every possible antecedent. Interestingly, \emph{all} of the irredundant polynomial inequalities represented in \cref{eq:ghzrejectclass,eq:wrejectclass,eq:symclass3,eq:symclass4} are found to be derivable by means of hypergraph transversals.

In conclusion, facet enumeration is the preferred method for deriving inequalities for the marginal problem when it is computationally tractable. When it is not, enumerating hypergraph transversals presents a good alternative. 


\section{Further Prospects for the Inflation Technique}\label{sec:otherprospects}

\cref{mainlemma} and \cref{maincorollary} state that any causal inference technique on an inflated causal structure $G'$ can be transferred to the original causal structure $G$. In the previous section, we have found that even extremely weak techniques on $G'$---namely the constraints implied by the existence of a joint distribution together with ancestral independences---can lead to significant and new results for causal inference on $G$. In the following three subsections, we consider some additional possibilities for constraints that might be exploited in this way to enhance the power of inflation further.

\subsection{Appealing to \textit{d}-Separation Relations in the Inflated Causal Structure beyond Ancestral Independance}\label{sec:fulldsep}

In \cref{sec:ineqs}, we considered the inflation technique using sets of observed variables on the inflated causal structure that were 
ai-expressible, that is, that can be written as a union of injectable sets that are ancestrally independent.
However, it is standard practice when deriving causal compatibility conditions for a causal structure to make use not just of ancestral independences, but of arbitrary $d$-separation relations among variables, and for this reason we had also introduced the notion of \emph{expressible set} in \cref{sec:ineqs}. We now comment on the utility of general expressible sets for the inflation technique.

In a given causal structure, if sets of variables $\bm{X}$ and $\bm{Y}$ are $d$-separated\footnote{The notion of $d$-separation is treated at length in~\cite{pearl2009causality,studeny2005probabilistic,WoodSpekkens,pusey2014gdag}, so we elect not to review it here.} by $\bm{Z}$, denoted $\bm{X}\perp_d \bm{Y}|\bm{Z}$, then a distribution is compatible with that causal structure only if it satisfies the conditional independence relation $\bm{X}\indep\bm{Y}|\bm{Z}$, that is,
 $\forall{\bm{x} \bm{y} \bm{z}}: P_{\bm{X}\bm{Y}|\bm{Z}}(\bm{x}\bm{y}|\bm{z})=P_{\bm{X}|\bm{Z}}(\bm{x}|\bm{z})P_{\bm{Y}|\bm{Z}}(\bm{y}|\bm{z})$. In terms of unconditioned probabilities, this reads
\begin{equation}\label{CIuncondprobs}
\forall{\bm{x} \bm{y} \bm{z}}: \p[\bm{X}\bm{Y}\bm{Z}]{\bm{x}\bm{y}\bm{z}}\p[\bm{Z}]{\bm{z}}=\p[\bm{X}\bm{Z}]{\bm{x}\bm{z}}\p[\bm{Y}\bm{Z}]{\bm{y}\bm{z}}.
\end{equation}
For $\bm{Z} = \emptyset$, $d$-separation of $\bm{X}$ and $\bm{Y}$ relative to $\bm{Z}$ is simply ancestral independence of $\bm{X}$ and $\bm{Y}$, and we infer factorization of the distribution on $\bm{X}$ and $\bm{Y}$. So it is natural to ask: can the inflation technique make use of arbitrary $d$-separation relations among sets of observed variables? 

The answer is that it can. 
Consider an inflation $G'$ wherein $\bm{X}'$ and $\bm{Y}'$ are $d$-separated by $\bm{Z}'$ and moreover where the sets $\bm{X}'\cup\bm{Z}'$, $\bm{Y}'\cup\bm{Z}'$ and $\bm{Z}'$ are injectable. In such an instance, the distribution on $\cup\bm{X}'\cup\bm{Y}'\cup\bm{Z}'$ can be inferred exclusively from distributions on injectable sets,
\begin{align}\label{eq:dsepsub}
 \p[\bm{X}'\bm{Y}'\bm{Z}']{\bm{x}\bm{y}\bm{z}}=\begin{cases}\frac{\p[\bm{X}'\bm{Z}']{\bm{x}\bm{z}}\p[\bm{Y}'\bm{Z}']{\bm{y}\bm{z}}}{\p[\bm{Z}']{\bm{z}}} & \quad\text{if }\p[\bm{Z}']{\bm{z}}>0, \\
 0 & \quad\text{if }\p[\bm{Z}']{\bm{z}}= 0.
 \end{cases}
\end{align}
It follows that if one includes expressible sets such as $\bm{X}'\cup\bm{Y}'\cup\bm{Z}'$ in the set of contexts defining the marginal problem,
then this simply increases the number of given marginal distributions, and one can solve the marginal problem as before by linear programming techniques. In the case where one derives inequalities on the marginal distributions, these remain linear inequalities, but ones that now include the joint probabilities $\p[\bm{X}'\bm{Y}'\bm{Z}']{\bm{x}\bm{y}\bm{z}}$. Upon substituting conditional independence relations such as \cref{eq:dsepsub} in order to derive causal compatibility inequalities, one still ends up with polynomial inequalities, as in the case of using ai-expressible sets only, after multiplying by the denominators.
As before, these causal compatibility inequalities for the inflation are translated into polynomial causal compatibility inequalities for the original causal structure per \cref{maincorollary}. 

In~\cref{sec:interestingproof}, we provide a concrete example of how a $d$-separation relation distinct from ancestral independence can be useful both for the problem of witnessing the incompatibility of a specific distribution with a causal structure and for the problem of deriving causal compatibility inequalities.

Per \cref{def:expressible}, the notion of expressibility is recursive: The set $\bm{X}'\cup\bm{Y}'\cup\bm{Z}'$ is expressible if $\bm{X}'\aindep\bm{Y}'|\bm{Z}'$ and $\bm{X}'\cup\bm{Z}'$, $\bm{Y}'\cup\bm{Z}'$ and $\bm{Z}'$ are all expressible. 
In general, one can obtain stronger causal compatibility inequalities, and stronger witnessing power when testing the compatibility of a specific distribution, by determining the maximal expressible sets instead of restricting attention to the maximal ai-expressible sets. 

\subsection{Imposing Symmetries from Copy-Index-Equivalent Subgraphs of the Inflated Causal Structure}\label{sec:copyindexequivalence}

By the definition of an inflation model (\cref{def:inflat}), if two variables in the inflated causal structure $G'$ are copy-index-equivalent, $A_i \sim A_j$, then each depends on its parents in the same fashion as $A$ depends on its parents in the original causal structure $G$, meaning that $\pfunc{A_i| \Pa[G']{A_i}}=\pfunc{A|\Pa[G]{A}}$ and $\pfunc{A_j| \Pa[G']{A_j}}=\pfunc{A|\Pa[G]{A}}$. Thus by transitivity, also $A_i$ and $A_j$ have the same dependence on their parents,
\begin{align}\label{eq:copyindexequivalence}
\pfunc{A_i| \Pa[G']{A_i}}=\pfunc{A_j|\Pa[G']{A_j}}.
 \end{align}
The ancestral subgraphs of $A_i$ and $A_j$ are also equivalent, and consequently equations like \cref{eq:copyindexequivalence} also hold for all of the ancestors of $A_i$ and $A_j$. We conclude that the marginal distributions of $A_i$ and $A_j$ must also be equal, $\pfunc{A_i}=\pfunc{A_j}$.
More generally, it may be possible to find pairs of contexts in $G'$ of any size such that constraints of the form of~\cref{eq:copyindexequivalence} imply that the marginal distributions on these two contexts must be equal. 

For example, consider the pair of contexts $\brackets{A_1 A_2 B_1}$ and $\brackets{A_1 A_2 B_2}$ for the marginal scenario defined by the Spiral inflation (\cref{fig:Tri222}). Neither of these two contexts is an injectable set. Nonetheless, because of~\cref{eq:copyindexequivalence}, we can conclude that their marginal distributions coincide in any inflation model,
\begin{align}\label{CIEconsequence}
\forall{a a' b}:\;\p[A_1 A_2 B_1]{a a' b} = \p[A_1 A_2 B_2]{a a' b}.
\end{align}
We can similarly conclude that in the inflation model these marginal distributions satisfy $P_{A_1 A_2 B_1}=P_{A_2 A_1 B_2}$---where now the order of $A_1$ and $A_2$ is opposite on the two sides of the equation---or equivalently, 
\begin{align}\label{CIEconsequence2}
\forall{a a' b}:\;\p[A_1 A_2 B_1]{a a' b} = \p[A_1 A_2 B_2]{a' a b}.
\end{align}
These constraints entail that $P_{A_1 A_2 B_2}$ must be symmetric under exchange of $A_1$ and $A_2$, which in itself is another equation of the type above.

Parameters such as $\p[A_1 A_2 B_1]{a_1 a_2 b}$, $\p[A_1 A_2 B_2]{a_1 a_2 b}$ and $\p[A_1 A_2]{a_1 a_2}$ can each be expressed as sums of the unknowns $\p[A_1 A_2 B_1 B_2 C_1 C_2]{a_1 a_2 b_1 b_2 c_1 c_2}$, so that each equation like \cref{CIEconsequence,CIEconsequence2} can be added to the system of equations and inequalities that constitute the starting point of the satisfiability problem (if one is seeking to test the compatibility of a given distribution with the inflated causal structure) or the quantifier elimination problem (if one is seeking to derive causal compatibility inequalities for the inflated causal structure). If any such additional relation yields stronger constraints at the level of the inflated causal structure, then one may obtain stronger constraints at the level of the original causal structure.

The general problem of finding pairs of contexts in the inflated causal structure for which relations of copy-index-equivalence imply equality of the marginal distributions, and the conditions under which such equalities may yield tighter inequalities, are discussed in more detail in \cref{sec:coincidingdetails}.

\subsection{Incorporating Nonlinear Constraints} \label{sec:nonlinearsatproblem}

In deriving causal compatibility inequalities and in witnessing causal incompatibility of a specific distribution, we restricted ourselves to starting from the marginal problem where the contexts are the (ai-)expressible sets, and wherein one imposes only linear constraints derived from the marginal problem. In this approach, facts about the causal structure only get incorporated in the construction of the marginal distribution on each expressible set, and the quantifier elimination step of the computational algorithm is linear. However, one can also incorporate facts about the causal structure as constraints on the quantifier elimination problem at the cost making the quantifier elimination problem nonlinear. 

Take the Spiral inflation of the Triangle scenario as an example. There is an ancestral independence therein that we did not use in our previous application of the inflation technique, namely, $A_1 A_2 \perp_d C_2$. It was not used because $\{A_1 A_2 C_2\}$ is not an expressible set. Nonetheless, we can incorporate this ancestral independence as an additional constraint in the quantifier elimination problem, namely,
\begin{align}\label{nonlinearequality2}
\forall{a_1 a_2 c_2}: \p[A_1 A_2 C_2]{a_1 a_2 c_2}=\p[A_1 A_2 ]{a_1 a_2 }\p[ C_2]{c_2}.
\end{align}
Recall that in the marginal problem, one seeks to eliminate the unknowns $\p[A_1 A_2 B_1 B_2 C_1 C_2]{a_1 a_2 b_1 b_2 c_1 c_2}$ from a set of linear equalities that define the marginal distributions, such as for instance
\begin{align}\label{lineareqex}
\p[A_2 B_2]{a_2 b_2} = \sum\nolimits_{a_1 b_1 c_1 c_2}\p[A_1 A_2 B_1 B_2 C_1 C_2]{a_1 a_2 b_1 b_2 c_1 c_2},
\end{align}
together with linear inequalities expressing the nonnegativity of the $\p[A_1 A_2 B_1 B_2 C_1 C_2]{a_1 a_2 b_1 b_2 c_1 c_2}$. 
We can incorporate the ancestral independence $A_1 A_2 \perp_d C_2$ as an additional constraint by defining a variant of the marginal problem wherein the set of linear equations such as \cref{lineareqex} is supplemented by the nonlinear \cref{nonlinearequality2} when one replaces every term therein with the corresponding sum over the $\p[A_1 A_2 B_1 B_2 C_1 C_2]{a_1 a_2 b_1 b_2 c_1 c_2}$. 
We can then proceed with quantifier elimination as we did before, eliminating the unknowns $\p[A_1 A_2 B_1 B_2 C_1 C_2]{a_1 a_2 b_1 b_2 c_1 c_2}$ from the system of equations in order to obtain constraints that involve only joint probabilities on expressible sets.

One can incorporate {\em any} $d$-separation relation in the inflated causal structure in this manner. For instance, if ${\bm X} \perp_d {\bm Y}| {\bm Z}$, then this implies the conditional independence relation of \cref{CIuncondprobs}, which can be incorporated as an additional nonlinear equality constraint when eliminating the unknowns $\p[A_1 A_2 B_1 B_2 C_1 C_2]{a_1 a_2 b_1 b_2 c_1 c_2}$. For instance, in the Spiral inflation of the Triangle scenario~(\cref{fig:Tri222}), the $d$-separation relation $A_1\aindep C_2|A_2 B_2$ implies the conditional independence relation
\begin{align}\label{nonlinearequality}
\forall{a_1 a_2 b_2 c_2}: \p[A_1 A_2 B_2 C_2]{a_1 a_2 b_2 c_2}\p[A_2 B_2]{a_2 b_2}=\p[A_1 A_2 B_2]{a_1 a_2 b_2}\p[A_2 B_2 C_2]{a_2 b_2 c_2}.
\end{align}
However, because $\{ A_1 A_2 B_2 C_2\}$ is not an expressible set, the method of \cref{sec:ineqs} does not take this $d$-separation relation into account. However, it can be incorporated if \cref{nonlinearequality} is included as an additional {\em nonlinear} constraint in the quantifier elimination problem. 

On the one hand, many modern computer algebra systems do have functions capable of tackling nonlinear quantifier elimination symbolically\footnote{For example \textit{Mathematica$^{_{\textit{\tiny\texttrademark}}}$}'s \href[pdfnewwindow]{http://reference.wolfram.com/language/ref/Resolve.html}{\texttt{Resolve}} command, \textit{Redlog}'s \href[pdfnewwindow]{http://www.redlog.eu/documentation/reals/rlqe.php}{\texttt{rlposqe}}, or \textit{Maple$^{_{\textit{\tiny\texttrademark}}}$}'s \href[pdfnewwindow]{http://maplesoft.com/support/help/Maple/view.aspx?path=RegularChains/SemiAlgebraicSetTools/RepresentingQuantifierFreeFormula}{\texttt{RepresentingQuantifierFreeFormula}}.}. 
Currently, however, it is generally not practical to perform nonlinear quantifier elimination on large polynomial systems with many unknowns to be eliminated. It may help to exploit results on the concrete algebraic-geometric structure of these particular systems~\cite{Garcia2}. 

If one is seeking merely to assess the compatibility of a {\em given} distribution with the causal structure, then one can avoid the quantifier elimination problem and simply try and solve an existence problem: after substituting the values that the given distribution prescribes for the outcomes on ai-expressible sets into the polynomial system in terms of the unknown global joint probabilities, one must only determine whether that system has a solution. Most computer algebra systems can resolve such \emph{satisfiability} questions quite easily\footnote{For example \textit{Mathematica$^{_{\textit{\tiny\texttrademark}}}$}'s \href[pdfnewwindow]{http://reference.wolfram.com/language/Experimental/ref/ExistsRealQ.html}{\texttt{Reduce\`{}ExistsRealQ}} function. Specialized satisfiability software such as SMT-LIB's \href[pdfnewwindow]{http://smtlib.cs.uiowa.edu/solvers.shtml}{\texttt{check-sat}} \cite{BarFT-SMTLIB} are particularly apt for this purpose.}.

It is also possible to use a mixed strategy of linear and nonlinear quantifier elimination, such as \citet{ChavesPolynomial} advocates. The explicit results of~\cite{ChavesPolynomial} are directly causal implications of the \emph{original} causal structure, achieved by applying a mixed quantifier elimination strategy. Perhaps further causal compatibility inequalities will be derivable by applying such a mixed quantifier elimination strategy to the inflated causal structure.

\begin{table}[ht]
\centering
\caption{
A comparison of different approaches for deriving constraints on compatibility at the level of the inflated causal structure, which then translate into constraints on compatibility at the level of the original causal structure.
}
\begin{tabular}{ |c|p{5.6cm}lp{4.8cm}|c| } 
\toprule
\parbox[c][][c]{5cm}{\justifying\noindent Type of constraints imposed on the joint distribution over all observed variables in the inflated graph} & \parbox[t][][t]{5cm}{\centering General problem} & $\to$ & Standard algorithm(s) & Difficulty \\
\midrule
\midrule
	\multirow{ 2}{*}{\parbox[b][][c]{4.7cm}{\justifying\noindent Marginal compatibility, i.e. the joint distribution should recover all expressible (or ai-expressible) distributions as marginals (\cref{sec:fulldsep}).}} & \raggedleft\parbox[c][1cm][c]{5.5cm}{\justifying\noindent Facet enumeration of marginal polytope (\cref{step:marginalsproblem})} & $\to$ & see \cref{sec:projalgorithms} & Hard \\\cline{2-5}

	& \raggedleft\parbox[c][1.5cm][c]{5.5cm}{\justifying\noindent Finding possibilistic constraints by identifying hypergraph transversals (\cref{sec:TSEM})} & $\to$ & see~\citet{eiter_dualization_2008} & Very easy \\

\hline
\parbox[c][3cm][c]{4.7cm}{\justifying\noindent Whenever two equivalent-up-to-copy-indices sets of observed variables have ancestral subgraphs which are also equivalent-up-to-copy-indices, then the marginals over said variables must coincide (\cref{sec:copyindexequivalence}). } & \raggedleft\parbox[c][1cm][c]{5.5cm}{\justifying\noindent Marginal problem with additional equality constraints, therefore linear quantifier elimination \linebreak (\cref{sec:coincidingdetails})} & $\to$ & \raggedright Fourier-Motzkin elimination~\cite{fordan1999projection,DantzigEaves,Bastrakov2015,BalasProjectionCone,Jones2008}, \linebreak Equality set projection \cite{JonesThesis2005,jones2004equality} & Hard \\\hline
\parbox[c][2cm][c]{4.7cm}{\justifying\noindent The joint distribution should satisfy all conditional independence relations implied by $d$-separation conditions on the observed variables (\cref{sec:nonlinearsatproblem}). } & \raggedleft\parbox[c][1cm][c]{5.5cm}{\justifying\noindent Real (nonlinear) quantifier elimination} & $\to$ & \raggedright Cylindrical algebraic\linebreak decomposition~\cite{ChavesPolynomial} & Very hard \\

\bottomrule
\end{tabular}
\label{table:difficulties}
\end{table}

\subsection{Implications of the Inflation Technique for Quantum Physics and Generalized Probabilistic Theories}
\label{sec:classicallity}

{This specialized subsection is intended specifically for those readers already somewhat proficient with fundamental concepts in quantum theory. Non-physicists may wish to skip ahead to the \hyperref[sec:conclusions]{conclusions}.}

Recent work has sought to explore quantum generalizations of the notion of a causal model, termed {\em quantum causal models} \cite{leifer2013conditionalstates,pusey2014gdag,BeyondBellII,Chaves2015infoquantum,ried2015quantum,costa2016quantum,allen2016quantum}. 
We here use 
the quantum generalization that is implied by the approach of~\cite{pusey2014gdag} and closely related to the one of~\cite{BeyondBellII}.

The causal structures are still represented by DAGs, supplemented with a distinction between observed and latent nodes. However, the latent nodes are now associated with families of quantum channels and the observed nodes are now associated with families of quantum measurements. Observed nodes are still labelled by random variables, which represent the outcome of the associated measurement. 
One also makes a distinction between edges in the DAG that carry classical information and edges that carry quantum information.\footnote{In many
cases this notion of quantum causal model can also be formulated in a manner that does not require a distinction between two kinds of edges~\cite{BeyondBellII}.} 
 An observed node can have incoming edges of either type: 
those that come from other observed nodes carry classical information, while those that come from latent nodes carry quantum information. Each quantum measurement in the set that is associated to an observed node acts on the collection of quantum systems received by this node (i.e., on the tensor product of the Hilbert spaces associated to the incoming edges). The classical variables that are received by the node act collectively as a control variable, determining which measurement in the set is implemented. Finally, the random variable that is associated to the node encodes the outcome of the measurement. All of the outgoing edges of an observed node are classical and simply broadcast the outcome of the measurement to the children nodes. 
A latent node can also have incoming edges that carry classical variables as well as incoming edges that carry quantum systems. Each quantum channel in the set that is associated to a latent node takes the collection of quantum systems associated to the incoming edges as its quantum input and the collection of quantum systems associated to the outgoing edges as its quantum output (the input and output spaces need not have the same dimension). The classical variables that are received by the node act collectively as a control variable, determining which channel in the set is implemented.

A quantum causal model is still ultimately in the service of explaining joint distributions of observed classical variables. The joint distribution of these variables is the only experimental data with which one can confront a given quantum causal model. The basic problem of causal inference for quantum causal models, therefore, concerns the compatibility of a joint distribution of observed classical variables with a given causal structure, where the model supplementing the causal structure is allowed to be quantum, in the sense defined above. If this happens, we say that the distribution is {\em quantumly compatible} with the causal structure.

One motivation for studying quantum causal models is that they offer a new perspective on an old problem in the field of quantum foundations: that of establishing precisely which of the principles of classical physics must be abandoned in quantum physics. It was noticed by~\citet{fritz2012bell} and~\citet{WoodSpekkens} that Bell's theorem~\cite{bell1966lhvm} states that there are distributions on observed nodes of the Bell causal structure that are quantumly compatible but not classically compatible with it. Moreover, it was shown in~\cite{WoodSpekkens} that these distributions cannot be explained by \emph{any} causal structure while complying with the additional principle that conditional independences should not be fine-tuned, i.e.,~while demanding that any observed conditional independence should be accounted for by a $d$-separation relation in the DAG. These results suggest that quantum theory is perhaps best understood as revising our notions of the nature of unobserved entities, and of how one represents causal dependences thereon and incomplete knowledge thereof, while 
nonetheless {\em preserving} the spirit of causality and the principle of no fine-tuning~\cite{leifer2013conditionalstates,Spekkens2015paradigm,henson2011ontic}.

Another motivation for studying quantum causal models is a practical one. Violations of Bell inequalities have been shown to constitute resources for information processing \cite{NoSigPolytope,scarani2012device,BancalDIApproach}. Hence it seems plausible that if one can find more causal structures for which there exist distributions that are quantumly compatible but not classically so, then this quantum-classical separation may also find applications to information processing. 
For example, it has been shown that in addition to the Bell scenario, such a quantum-classical separation also exists 
in the bilocality scenario \cite{BilocalCorrelations} and the Triangle scenario~\cite{fritz2012bell}, and it is likely that many more causal structures with this property will be found, some with potential applicability to information processing.

So for both foundational and practical reasons, there is good reason to find examples of causal structures that exhibit a quantum-classical separation.
However, this is by no means an easy task.
The set of distributions that are quantumly compatible with a given causal structure is quite hard to separate from the set of distributions that are classically compatible~\cite{pusey2014gdag,fritz2012bell}. For example, both the classical and quantum sets respect the conditional independence relations among observed nodes that are implied by the $d$-separation relations of the DAG~\cite{pusey2014gdag}, and entropic inequalities are only of very limited use~\cite{chaves2012entropic,fritz2012bell}. We hope that the inflation technique will provide better tools for finding such separations.

In addition to quantum generalizations of causal models, one can define generalizations for other operational theories that are neither classical nor quantum~\cite{pusey2014gdag,BeyondBellII}.
Such generalizations are formalized using the framework of {\em generalized probabilistic theories} (GPTs) \cite{Barnum2012GPT,Janotta2014GPT}, which is sufficiently general to describe any operational theory that makes statistical predictions about the outcomes of experiments and passes some basic sanity checks. Some constraints on compatibility can be proven to be \emph{theory-independent} in that they apply not only to classical and quantum causal models, but to any kind of generalized probabilistic causal model~\cite{pusey2014gdag}. For example, the classically-valid conditional independence relations that hold among observed variables in a causal structure are all also valid in the GPT framework.
Another example is the entropic monogamy inequality \cref{eq:monogomyofcorrelations}, which was proven in \cite{pusey2014gdag} to be GPT valid as well. These kinds of constraints are of interest because they clarify what any conceivable theory of physics must satisfy on a given causal structure. 

The essential element in deriving such constraints is to only make reference to the observed nodes, as done in~\cite{pusey2014gdag}. In fact, we now understand the argument of~\cite{pusey2014gdag} to be an instance of the inflation technique. Nonetheless, we have seen that the inflation technique often yields inequalities that hold for the {\em classical} notion of compatibility, while having quantum and GPT violations, such as the Bell inequalities of \cref{example:noPR} of \cref{subsec:witnessingincompat} and \cref{sec:Bellscenarios}. In fact, inflation can be used to derive inequalities with quantum violations for the Triangle scenario as well \cite{TC2016trianglequantum}.

So what distinguishes applications of the inflation technique that yield inequalities for GPT compatibility from those that yield inequalities for classical compatibility? The distinction rests on a structural feature of the inflation:

\begin{definition}
In $G'\in\inflations{G}$, an \tblue{inflationary fan-out} is a latent node that has two or more children that are copy-index equivalent. 
\end{definition}

 The Web and Spiral inflations of the Triangle scenario, depicted in \cref{fig:TriFullDouble} and \cref{fig:Tri222} respectively,
contain one or more inflationary fan-outs, as does the inflation of the Bell causal structure that is depicted in \cref{fig:BellDagCopy1}. On the other hand, the simplest inflation of the Triangle scenario that we consider in this article, the Cut inflation depicted in \cref{fig:simplestinflation}, does not contain any inflationary fan-outs.

Our main observation is that if one uses an inflation without an inflationary fan-out, then the resulting inequalities derived by the inflation technique will all be GPT valid. In other words, one can only hope to detect a GPT-classical separation if one uses an inflation that \emph{has} at least one inflationary fan-out. We now explain the intuition for why this is the case. 
In the classical causal model obtained by inflation, the copy-index-equivalent children of an inflationary fan-out causally depend on their parent node in precisely the same way as their counterparts in the original causal structure do. For example, this dependence may be such that these two children are exact \emph{copies} of the inflationary fan-out node. So when one tries to write down a GPT version of our notion of inflation, one quickly runs into trouble: in quantum theory, the {\em no-broadcasting theorem} shows that such duplication is impossible in a strong sense~\cite{NoCloningQuantum1996}, and an analogous theorem holds for GPTs~\cite{NoCloningGeneral2006}. This is why in the presence of an inflationary fan-out, one cannot expect our inequalities to hold in the quantum or GPT case, which is consistent with the fact that they often do have quantum and GPT violations.

On the other hand, for any inflation that does not contain an inflationary fan-out, the notion of an inflation model generalizes to all GPTs; we sketch how this works for the case of quantum theory. By the definition of inflation, any node in $G'$ has a set of incoming edges equivalent to its counterpart in $G$, while by the assumption that the inflated causal structure does not contain any inflationary fan-outs, any node in $G'$ has either the equivalent set of outgoing edges as its counterpart in $G$, or some pruning of this set. In the former case, 
one associates to this node the same set of quantum channels (if it is a latent node) or measurements (if it is an observed node) that are associated to its counterpart. In the latter case, one simply applies the partial trace operation on the pruned edges (if it is a latent node) or a marginalization on the pruned edges (if it is an observed node). 
That these prescriptions make sense depends crucially on the assumption that $G'$ is an inflation of $G$, so that the ancestries of any node in $G'$ mirrors the ancestry of the corresponding node in $G$ perfectly. Hence for inflations $G'$ without inflationary fan-outs, we have quantum analogues of \cref{mainlemma} and \cref{maincorollary}. The problem of quantum causal inference on $G$ therefore translates into the corresponding problem on $G'$, and any constraint that we can derive on $G'$ translates back to $G$. In particular, our \cref{example:noGHZ,example:polytriangle,example:entropic} also hold for quantum causal inference: perfect correlation is not only classically incompatible with the Triangle scenario, it is quantumly incompatible as well, and the inequalities \cref{eq:polymonogamy,eq:monogomyofcorrelations} have no quantum violations. 

All of these assertions about inflations that do not contain any inflationary fan-outs apply not only to quantum causal models, but to GPT causal models as well, using the definition of the latter provided in \cite{pusey2014gdag}. 

In the remainder of this section, we discuss the relation between the quantum and the GPT case. Since quantum theory is a particular generalized probabilistic theory, quantum compatibility trivially implies GPT compatibility. Through the work of Tsirelson \cite{Tsirelson1980} and \citet{PROriginal}, it is known that the converse is not true: the Bell scenario manifests a GPT-quantum separation. The identification of distributions witnessing this difference, and the derivation of quantum causal compatibility inequalities with GPT violations, has been a focus of much foundational research in recent years. Traditionally, the foundational question has always been: why does quantum theory predict correlations that are {\em stronger} than one would expect classically? But now there is a new question being asked: why does quantum theory only allow correlations that are {\em weaker} than those predicted by other GPTs? There has been some interesting progress in identifying physical principles that can pick out the precise correlations that are exhibited by quantum theory \cite{PopescuReviewNatureComm,ScaraniML,Rohrlich2014,InfoCausArXiv,LONatureComm,LOExploring,EPNBody,barnum2014interference,AlmostQuantum}. Further opportunities for identifying such principles would be useful. This motivates the problem of classifying causal structures into those which have a quantum-classical separation, those which have a GPT-quantum separation and those which have both. Similarly, one can try to classify causal compatibility \emph{inequalities} into those which are GPT-valid, those which are GPT-violable but quantumly valid, and those which are quantum-violable but classically valid. 

The problem of deriving inequalities that are GPT-violable but quantumly valid is particularly interesting. 
\citet{Chaves2015infoquantum} have derived some entropic inequalities that can do so. At present, however, we do not see a way of applying the inflation technique to this problem. 

\section{Conclusions}\label{sec:conclusions}

We have described the \emph{inflation technique} for causal inference in the presence of latent variables.

We have shown how many existing techniques for witnessing incompatibility and for deriving causal compatibility inequalities can be enhanced by the inflation technique, independently of whether these pertain to entropic quantities, correlators or probabilities. The computational difficulty of achieving this enhancement depends on the seed technique. We summarize the computational difficulty of the approaches that we have considered in~\cref{table:difficulties}. A similar table could be drawn for the satisfiability problem, with relative difficulties preserved, but where none of the variants of the problem are computationally hard.

Especially in \cref{sec:ineqs}, we have focused on one particular seed technique: the existence of a joint distribution on all observed nodes together with ancestral independences. We have shown how a complete or partial solution of the marginal problem for the ai-expressible sets of the inflated causal structure can be leveraged to obtain criteria for causal compatibility, both at the level of witnessing particular distributions as incompatible and deriving causal compatibility inequalities. These inequalities are polynomial in the joint probabilities of the observed variables. They are capable of exhibiting the incompatibility of the W-type distribution with the Triangle scenario, while entropic techniques cannot, so that our polynomial inequalities are stronger than entropic inequalities in at least some cases (see \cref{example:noWdist} of \cref{subsec:witnessingincompat}). As far as we can tell, our inequalities are not related to the nonlinear causal compatibility inequalities which have been derived specifically to constrain classical networks \cite{TavakoliStarNetworks,RossetNetworks,TavakoliNoncyclicNetworks}, nor to the nonlinear inequalities which account for interventions to a given causal structure \cite{kang2007polynomialconstraints,steeg2011relaxation}.

We have shown that \emph{some} of the causal compatibility inequalities we derive by the inflation technique are necessary conditions not only for compatibility with a classical causal model, but also for compatibility with a causal model in {\em any} generalized probabilistic theory, which includes quantum causal models as a special case. It would be enlightening to understand the general extent to which our polynomial inequalities for a given causal structure can be violated by a distribution arising in a quantum causal model. A variety of techniques exist for estimating the amount by which a Bell inequality \cite{NPA2008Long,I3322NPA1} is violated in quantum theory, but even finding a quantum violation of one of our \emph{polynomial} inequalities for causal structures other than the Bell scenario presents a new task for which we currently lack a systematic approach. Nevertheless, we know that there exists a difference between classical and quantum also beyond Bell scenarios~\cite[Theorem~2.16]{fritz2012bell}, and we hope that our polynomial inequalities will perform better in probing this separation than entropic inequalities do~\cite{pusey2014gdag,Chaves2015infoquantum}. 

We have shown that the inflation technique can also be used to derive causal compatibility inequalities that hold for arbitrary generalized probabilistic theories, a significant generalization of the results of \cite{pusey2014gdag}. Such inequalities are also very significant insofar as they constitute a restriction on the sorts of statistical correlations that could arise in a given causal scenario even if quantum theory is superseded by some alternative physical theory. As long as the successor theory falls within the framework of generalized probabilistic theories, the restriction will hold. 

Finally, an interesting question is whether it might be possible to modify our methods somehow to derive causal compatibility inequalities that hold for quantum theory and are violated by some GPT. Since the initial drafting of this manuscript, such a modification has been identified~\cite{Wolfe2018quantuminflation}.

A single causal structure has an unlimited number of potential inflations. Selecting a good inflation from which strong polynomial inequalities can be derived is an interesting challenge. To this end, it would be desirable to understand how particular features of the original causal structure are exposed when different nodes in the causal structure are duplicated. By isolating which features are exposed in each inflation, we could conceivably quantify the utility for causal inference of each inflation. In so doing, we might find that inflations beyond a certain level of variable duplication need not be considered. The multiplicity beyond which further inflation is irrelevant may be related to the maximum degree of those polynomials which tightly characterize a causal scenario. Presently, however, it is not clear how to upper bound either number, 
though a finite upper bound on the maximum degree of the polynomials follows from the semialgebraicity of the compatible distributions, per Ref.~\cite{rosset2016finite}.

Causal compatibility inequalities are, by definition, merely {\em necessary} conditions for compatibility. Depending on what kind of causal inference methods one uses at the level of an inflated causal structure $G'$, one may or may not obtain sufficient conditions. An interesting question is: if one only uses the existence of a joint distribution and ancestral independences at the level of $G'$, then does one obtain sufficient conditions as $G'$ varies? In other words: if a given distribution is such that for every inflation $G'$, the marginal problem of \cref{sec:ineqs} is solvable, then is the distribution compatible with the original causal structure? This occurs for the Bell scenario, where it is enough to consider only one particular inflation (\cref{sec:Bellscenarios}). 

Significantly, since the initial drafting of this manuscript, Ref.~\cite{Navascues2017completelysolves} has proven that the inflation technique indeed gives necessary and sufficient conditions for causal compatibility: any incompatible distribution is witnessed as incompatible by a suitably large inflation. Ref.~\cite{Navascues2017completelysolves} also provides other interesting results, such as a prescription for how to generate all relevant inflations, as well as an explicit demonstration of the inflation technique as applied to Pearl's instrumental scenario. 


\begin{acknowledgments}
E.W.~would like to thank Rafael Chaves, Miguel Navascues, and T.C. Fraser for suggestions which have improved this manuscript. T.F.~would like to thank Nihat Ay and Guido Mont\'ufar for discussion and references. Part of this research was conducted while T.F.~was with the Max Planck Institute for Mathematics in the Sciences. This project/publication was made possible in part through the support of grant \href[pdfnewwindow]{https://www.templeton.org/grant/quantum-causal-structures}{\#69609} from the John Templeton Foundation. The opinions expressed in this publication are those of the authors and do not necessarily reflect the views of the John Templeton Foundation. This research was supported in part by Perimeter Institute for Theoretical Physics. Research at Perimeter Institute is supported in part by the Government of Canada through the Department of Innovation, Science and Economic Development Canada and by the Province of Ontario through the Ministry of Economic Development, Job Creation and Trade. 
\end{acknowledgments}

\appendix
\numberwithin{equation}{section}
\let\stdsection\section
\renewcommand{\section}{\clearpage\stdsection} 

\section{Algorithms for Solving the Marginal Constraint Problem}\label{sec:projalgorithms}

By solving the marginal constraint problem, what we mean is to determine all the facets of the marginal polytope for a given marginal scenario. Since the vertices of this polytope are precisely the deterministic assignments of values to all variables, which are easy to enumerate, solving the marginal constraint problem is an instance of a \tblue{facet enumeration problem}: given the vertices of a convex polytope, determine its facets. This is a well-studied problem in combinatorial optimization for which a variety of algorithms are available~\cite{avis_convexhull_2015}. 

A generic facet enumeration problem takes a matrix $\bm{V}\in\mathbb{R}^{d\times n}$, which lists the vertices as its columns, and asks for an inequality description of the set of vectors $\bm{b}\in\mathbb{R}^d$ that can be written as a convex combination of the vertices using weights $\bm{x}\in\mathbb{R}^n$ that are nonnegative and normalized,
\begin{align}
	\label{projsimplex}
	\left\{\: \bm{b}\in\mathbb{R}^d \quad\bigg|\quad \exists \bm{x}\in\mathbb{R}^n:\; \bm{b} = \bm{V}\bm{x} ,\;\; \bm{x}\geq \bm{0},\;\; {{\sum_i}{x_i}}=1 \:\right\}.
\end{align}
To solve the marginal problem one uses the marginal description matrix introduced in \cref{step:marginalsproblem} as the input to the facet enumeration algorithm, i.e. $\bm{V}=\bm{M}$, see \cref{eq:marginalproblemgeneric}.

The oldest-known method for facet enumeration relies on \tblue{linear quantifier elimination} in the form of Fourier-Motzkin (FM) elimination~\cite{fordan1999projection,DantzigEaves}. This refers to the fact that one starts with the system $\bm{b}= \bm{V}\bm{x}$, $\bm{x}\geq \bm{0}$ and ${{\sum_i}{x_i}}=1$, which is the half-space representation of a convex polytope (a simplex), and then one needs to project onto $\bm{b}$-space by \emph{eliminating} the variables $\bm{x}$ to which the existential \emph{quantifier} $\exists \bm{x}$ refers. The Fourier-Motzkin algorithm is a particular method for performing this quantifier elimination one variable at a time; when applied to~\cref{projsimplex}, it is equivalent to the \emph{double description method}~\cite{DantzigEaves,Fukuda1996}. Linear quantifier elimination routines are available in many software tools\footnote{For example \textit{MATLAB$^{_{\textit{\tiny\texttrademark}}}$}'s \href[pdfnewwindow]{http://people.ee.ethz.ch/~mpt/2/docs/refguide/mpt/@polytope/projection.html}{\texttt{MPT2}}/\href[pdfnewwindow]{http://ellipsoids.googlecode.com/svn-history/r2740/branches/issue_119_vrozova/tbxmanager/toolboxes/mpt/3.0.14/all/mpt3-3_0_14/mpt/modules/geometry/sets/@Polyhedron/projection.m}{\texttt{MPT3}}, \textit{Maxima}'s \href[pdfnewwindow]{http://maxima.sourceforge.net/docs/manual/de/maxima_75.html}{\texttt{fourier\_elim}}, \textit{lrs}'s \href[pdfnewwindow]{http://cgm.cs.mcgill.ca/~avis/C/lrslib/USERGUIDE.html\#fourier}{\texttt{fourier}}, or \textit{Maple$^{_{\textit{\tiny\texttrademark}}}$}'s (v17+) \href[pdfnewwindow]{http://www.maplesoft.com/support/help/maple/view.aspx?path=RegularChains/SemiAlgebraicSetTools/LinearSolve}{\texttt{LinearSolve}} and \href[pdfnewwindow]{http://www.maplesoft.com/support/help/Maple/view.aspx?path=RegularChains/SemiAlgebraicSetTools/Projection}{\texttt{Projection}}. The efficiency of most of these software tools, however, drops off markedly when the dimension of the final projection is much smaller than the initial space of the inequalities. Fast facet enumeration aided by Chernikov rules \cite{Shapot2012,Bastrakov2015} is implemented in \href[pdfnewwindow]{https://www.inf.ethz.ch/personal/fukudak/cdd_home/}{\textit{cdd}},
\href[pdfnewwindow]{http://comopt.ifi.uni-heidelberg.de/software/PORTA/}{\textit{PORTA}}, \href[pdfnewwindow]{http://sbastrakov.github.io/qskeleton/}{\textit{qskeleton}}
, and \href[pdfnewwindow]{http://www.uic.unn.ru/~zny/skeleton/}{\textit{skeleton}}. In the authors experience \href[pdfnewwindow]{http://www.uic.unn.ru/~zny/skeleton/}{\textit{skeleton}} seemed to be the most efficient. Additionally, the package \href[pdfnewwindow]{https://polymake.org/doku.php/researchdata/polymakeilp}{\textit{polymake}} offers multiple algorithms as options for computing convex hulls.}. The authors found it convenient to custom-code a linear quantifier elimination routine in \textit{Mathematica$^{_{\textit{\tiny\texttrademark}}}$}.

Other algorithms for facet enumeration that are not based on linear quantifier elimination include the following. \emph{Lexicographic reverse search} (LRS)~\cite{Avis2000lrs} explores the entire polytope by repeatedly pivoting from one facet to an adjacent one, and is implemented in~\href[pdfnewwindow]{http://cgm.cs.mcgill.ca/~avis/C/lrslib/USERGUIDE.html#Installation\%20Section}{\texttt{lrs}}. Equality Set Projection (ESP)~\cite{jones2004equality,JonesThesis2005} is also based on pivoting from facet to facet, though its implementation is less stable\footnote{ESP \cite{jones2004equality,JonesThesis2005,Jones2008} is supported by \href[pdfnewwindow]{http://people.ee.ethz.ch/~mpt/2/docs/refguide/mpt/@polytope/projection.html}{\texttt{MPT2}} but not \href[pdfnewwindow]{http://people.ee.ethz.ch/~mpt/3/}{\texttt{MPT3}}, and by the (undocumented) option of \href[pdfnewwindow]{https://github.com/tulip-control/polytope/blob/master/polytope/polytope.py\#L1412}{projection} in the \href[pdfnewwindow]{https://pypi.python.org/pypi/polytope}{\textit{polytope}} (v0.1.2 2016-07-13) python module.}. These algorithms could be interesting to use in practice, since each pivoting step churns out a new facet; by contrast, Fourier-Motzkin type algorithms only generate the entire list of facets at once, after all the quantifiers have been eliminated one by one, see Ref.~\cite{Chaves2018FMprojection} for a recent comparative review.

It may also be possible to exploit special features of marginal polytopes in order to facilitate their facet enumeration, such as their high degree of symmetry: permuting the outcomes of each variable maps the polytope to itself, which already generates a sizeable symmetry group, and oftentimes there are additional symmetries given by permuting some of the variables. This simplifies the problem of facet enumeration~\cite{bremner_symmetries_2009,Schurmann2013}, and it may be interesting to apply dedicated software\footnote{Such as \href[pdfnewwindow]{http://comopt.ifi.uni-heidelberg.de/software/PANDA/}{\textit{PANDA}}, \href[pdfnewwindow]{http://mathieudutour.altervista.org/Polyhedral/}{\textit{Polyhedral}}, or \href[pdfnewwindow]{http://www.math.uni-rostock.de/~rehn/software/sympol.html}{\textit{SymPol}}. The authors found \textit{SymPol} to be rather effective for some small test problems, using the options ``\texttt{./sympol -a --cdd}".} to the facet enumeration problem of marginal polytopes~\cite{Kaibel2010,rehn_tools_2012,panda_2015}.

\section{Explicit Marginal Description Matrix of the Cut Inflation with Binary Observed Variables}\label{sec:explicitmatrix}

The three maximal ai-expressible sets of the Cut inflation (\cref{fig:simpleinflation} on Pg. \pageref{fig:simpleinflation}) are $\{A_2 B_1\}$, $\{B_1 C_1\}$, and $\{A_2 C_1\}$. Taking the variables to be binary, each ai-expressible set corresponds to $2^2=4$ equations pertinent to the marginal problem. The three \emph{sets} of equations which relate the marginal probabilities to a posited joint distribution are given by
\begin{align}\label[eqs]{eq:marginalequalities222}
\begin{split}
&\forall{a_2 b_1}:\;\p[A_2 B_1]{a_2 b_1} = \sum\nolimits_{c_1}\p[A_2 B_1 C_1]{a_2 b_1 c_1 },\\
&\forall{b_1 c_1}:\;\p[B_1 C_1]{b_1 c_1} = \sum\nolimits_{a_2}\p[A_2 B_1 C_1]{a_2 b_1 c_1 },\\
&\forall{a_2 c_1}:\;\p[A_2 C_1]{a_2 c_1} = \sum\nolimits_{b_1}\p[A_2 B_1 C_1]{a_2 b_1 c_1 }.
\end{split}
\end{align}
As we noted in the main text, such conditions can be expressed in terms of a single matrix equality, $\bm{M} \bm{v} = \bm{b}$ where $\bm{v}$ is the \tblue{joint distribution vector}, $\bm{b}$ is the \tblue{marginal distribution vector} and $\bm{M}$ is the \tblue{marginal description matrix}. In the Cut inflation example, the joint distribution vector $\bm{v}$ has 8 elements, whereas the marginal distribution vector $\bm{b}$ has 12, i.e.
\begin{align}
\bm{v}=\begin{pmatrix}
 P_{A_2 B_1 C_1}(0 0 0) \\
 P_{A_2 B_1 C_1}(0 0 1) \\
 P_{A_2 B_1 C_1}(0 1 0) \\
 P_{A_2 B_1 C_1}(0 1 1) \\
 P_{A_2 B_1 C_1}(1 0 0) \\
 P_{A_2 B_1 C_1}(1 0 1) \\
 P_{A_2 B_1 C_1}(1 1 0) \\
 P_{A_2 B_1 C_1}(1 1 1)
\end{pmatrix}
,\qquad
\bm{b}=
\begin{pmatrix}
 p_{A_2 B_1}(0 0) \\
 p_{A_2 B_1}(0 1) \\
 p_{A_2 B_1}(1 0) \\
 p_{A_2 B_1}(1 1) \\
 p_{A_2 C_1}(0 0) \\
 p_{A_2 C_1}(0 1) \\
 p_{A_2 C_1}(1 0) \\
 p_{A_2 C_1}(1 1) \\
 p_{B_1 C_1}(0 0) \\
 p_{B_1 C_1}(0 1) \\
 p_{B_1 C_1}(1 0) \\
 p_{B_1 C_1}(1 1)
\end{pmatrix}=
\begin{pmatrix}
 P_{A_2 B_1 C_1}(0 0 \_) \\
 P_{A_2 B_1 C_1}(0 1 \_) \\
 P_{A_2 B_1 C_1}(1 0 \_) \\
 P_{A_2 B_1 C_1}(1 1 \_) \\
 P_{A_2 B_1 C_1}(0 \_ 0) \\
 P_{A_2 B_1 C_1}(0 \_ 1) \\
 P_{A_2 B_1 C_1}(1 \_ 0) \\
 P_{A_2 B_1 C_1}(1 \_ 1) \\
 P_{A_2 B_1 C_1}(\_ 0 0) \\
 P_{A_2 B_1 C_1}(\_ 0 1) \\
 P_{A_2 B_1 C_1}(\_ 1 0) \\
 P_{A_2 B_1 C_1}(\_ 1 1) 
\end{pmatrix},
\end{align}
and hence the marginal description matrix $\bm{M}$ is a $12\times 8$ matrix of zeroes and ones, i.e.
\begin{align}
\bm{M}=\begin{pmatrix}
 \bm{1} & \bm{1} & {\scriptscriptstyle ^0} & {\scriptscriptstyle ^0} & {\scriptscriptstyle ^0} & {\scriptscriptstyle ^0} & {\scriptscriptstyle ^0} & {\scriptscriptstyle ^0} \\
 {\scriptscriptstyle ^0} & {\scriptscriptstyle ^0} & \bm{1} & \bm{1} & {\scriptscriptstyle ^0} & {\scriptscriptstyle ^0} & {\scriptscriptstyle ^0} & {\scriptscriptstyle ^0} \\
 {\scriptscriptstyle ^0} & {\scriptscriptstyle ^0} & {\scriptscriptstyle ^0} & {\scriptscriptstyle ^0} & \bm{1} & \bm{1} & {\scriptscriptstyle ^0} & {\scriptscriptstyle ^0} \\
 {\scriptscriptstyle ^0} & {\scriptscriptstyle ^0} & {\scriptscriptstyle ^0} & {\scriptscriptstyle ^0} & {\scriptscriptstyle ^0} & {\scriptscriptstyle ^0} & \bm{1} & \bm{1} \\
 \bm{1} & {\scriptscriptstyle ^0} & \bm{1} & {\scriptscriptstyle ^0} & {\scriptscriptstyle ^0} & {\scriptscriptstyle ^0} & {\scriptscriptstyle ^0} & {\scriptscriptstyle ^0} \\
 {\scriptscriptstyle ^0} & \bm{1} & {\scriptscriptstyle ^0} & \bm{1} & {\scriptscriptstyle ^0} & {\scriptscriptstyle ^0} & {\scriptscriptstyle ^0} & {\scriptscriptstyle ^0} \\
 {\scriptscriptstyle ^0} & {\scriptscriptstyle ^0} & {\scriptscriptstyle ^0} & {\scriptscriptstyle ^0} & \bm{1} & {\scriptscriptstyle ^0} & \bm{1} & {\scriptscriptstyle ^0} \\
 {\scriptscriptstyle ^0} & {\scriptscriptstyle ^0} & {\scriptscriptstyle ^0} & {\scriptscriptstyle ^0} & {\scriptscriptstyle ^0} & \bm{1} & {\scriptscriptstyle ^0} & \bm{1} \\
 \bm{1} & {\scriptscriptstyle ^0} & {\scriptscriptstyle ^0} & {\scriptscriptstyle ^0} & \bm{1} & {\scriptscriptstyle ^0} & {\scriptscriptstyle ^0} & {\scriptscriptstyle ^0} \\
 {\scriptscriptstyle ^0} & \bm{1} & {\scriptscriptstyle ^0} & {\scriptscriptstyle ^0} & {\scriptscriptstyle ^0} & \bm{1} & {\scriptscriptstyle ^0} & {\scriptscriptstyle ^0} \\
 {\scriptscriptstyle ^0} & {\scriptscriptstyle ^0} & \bm{1} & {\scriptscriptstyle ^0} & {\scriptscriptstyle ^0} & {\scriptscriptstyle ^0} & \bm{1} & {\scriptscriptstyle ^0} \\
 {\scriptscriptstyle ^0} & {\scriptscriptstyle ^0} & {\scriptscriptstyle ^0} & \bm{1} & {\scriptscriptstyle ^0} & {\scriptscriptstyle ^0} & {\scriptscriptstyle ^0} & \bm{1} 
\end{pmatrix}
\end{align}
such that $\bm{M}\bm{v}=\bm{b}$ per \cref{eq:marginalproblemgeneric}.

\section{Constraints on Marginal Distributions from Copy-Index Equivalence Relations}\label{sec:coincidingdetails}

In \cref{sec:copyindexequivalence}, we noted that every copy of a variable in an inflation model has the same probabilistic dependence on its parents as every other copy. It followed that for certain pairs of marginal contexts, the marginal distributions in any inflation model are necessarily equal. We now describe not only how to identify all such pairs of contexts, but also how to identify weak pairs, who's corresponding symmetry imposition cannot help strengthen the final constraints.

Given $\bm{X},\bm{Y}\subseteq\nodes{G'}$ in an inflated causal structure $G'$, let us say that a map $\varphi:\bm{X}\to\bm{Y}$ is a \tblue{copy isomorphism} if it is a graph isomorphism\footnote{A graph isomorphism is a bijective map between the nodes of one graph and the nodes of another, such that both the map and its inverse take edges to edges.} between $\subgraph{\bm{X}}$ and $\subgraph{\bm{Y}}$ such that $\varphi(X)\sim X$ for all $X\in\bm{X}$, meaning that $\varphi$ maps every node $X\in\bm{X}$ to a node $Y{=}\varphi(X)\in\bm{Y}$ such that $Y$ is equivalent to $X$ under dropping the copy-index.

Furthermore, we say that a copy isomorphism $\varphi : \bm{X}\to\bm{Y}$ is an \tblue{inflationary isomorphism} whenever it can be extended to a copy isomorphism on the \emph{ancestral} subgraphs, $\Phi : \An{\bm{X}}\to\An{\bm{Y}}$. 
A copy isomorphism $\Phi: \An{\bm{X}}\to\An{\bm{Y}}$ defines an inflationary isomorphism $\varphi:\bm{X}\to\bm{Y}$ if and only if $\Phi(\bm{X}) = \bm{Y}$.
So in practice, one can either start with $\varphi : \bm{X}\to\bm{Y}$ and try to extend it to $\Phi : \An{\bm{X}}\to\An{\bm{Y}}$, or start with such a $\Phi$ and see whether it maps $\bm{X}$ to $\bm{Y}$ and thereby restricts to a $\varphi$.

For given observed $\bm{V}_1$ and $\bm{V}_2$, a sufficient condition for equality of their marginal distributions in an inflation model is that there exists an inflationary isomorphism between them. Because $\bm{V}_1$ and $\bm{V}_2$ might themselves contain several variables that are copy-index equivalent (recall the examples of \cref{sec:copyindexequivalence}), equating the distribution $P_{\bm{V}_1}$ with the distribution $P_{\bm{V}_2}$ in an unambiguous fashion requires one to specify a correspondence between the variables that make up $\bm{V}_1$ and those that make up $\bm{V}_2$. This is exactly the data provided by the inflationary isomorphism $\varphi$. This result is summarized in the following lemma.

\begin{lemma}
Let $G'$ be an inflation of $G$, and let $\bm{V}_1,\bm{V}_2\subseteq\obsnodes{G'}$. Then every inflationary isomorphism $\varphi:\bm{V}_1\to\bm{V}_2$ induces an equality $P_{\bm{V}_1} = P_{\bm{V}_2}$ for inflation models, where the variables in $\bm{V}_1$ are identified with those in $\bm{V}_2$ according to $\varphi$.
\label[lemma]{lem:coincide}
\end{lemma}

This applies in particular when $\bm{V}_1 = \bm{V}_2$, in which case the statement is that the distribution $P_{\bm{V_1}}$ is \emph{invariant under permuting the variables} according to $\varphi$.

\cref{lem:coincide} is best illustrated by returning to our example from \cref{sec:copyindexequivalence} 
which considered the Spiral inflation of~\cref{fig:Tri222} and the pair of contexts $\bm{V}_1 = \{ A_1 A_2 B_1\}$ and $\bm{V}_2 =\{ A_1 A_2 B_2\}$. The map 
\begin{align}\label{copyisomorph}
	\varphi \: : \: A_1 \mapsto A_1,\qquad A_2\mapsto A_2,\qquad B_1\mapsto B_2
\end{align}
is a copy isomorphism between $\bm{V}_1$ and $\bm{V}_2$
because it trivially implements a graph isomorphism (both subgraphs are edgeless), and it maps each variable in $\bm{V}_1$ to a variable in $\bm{V}_2$ that is copy-index equivalent. There is a unique choice to extend $\varphi$ to a copy isomorphism $\Phi:\An{\bm{V}_1}\to\An{\bm{V}_2}$, namely, by extending~\cref{copyisomorph} to the ancestors via
\begin{align}
\Phi \: : \: X_1\mapsto X_1,\qquad Y_1\mapsto Y_1, \qquad Y_2 \mapsto Y_2, \qquad Z_1 \mapsto Z_2,
\end{align}
which is again a copy isomorphism. 
Therefore $\varphi$ is indeed an inflationary isomorphism. From \cref{lem:coincide}, we then conclude that any inflation model satisfies $P_{A_1 A_2 B_1} = P_{A_1 A_2 B_2}$.

Similarly, the map 
\begin{align}\label{copyisomorph2}
	\varphi' \: : \: A_1 \mapsto A_2,\qquad A_2\mapsto A_1,\qquad B_1\mapsto B_2
\end{align}
is \emph{also} easily verified to be a copy isomorphism between $\subgraph{\bm{V}_1}$ and $\subgraph{\bm{V}_2}$, and there is again a unique choice to extend $\varphi'$ to a copy isomorphism $\Phi':\ansubgraph{\bm{V}_1}\to\ansubgraph{\bm{V}_2}$, by extending~\cref{copyisomorph2} with
\begin{align}
\Phi' \: : \: X_1\mapsto X_1,\qquad Y_1\mapsto Y_2, \qquad Y_2 \mapsto Y_1, \qquad Z_1 \mapsto Z_2,
\end{align}
so that $\varphi'$ too is verified to be an inflationary isomorphism. From \cref{lem:coincide}, we then conclude that every inflation model also satisfies $P_{A_1 A_2 B_1} = P_{A_2 A_1 B_2}$. (And this in turn implies that for the context $\{A_1 A_2\}$, the marginal distribution satisfies the permutation invariance $P_{A_1 A_2} = P_{A_2 A_1}$.)

\begin{figure}[b]
 \centering
 \begin{minipage}[t]{0.22\linewidth} \centering
 \includegraphics[scale=1]{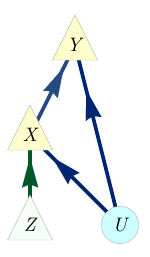}
 \caption{The instrumental scenario of \citet{pearl1995instrumental}.}
 \label{fig:ISorigDAG}
 \end{minipage}\hfill
 \begin{minipage}[t]{0.4\linewidth} \centering
 \includegraphics[scale=1]{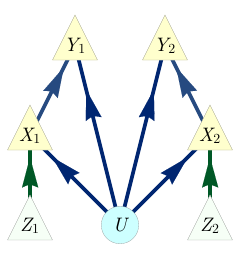}
 \caption{An inflation of the instrumental scenario which illustrates why coinciding ancestral subgraphs doesn't necessarily imply coinciding marginal distributions.}
 \label{fig:IScopyDAG}
 \end{minipage}\hfill 
 \begin{minipage}[t]{0.28\linewidth} \centering
 \includegraphics[scale=1]{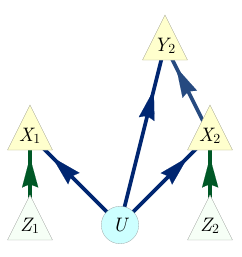}
 \caption{The ancestral subgraph of \cref{fig:IScopyDAG} for either $\{X_1 Y_2 Z_1\}$ or $\{X_1 Y_2 Z_2\}$.}
 \label{fig:ancestralsubgraphnotenough}
 \end{minipage}
\end{figure}

In order to avoid any possibility of confusion, we emphasize that it is not a plain copy isomorphism between the subgraphs of $\bm{V}_1$ and $\bm{V}_2$ themselves which results in coinciding marginal distributions, nor a copy isomorphism between the ancestral subgraphs of $\bm{V}_1$ and $\bm{V}_2$. Rather, it is an inflationary isomorphism between the subgraphs, i.e., a copy isomorphism between the ancestral subgraphs that restricts to a copy isomorphism between the subgraphs. To see why a copy isomorphism between ancestral subgraphs {\em by itself} may not be sufficient for deriving equality of marginal distributions, we offer the following example. Take as the original causal structure the instrumental scenario of \citet{pearl1995instrumental}, and consider the inflation depicted in \cref{fig:IScopyDAG}. Consider the pair of contexts $\bm{V}_1 = \{ X_1 Y_2 Z_1\}$ and $\bm{V}_2 = \{ X_1 Y_2 Z_2\}$ on the inflated causal structure. Since $\subgraph{\bm{V}_1}$ and $\subgraph{\bm{V}_2}$ are not isomorphic, there is no copy isomorphism between the two. On the other hand, 
the ancestral subgraphs are both given by the causal structure of~\cref{fig:ancestralsubgraphnotenough}, so that the identity map is a copy isomorphism between $\ansubgraph{X_1 Y_2 Z_1}$ and $\ansubgraph{X_1 Y_2 Z_2}$.

One can try to make use of \cref{lem:coincide} when deriving polynomial inequalities with inflation via solving the marginal problem, by imposing the resulting equations of the form $P_{\bm{V}_1} = P_{\bm{V}_2}$ as additional constraints, one constraint for each inflationary isomorphism $\varphi : \bm{V}_1\to\bm{V}_2$ between sets of observed nodes. This is advantageous to speed up to the linear quantifier elimination, since one can solve each of the resulting equations for one of the unknown joint probabilities and thereby eliminate that probability directly without Fourier-Motzkin elimination. Moreover, one could hope that these additional equations also result in tighter constraints on the marginal problem, which would in turn yield tighter causal compatibility inequalities. Our computations have so far not revealed any example of such a tightening.

In some cases, this lack of impact can be explained as follows.
Suppose that $\varphi:\bm{V}_1\to\bm{V}_2$ is an inflationary isomorphism such that $\varphi$ can be extended to a copy automorphism 
$\Phi':G'\to G'$, which maps the entirety of the inflated causal structure onto itself. An inflationary isomorphism can always be extended to some copy isomorphism between the ancestral subgraphs $\Phi:\ansubgraph{\bm{V}_1}\to\ansubgraph{\bm{V}_2}$ by definition, but not every inflationary isomorphism can also be extended to a full copy automorphism of $G'$.
In those cases where $\varphi$ \emph{can} be extended to a copy automorphism, 
the \emph{irrelevance} of the additional constraint $P_{\bm{V}_1} = P_{\bm{V}_2}$ to the marginal problem for inflation models can be explained by the following argument. 

Suppose that some joint distribution $P_{\obsnodes{G'}}$ solves the unconstrained marginal problem, i.e.,~without requiring $P_{\bm{V}_1} = P_{\bm{V}_2}$. Now apply the automorphism $\Phi'$ to the variables in $P_{\obsnodes{G'}}$, switching the variables around, to generate a new distribution $P'_{\obsnodes{G'}}:=P_{\Phi'(\obsnodes{G'})}$. Because the set of marginal distributions that arise from inflation models is invariant under this switching of variables, we conclude that $P'$ is also a solution to the unconstrained marginal problem. Taking the uniform mixture of $P$ and $P'$ is therefore still a solution of the unconstrained marginal problem. But this uniform mixture also satisfies the supplementary constraint $P_{\bm{V}_1} = P_{\bm{V}_2}$. Hence the supplementary constraint is satisfiable whenever the unconstrained marginal problem is solvable, which makes adding the constraint irrelevant.

This argument does not apply when the inflationary isomorphism $\varphi:\bm{V}_1\to\bm{V}_2$ cannot be extended to a copy automorphism of the entire inflated causal structure. It also does not apply if one uses $d$-separation conditions beyond ancestral independence on the inflated causal structure as additional constraints (\cref{sec:fulldsep}), because in this case the set of compatible distributions is not necessarily convex. In either of these cases, it is unclear whether or not constraints arising from copy-index equivalence could yield tighter inequalities. 

\section{Using the Inflation Technique to Certify a Causal Structure as ``Interesting"\label{sec:interestingproof}}

By considering all possible $d$-separation relations on the observed nodes of a causal structure, one can infer the set of all conditional independence (CI) relations that must hold in any distribution compatible with it. Due to the presence of latent variables, satisfying these CI relations 
is generally not sufficient for compatibility. Henson, Lal and Pusey (HLP) \cite{pusey2014gdag} introduced the term \tblue{interesting} for those causal structures for which this happens, and derived a partial classification of causal structures into interesting and non-interesting ones by
finding necessary criteria for a causal structure to be interesting, and they also conjectured their criteria to be sufficient. As evidence in favour of this conjecture, they enumerated all possible isomorphism classes of causal structure with up to six nodes satisfying their criteria, which resulted in only 21 equivalence classes of potentially interesting causal structures. Of those 21, they further proved that 18 were indeed interesting by writing down explicit distributions which are incompatible despite satisfying the observed CI relations. Incompatibility was certified by means of entropic inequalities. 

That left three classes of causal structures as \emph{potentially} interesting. For each of these, HLP
 derived both: (i) the set of Shannon-type entropic inequalities that take into account the CI relations among the observed variables, and (ii) the set of Shannon-type entropic inequalities that also take into account CI relations among latent variables.
Finding the second set to be larger than the first constitutes evidence that the causal structure is interesting.
The evidence is not conclusive, however, because the Shannon-type inequalities that are included in the second set but not the first might be non-Shannon-type inequalities that merely follow from the CI relations among the observed variables~\cite{pusey2014gdag}.

One way to close this loophole would be to show that the novel Shannon-type inequalities imply constraints beyond some inner approximation to the genuine entropic cone corresponding to the CI relations among observed variables, perhaps along the lines of~\cite{weilenmann2016entropic}. Another is to use causal compatibility inequalities beyond entropic inequalities to identify some CI-respecting but incompatible distributions. \citet{pianaar2016interesting} accomplished precisely this
by considering the different values that an observed root variable may take. In the following, we demonstrate how the inflation technique can be used for the same purpose. 

\begin{figure}[t]
\centering
\begin{minipage}[t]{0.32\linewidth}
\centering
\includegraphics[scale=1]{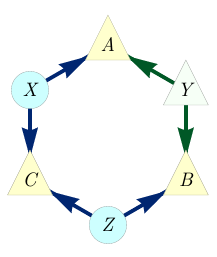}
\caption{Causal structure \#15 in \cite{pusey2014gdag}. The $d$-separation relations are $C\aindep Y$ and $A \aindep B\,|\,Y$.}\label{fig:GDAG15}
\end{minipage}
\hfill
\begin{minipage}[t]{0.32\linewidth}
\centering
\includegraphics[scale=1]{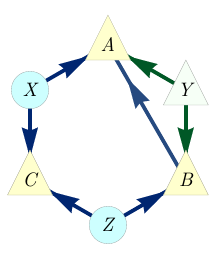}
\caption{Causal structure \#16 in \cite{pusey2014gdag}. The only $d$-separation relation is $C\aindep Y$.}\label{fig:GDAG16}
\end{minipage}
\hfill
\begin{minipage}[t]{0.32\linewidth}
\centering
\includegraphics[scale=1]{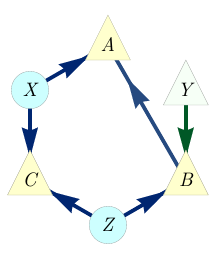}
\caption{Causal structure \#20 in \cite{pusey2014gdag}. The $d$-separation relations are $C\aindep Y$ and $A \aindep Y\,|\,B$.}\label{fig:GDAG20}
\end{minipage}
\hfill
\end{figure}

\subsection{Certifying that Henson-Lal-Pusey's Causal Structure \#16 is ``Interesting"} 
\label{example:Pienaar}

\begin{figure}[htb]
\centering
\centering
\begin{minipage}[t]{0.5\linewidth}
\centering
\includegraphics[scale=1]{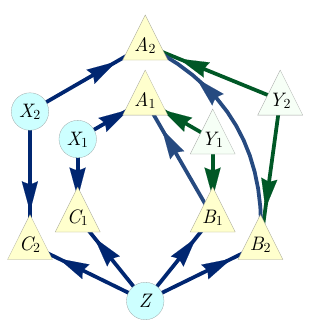}
\caption{The Russian dolls inflation of \cref{fig:GDAG16}.}\label{fig:Inflated16new}
\end{minipage}
\end{figure}

\citet{pianaar2016interesting} identified a distribution which satisfies the only CI relation that must hold among the observed variables in HLP's causal structure \#16 (\cref{fig:GDAG16} here), namely, $C\indep Y$, but which is nonetheless incompatible with it:
\begin{align}\label{eq:pienaardistro}
\! P^{\text{Pienaar}}_{A B C Y}:=\frac{[0000]+[0110]+[0001]+[1011]}{4},\quad\text{i.e.,}\quad P^{\text{Pienaar}}_{A B C Y}(a b c y)=\begin{cases}\tfrac{1}{4}&\text{if } y\cdot c = a \text{ and } (y \oplus 1)\cdot c = b, \\ 0&\text{otherwise}.\end{cases}
\end{align}
It is useful to compute the conditional on $Y$,
\begin{align}\label{eq:pienaardistroconditiona}
 P^{\text{Pienaar}}_{A B C|Y}(\cdot \cdot \cdot | y)=\begin{cases}\frac{1}{2}\parenths{[000]+[011]}&\text{if } y\eql 0, \\ \frac{1}{2}\parenths{[000]+[101]}&\text{if } y\eql 1.\end{cases}
\end{align}
This makes it evident that the distribution can be described as follows: if $Y=0$, then $A=0$ while $B$ and $C$ are uniformly random and perfectly correlated, while if $Y=1$, then $B=0$ and $A$ and $C$ are uniformly random and perfectly correlated.

Here, we will establish the incompatibility of Pienaar's distribution with HLP's causal structure \#16 (\cref{fig:GDAG16} here) using the inflation technique. To do so, we use the inflation
depicted in \cref{fig:Inflated16new}, which we term the {\em Russian dolls} inflation. 
We will make use of the fact that $\{ A_1 C_1 Y_1\}$, $\{ B_2 C_2 Y_2\}$ and $\{B_2 C_1 Y_2 \}$ are injectable sets, together with the fact that $\{ A_1 C_2 Y_1\}$ is an expressible set. 

We begin by demonstrating how the $d$-separation relations in the Russian dolls inflation imply that $\{ A_1 C_2 Y_1\}$ is expressible. 
First, we note that the set $\{ A_1 B_1 C_2 Y_1\}$ is expressible because the $d$-separation relation $A_1 \perp_d C_2\, |\, B_1 Y_1$ implies that
\begin{align}\label{express1}
P_{A_1 B_1 C_2 Y_1} = \frac{P_{A_1 B_1 Y_1} P_{C_2 B_1 Y_1} }{P_{B_1 Y_1}},
\end{align}
and the sets $\{A_1 B_1 Y_1\}, \{C_2 B_1 Y_1\},$ and $\{B_1 Y_1\}$ are injectable. The expressibility of $\{ A_1 C_2 Y_1\}$ then follows from the expressibility of $\{ A_1 B_1 C_2 Y_1\}$ and the fact that the distribution on the former can be obtained from the distribution on the latter by marginalization,
\begin{align}\label{express0}
P_{A_1 C_2 Y_1}(acy) = \sum_{b} P_{ A_1 B_1 C_2 Y_1}(abcy).
\end{align}

It follows that the distribution $P_{A_1 C_2 Y_1}$ in the inflation model associated to the Pienaar distribution can be computed by first writing down the distributions on the relevant injectable sets,
\begin{align}
\begin{split}
\label{injsnew}
P_{B_2 C_2 Y_2}(bc y)&= P^{\text{Pienaar}}_{BCY}(bc y), \\
P_{A_1 C_1 Y_1}(ac y)&= P^{\text{Pienaar}}_{ACY}(ac y),\\
P_{B_2 C_1 Y_2}(bc y)&= P^{\text{Pienaar}}_{BCY}(bc y),
\end{split}
\end{align} 
and from \cref{express0} and~\eqref{express1}, as well as the injectability of $\{A_1 B_1 Y_1\}, \{C_2 B_1 Y_1\},$ and $\{B_1 Y_1\}$, 
we infer that 
\begin{align}
P_{A_1 C_2 Y_1}(a c y) = \sum_{b} \frac{P^{\text{Pienaar}}_{A B Y}(a b y) P^{\text{Pienaar}}_{C B Y}(c b y) }{P^{\text{Pienaar}}_{B Y}(b y)}.
\label{express1old}
\end{align}

We are now in a position to derive a contradiction. 
Our derivation will begin by setting $Y_2=0$ and $Y_1=1$. It is therefore convenient to condition on $Y_1$ and $Y_2$ in the distributions of interest and set them equal to these values, and to express these in terms of the conditioned Pienaar distribution via \cref{injsnew},
\begin{align}
P_{B_2 C_2 |Y_2}(bc| 0)&= P^{\text{Pienaar}}_{BC|Y}(bc|0) \nonumber\\
P_{A_1 C_1 |Y_1}(ac|1)&= P^{\text{Pienaar}}_{AC|Y}(ac|1)\nonumber\\
P_{B_2 C_1 |Y_2}(bc|0)&= P^{\text{Pienaar}}_{BC|Y}(bc|0) 
\end{align} 
and similarly from \cref{express1old},
\begin{align}
P_{A_1 C_2 |Y_1}(ac|1) &=\sum_{b} \frac{P^{\text{Pienaar}}_{A B |Y}(a b| 1) P^{\text{Pienaar}}_{C B| Y}(c b |1) }{P^{\text{Pienaar}}_{B |Y}(b |1)}.
\end{align}
From these and \cref{eq:pienaardistroconditiona}, we infer
\begin{align}
P_{B_2 C_2 |Y_2}(\cdot \cdot| 0)&= \frac{1}{2}([00]+[11]),\label{inj1}\\
P_{A_1 C_1 |Y_1}(\cdot \cdot| 1)&= \frac{1}{2}([00]+[11]), \label{inj2}\\
P_{B_2 C_1 |Y_2}(\cdot \cdot| 0)&= \frac{1}{2}([00]+[11]),\label{inj3} \\
P_{A_1 C_2 |Y_1}(\cdot \cdot|1) &= \frac{1}{4}([00]+[01]+[10]+[11]).\label{inj4}
\end{align} 
Henceforth, we leave the condition that $Y_2=0$ and $Y_1=1$ implicit. 
From \cref{inj3}, we have 
\begin{align}
\text{With probability 1/2,}\; B_2=0\; \text{and}\;C_1=0.
\end{align}
From \cref{inj1}, we have
\begin{align}
\text{If}\; B_2=0\; \text{then}\;C_2=0.
\end{align}
From \cref{inj2}, we have
\begin{align}
\text{If}\; C_1=0\; \text{then}\;A_1=0.
\end{align}
These three statements imply that 
\begin{align}
\text{The probability that } C_2=0 \text{ and } A_1=0 \text{ is } \ge 1/2.
\end{align}
However, \cref{inj4} implies that the probability of $C_2=0\; \text{and}\;A_1=0$ is only $p=1/4$. 
We have therefore arrived at a contradiction. This establishes the incompatibility of the Pienaar distribution with HLP's causal structure \#16. Our reasoning is again a form of the Hardy-type arguments from \cref{sec:TSEM}.

\subsection{Deriving a Causal Compatibility Inequality for HLP's Causal Structure \#16}

We can also turn the above argument into an inequality. Using the methods of \cref{sec:TSEM}, it is straightforward to show that the assumption of a joint distribution on $\{ A_1 B_2 C_1 Y_1 Y_2\}$ implies the inequality on marginals,
\begin{align}\label{startingineq}
 \p[B_2 C_1Y_1 Y_2]{0010} \leq \p[B_2 C_2Y_1 Y_2]{0110} + \p[A_1 C_1 Y_1 Y_2]{10 10}+\p[A_1 C_2 Y_1 Y_2]{0010}.
\end{align}
From the following four ancestral independences in the inflated causal structure, $B_2 C_1 Y_2 \perp_d Y_1$, $B_2 C_2 Y_2 \perp_d Y_1$, $A_1 C_1 Y_1 \perp_d Y_2$, and $A_1 C_2 Y_1 \perp_d Y_2$, we infer, respectively, the following factorization conditions:
\begin{align}
\begin{split}
P_{B_2 C_1 Y_2 Y_1}=P_{B_2 C_1 Y_2}P_{Y_1},\\
P_{B_2 C_2 Y_2 Y_1}=P_{B_2 C_2 Y_2}P_{Y_1},\\
P_{A_1 C_1 Y_1 Y_2}=P_{A_1 C_1 Y_1 }P_{Y_2},\\
P_{A_1 C_2 Y_1 Y_2}=P_{A_1 C_2 Y_1 }P_{Y_2}.
\end{split}
\end{align}
Substituting these into \cref{startingineq}, we obtain:
\begin{align}\label{ntccineqinflationDAG}
 \p[B_2 C_1Y_2 ]{000}\p[Y_1]{1} \leq \p[B_2 C_2Y_2]{010} \p[ Y_1]{1} + \p[A_1 C_1 Y_1 ]{10 1}\p[Y_2]{0}+\p[A_1 C_2 Y_1]{001}\p[ Y_2]{0}.
\end{align}
This is a nontrivial causal compatibility inequality for the inflated causal structure. However, in this form, it cannot be translated into one for the observed variables in the original causal structure: the sets $\{B_2 C_1 Y_2 \}$, $\{ B_2 C_2 Y_2\}$ and $\{ A_1 C_1 Y_1\}$ are injectable, and the singleton sets $\{ Y_1\} $ and $\{Y_2\}$ are injectable (by the definition of inflation), the set $\{ A_1 C_2 Y_1 \}$ is merely expressible. Therefore, we must substitute the expression for $P_{A_1 C_2 Y_1}$ given by \cref{express0,express1} into \cref{ntccineqinflationDAG}, to obtain
\begin{align}\label{ntccineqinflationDAG2}
 \p[B_2 C_1Y_2 ]{000}\p[Y_1]{1} \leq \p[B_2 C_2 Y_2]{010} \p[ Y_1]{1} + \p[A_1 C_1 Y_1 ]{10 1}\p[Y_2]{0}+\sum_{b} \frac{P_{A_1 B_1 Y_1}(0 b 1) P_{B_1 C_2 Y_1}(0 b 1) }{P_{B_1 Y_1}(b 1)} \p[ Y_2]{0}.
\end{align}
This is also a nontrivial causal compatibility inequality for the inflated causal structure, but now it refers exclusively to distributions on injectable sets. As such, we can directly translate it into a nontrivial causal compatibility inequality for the original causal structure, namely, 
\begin{align}\label{ntccineqinflationDAG3}
 \p[B C Y ]{000}\p[Y ]{1} \leq \p[B C Y]{010} \p[ Y ]{1} + \p[A C Y ]{10 1}\p[Y]{0}+\sum_{b} \frac{P_{A B Y}(0 b 1) P_{ B C Y}(0 b 1) }{P_{B Y}(b 1)} \p[ Y]{0}.
\end{align}
Dividing by $\p[Y ]{0}\p[Y ]{1}$, and using the definition of conditional probabilities, 
this inequality can be expressed in the form
\begin{align}\label{ntccineqDAG}
 \p[B C |Y ]{00|0} \leq \p[B C| Y]{01|0} + \p[A C |Y ]{10 |1}+\sum_{b} \frac{P_{A B| Y}(0 b |1) P_{ B C |Y}(0 b |1) }{P_{B |Y}(b| 1)} .
\end{align}
This inequality is strong enough to witness the incompatibility of Pienaar's distribution \cref{eq:pienaardistro} with HLP's causal structure $\#16$.

\subsection{Certifying that Henson-Lal-Pusey's Causal Structures \#15 and \#20 are ``Interesting"} 

Any distribution $P_{ABCY}$ that is incompatible with HLP's causal structure \#16 is also incompatible with HLP's causal structures \#15 (\cref{fig:GDAG15} here) and \#20 (\cref{fig:GDAG20} here) because the causal models defined by HLP's causal structures \#15 and \#20 are included among the causal models defined by HLP's causal structure \#16 (\cref{fig:GDAG16} here). Consequently, \cref{ntccineqDAG} is also a valid causal compatibility inequality for HLP's causal structure \#15 and for HLP's causal structure \#20.

It follows that if one can find a distribution that exhibits all of the observable CI relations implied by \emph{either} of HLP's causal structures \#15 and \#20, namely, $C\indep Y$ (per \#15 \emph{and} \#16), $A\indep B\,|\,Y$ (per \#15), and $A\indep Y\,|\,B$ (per \#20), and which moreover is \emph{not} compatible with HLP's causal structure \#16, then this proves---in one go---that HLP's causal structures \#15, \#16 and \#20 are interesting. Any distribution $P_{ABCY}$ with the conditional\footnote{We take the definition of the conditional $P_{A B C|Y}$ from the distribution $P_{ABCY}$ as also implying $P_Y(0) > 0$ and $P_Y(1) > 0$.}
\begin{align}\label{eq:allcirespectingdistro}
 P_{A B C|Y}(a b c| y):=
 \begin{cases}\frac{1}{4}\parenths{[000]+[111]+[011]+[100]}&\text{if }\, y\eql 0, \\ \frac{1}{4}\parenths{[000]+[111]+[010]+[101]}&\text{if }\, y\eql 1,\end{cases}
\end{align}
achieves this because it satisfies the required CI relations while also violating \cref{ntccineqDAG}.

\section{The Copy Lemma and Non-Shannon type Entropic Inequalities}\label{sec:NonShannon}

The inflation technique may also be useful outside beyond causal inference. As we argue in the following, inflation is secretly what underlies the \tblue{Copy Lemma} in the derivation of non-Shannon type entropic inequalities~\cite[Chapter~15]{yeung_network_2008}. The following formulation of the Copy Lemma is the one of \citet{kaced_equivalence_2013}.

\begin{lemma}
	Let $A$, $B$ and $C$ be random variables with joint distribution $P_{ABC}$. Then there exists a fourth random variable $A'$ and joint distribution $P_{AA'BC}$ such that:
	\begin{enumerate}
		\item $P_{AB} = P_{A'B}$,
		\item $A' \indep AC \:|\: B$.
	\label{copylemma}
	\end{enumerate}
\end{lemma}

The proof via inflation is as follows.

\begin{proof}
Every joint distribution $P_{ABC}$ is compatible with the causal structure of \cref{fig:beforecopy}. This follows from the fact that one may take $X$ to be any \tblue{sufficient statistic} for the joint variable $(A,C)$ given $B$, such as $X := (A,B,C)$. Next, we consider the inflation of \cref{fig:beforecopy} depicted in~\cref{fig:aftercopy}. The maximal injectable sets are $\{ A_1 B_1 C_1\}$ and $\{A_2 B_1\}$. By \cref{mainlemma}, because $P_{ABC}$ is assumed to be compatible with~\cref{fig:beforecopy}, it follows that the family of marginals $\{ P_{A_1 B_1 C_1}, P_{A_2 B_1}\}$, where $P_{A_1 B_1 C_1}:= P_{A B C}$ and $P_{A_2 B_1} := P_{AB}$, is compatible with the inflation of~\cref{fig:aftercopy}. The resulting joint distribution $P_{A_1 A_2 B_1 C_1}$ has marginals $P_{A_1 B_1}= P_{A_2 B_1} =P_{AB}$ and satisfies the conditional independence relation $A_2 \indep A_1 C_1 \:|\: B_1$, since $A_2$ is $d$-separated from $A_1 C_1$ by $B_1$ in \cref{fig:aftercopy}.
\end{proof}

While it is also not hard to write down the distribution constructed in the proof explicitly as $P_{A_1 A_2 B_1 C_1} := P_{A_1 B_1 C_1} P_{A_2 B_1} P_{B_1}^{-1}$~\cite[Lemma~15.8]{yeung_network_2008}, the fact that one can reinterpret it using the inflation technique is significant. For one, all the non-Shannon type inequalities derived by \citet{zeger_2011_nonshannon} are obtained by applying some Shannon-type inequality to the distribution derived from the Copy Lemma. Our result shows, therefore, that one can understand these non-Shannon type inequalities for a causal structure as arising from Shannon-type inequalities applied to an inflated causal structure. We thus speculate that the inflation technique may be a more general-purpose tool for deriving non-Shannon-type entropic inequalities. A natural direction for future research is to explore whether more sophisticated applications of the inflation technique might result in \emph{new} examples of such inequalities. 

\begin{figure}[H]
\centering
\begin{minipage}[t]{0.4\linewidth}
\centering
\includegraphics[scale=1]{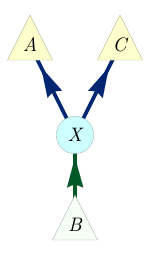}
\caption{A causal structure that is compatible with any distribution $P_{ABC}$.}\label{fig:beforecopy}
\end{minipage}
\hfill
\begin{minipage}[t]{0.4\linewidth}
\centering
\includegraphics[scale=1]{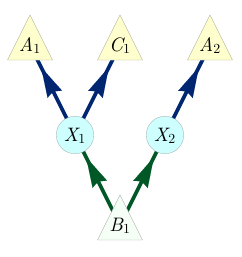}
\caption{An inflation of \cref{fig:beforecopy}.}\label{fig:aftercopy}
\end{minipage}
\end{figure}

\section{Causal Compatibility Inequalities for the Triangle Scenario in Machine-Readable Format}
\label{sec:38ineqs}

\cref{tab:machinereadable} lists the fifty two numerically irredundant polynomial inequalities resulting from consistent marginals of the Spiral inflation of \cref{fig:Tri222}. Stronger inequalities can be derived be considering larger inflations, such as the Web inflation of \cref{fig:TriFullDouble}. Each row in the table specifies the coefficient of the corresponding correlator monomial. As noted previously, these inequalities also follow from the hypergraph transversals technique per \cref{sec:TSEM}. 

\begin{table*}[hb]\centering\caption{A machine-readable and closed-under-symmetries version of the table in \cref{sec:CCineqs}. 
 }\label{tab:machinereadable}
\resizebox{\textwidth}{!}{
\begin{tabular}{c@{\hspace{1em}}ccc@{\hspace{1em}}ccc@{\hspace{1em}}c@{\hspace{1em}}ccc@{\hspace{1em}}ccc@{\hspace{1em}}c} 
constant & \(\expec{A}\) & \(\expec{B}\) & \(\expec{C}\) & \(\expec{A B}\) & \(\expec{A C}\) & \(\expec{B C}\) & \(\expec{A B C}\) & \(\expec{A} \expec{B}\) & \(\expec{A} \expec{C}\) & \(\expec{B} \expec{C}\) & \(\expec{A} \expec{B C}\) & \(\expec{A C} \expec{B}\) & \(\expec{A B} \expec{C}\) &
 \(\expec{A} \expec{B} \expec{C}\) \\\bottomrule
 1 & 0 & 0 & 0 & -1 & -1 & 0 & 0 & 0 & 0 & 1 & 0 & 0 & 0 & 0 \\
 1 & 0 & 0 & 0 & -1 & 1 & 0 & 0 & 0 & 0 & -1 & 0 & 0 & 0 & 0 \\
 1 & 0 & 0 & 0 & 1 & -1 & 0 & 0 & 0 & 0 & -1 & 0 & 0 & 0 & 0 \\
 1 & 0 & 0 & 0 & 1 & 1 & 0 & 0 & 0 & 0 & 1 & 0 & 0 & 0 & 0 \\
 1 & 0 & 0 & 0 & -1 & 0 & -1 & 0 & 0 & 1 & 0 & 0 & 0 & 0 & 0 \\
 1 & 0 & 0 & 0 & -1 & 0 & 1 & 0 & 0 & -1 & 0 & 0 & 0 & 0 & 0 \\
 1 & 0 & 0 & 0 & 1 & 0 & -1 & 0 & 0 & -1 & 0 & 0 & 0 & 0 & 0 \\
 1 & 0 & 0 & 0 & 1 & 0 & 1 & 0 & 0 & 1 & 0 & 0 & 0 & 0 & 0 \\
 1 & 0 & 0 & 0 & 0 & -1 & -1 & 0 & 1 & 0 & 0 & 0 & 0 & 0 & 0 \\
 1 & 0 & 0 & 0 & 0 & -1 & 1 & 0 & -1 & 0 & 0 & 0 & 0 & 0 & 0 \\
 1 & 0 & 0 & 0 & 0 & 1 & -1 & 0 & -1 & 0 & 0 & 0 & 0 & 0 & 0 \\
 1 & 0 & 0 & 0 & 0 & 1 & 1 & 0 & 1 & 0 & 0 & 0 & 0 & 0 & 0 \\
 3 & -1 & -1 & -1 & 2 & 2 & 2 & 1 & 1 & 1 & 1 & -1 & -1 & -1 & 1 \\
 3 & -1 & -1 & 1 & 2 & -2 & -2 & -1 & 1 & -1 & -1 & 1 & 1 & 1 & -1 \\
 3 & -1 & 1 & -1 & -2 & 2 & -2 & -1 & -1 & 1 & -1 & 1 & 1 & 1 & -1 \\
 3 & -1 & 1 & 1 & -2 & -2 & 2 & 1 & -1 & -1 & 1 & -1 & -1 & -1 & 1 \\
 3 & 1 & -1 & -1 & -2 & -2 & 2 & -1 & -1 & -1 & 1 & 1 & 1 & 1 & -1 \\
 3 & 1 & -1 & 1 & -2 & 2 & -2 & 1 & -1 & 1 & -1 & -1 & -1 & -1 & 1 \\
 3 & 1 & 1 & -1 & 2 & -2 & -2 & 1 & 1 & -1 & -1 & -1 & -1 & -1 & 1 \\
 3 & 1 & 1 & 1 & 2 & 2 & 2 & -1 & 1 & 1 & 1 & 1 & 1 & 1 & -1 \\
 4 & -2 & 0 & 0 & -3 & -2 & -2 & 1 & 1 & 0 & 2 & 1 & 1 & 0 & -1 \\
 4 & -2 & 0 & 0 & -3 & 2 & 2 & -1 & 1 & 0 & -2 & -1 & -1 & 0 & 1 \\
 4 & -2 & 0 & 0 & 3 & -2 & 2 & -1 & -1 & 0 & -2 & -1 & -1 & 0 & 1 \\
 4 & -2 & 0 & 0 & 3 & 2 & -2 & 1 & -1 & 0 & 2 & 1 & 1 & 0 & -1 \\
 4 & 2 & 0 & 0 & -3 & -2 & -2 & -1 & 1 & 0 & 2 & -1 & -1 & 0 & 1 \\
 4 & 2 & 0 & 0 & -3 & 2 & 2 & 1 & 1 & 0 & -2 & 1 & 1 & 0 & -1 \\
 4 & 2 & 0 & 0 & 3 & -2 & 2 & 1 & -1 & 0 & -2 & 1 & 1 & 0 & -1 \\
 4 & 2 & 0 & 0 & 3 & 2 & -2 & -1 & -1 & 0 & 2 & -1 & -1 & 0 & 1 \\
 4 & 0 & -2 & 0 & -2 & -2 & -3 & 1 & 0 & 2 & 1 & 0 & 1 & 1 & -1 \\
 4 & 0 & -2 & 0 & -2 & 2 & 3 & -1 & 0 & -2 & -1 & 0 & -1 & -1 & 1 \\
 4 & 0 & -2 & 0 & 2 & -2 & 3 & 1 & 0 & 2 & -1 & 0 & 1 & 1 & -1 \\
 4 & 0 & -2 & 0 & 2 & 2 & -3 & -1 & 0 & -2 & 1 & 0 & -1 & -1 & 1 \\
 4 & 0 & 2 & 0 & -2 & -2 & -3 & -1 & 0 & 2 & 1 & 0 & -1 & -1 & 1 \\
 4 & 0 & 2 & 0 & -2 & 2 & 3 & 1 & 0 & -2 & -1 & 0 & 1 & 1 & -1 \\
 4 & 0 & 2 & 0 & 2 & -2 & 3 & -1 & 0 & 2 & -1 & 0 & -1 & -1 & 1 \\
 4 & 0 & 2 & 0 & 2 & 2 & -3 & 1 & 0 & -2 & 1 & 0 & 1 & 1 & -1 \\
 4 & 0 & 0 & -2 & -2 & -3 & -2 & 1 & 2 & 1 & 0 & 1 & 0 & 1 & -1 \\
 4 & 0 & 0 & -2 & -2 & 3 & 2 & 1 & 2 & -1 & 0 & 1 & 0 & 1 & -1 \\
 4 & 0 & 0 & -2 & 2 & -3 & 2 & -1 & -2 & 1 & 0 & -1 & 0 & -1 & 1 \\
 4 & 0 & 0 & -2 & 2 & 3 & -2 & -1 & -2 & -1 & 0 & -1 & 0 & -1 & 1 \\
 4 & 0 & 0 & 2 & -2 & -3 & -2 & -1 & 2 & 1 & 0 & -1 & 0 & -1 & 1 \\
 4 & 0 & 0 & 2 & -2 & 3 & 2 & -1 & 2 & -1 & 0 & -1 & 0 & -1 & 1 \\
 4 & 0 & 0 & 2 & 2 & -3 & 2 & 1 & -2 & 1 & 0 & 1 & 0 & 1 & -1 \\
 4 & 0 & 0 & 2 & 2 & 3 & -2 & 1 & -2 & -1 & 0 & 1 & 0 & 1 & -1 \\
 4 & 0 & 0 & 0 & -2 & -2 & -2 & -1 & 2 & 2 & 2 & -1 & -1 & -1 & 0 \\
 4 & 0 & 0 & 0 & -2 & -2 & -2 & 1 & 2 & 2 & 2 & 1 & 1 & 1 & 0 \\
 4 & 0 & 0 & 0 & -2 & 2 & 2 & -1 & 2 & -2 & -2 & -1 & -1 & -1 & 0 \\
 4 & 0 & 0 & 0 & -2 & 2 & 2 & 1 & 2 & -2 & -2 & 1 & 1 & 1 & 0 \\
 4 & 0 & 0 & 0 & 2 & -2 & 2 & -1 & -2 & 2 & -2 & -1 & -1 & -1 & 0 \\
 4 & 0 & 0 & 0 & 2 & -2 & 2 & 1 & -2 & 2 & -2 & 1 & 1 & 1 & 0 \\
 4 & 0 & 0 & 0 & 2 & 2 & -2 & -1 & -2 & -2 & 2 & -1 & -1 & -1 & 0 \\
 4 & 0 & 0 & 0 & 2 & 2 & -2 & 1 & -2 & -2 & 2 & 1 & 1 & 1 & 0 \\
\end{tabular}}
\end{table*}

\section{Recovering the Bell Inequalities from the Inflation Technique}
\label{sec:Bellscenarios}

To further illustrate the power of the inflation technique, we now demonstrate how to recover all Bell inequalities~\cite{Brunner2013Bell,bell1966lhvm,CHSHOriginal} via our method. To keep things simple we only discuss the case of a bipartite Bell scenario with two values for both ``settings'' and ``outcome'' variables, but the case of more parties and/or more values per settings or outcome variable is totally analogous.

The causal structure associated to the Bell \cite{bell1964einstein,Brunner2013Bell,bell1966lhvm,CHSHOriginal} scenario [\citealp{pusey2014gdag}~(Fig.~E\#2), \citealp{WoodSpekkens}~(Fig.~19), \citealp{chaves2014novel}~(Fig.~1), \citealp{BeyondBellII}~(Fig.~1), \citealp{wolfe2015nonconvexity}~(Fig.~2b), \citealp{steeg2011relaxation}~(Fig.~2)] is depicted in \cref{fig:NewBellDAG1}. The observed variables are $A,B,X,Y$, and $\Lambda$ is the latent common cause of $A$ and $B$. One traditionally works with the conditional distribution $P_{AB|XY}$, to be understood as an array of distributions indexed by the possible values of $X$ and $Y$, instead of with the original distribution $P_{ABXY}$, which is what we do.

In the inflation of \cref{fig:BellDagCopy1}, the maximal ai-expressible sets are
\begin{align}\begin{split}
	\label{eq:bellcontexts}
&\brackets{A_1 B_1 X_1 X_2 Y_1 Y_2}, \qquad
\brackets{A_1 B_2 X_1 X_2 Y_2 Y_2}, \qquad
\brackets{A_2 B_1 X_1 X_2 Y_2 Y_2}, \qquad
\brackets{A_2 B_2 X_1 X_2 Y_2 Y_2},
\end{split}\end{align}
where notably every maximal ai-expressible set contains all ``settings'' variables $X_1$ to $Y_2$. The marginal distributions on these ai-expressible sets are then specified by the original observed distribution via
\begin{align}\begin{split}&\forall{a b x_1 x_2 y_1 y_2}:\; \begin{cases}
	P_{A_1 B_1 X_1 X_2 Y_1 Y_2}(a b x_1 x_2 y_1 y_2) = P_{A B X Y}(a b x_1 y_1) P_X(x_2) P_Y(y_2), \\
	P_{A_1 B_2 X_1 X_2 Y_1 Y_2}(a b x_1 x_2 y_1 y_2) = P_{A B X Y}(a b x_1 y_2) P_X(x_2) P_Y(y_1), \\
	P_{A_2 B_1 X_1 X_2 Y_1 Y_2}(a b x_1 x_2 y_1 y_2) = P_{A B X Y}(a b x_2 y_1) P_X(x_1) P_Y(y_2), \\
	P_{A_2 B_2 X_1 X_2 Y_1 Y_2}(a b x_1 x_2 y_1 y_2) = P_{A B X Y}(a b x_2 y_2) P_X(x_1) P_Y(y_1), \\
\hspace{2.5pc}	P_{X_1 X_2 Y_1 Y_2}(x_1 x_2 y_1 y_2) = P_X(x_1) P_X(x_2) P_Y(y_1) P_Y(y_2).
\end{cases}\end{split}\end{align}
By dividing each of the first four equations by the fifth, we obtain
\begin{align}\begin{split}
	\label{eq:bellfactor}
	\forall{a b x_1 x_2 y_1 y_2}:\; \begin{cases}
	P_{A_1 B_1 | X_1 X_2 Y_1 Y_2}(a b | x_1 x_2 y_1 y_2) = P_{A B | X Y}(a b | x_1 y_1), \\
	P_{A_1 B_2 | X_1 X_2 Y_1 Y_2}(a b | x_1 x_2 y_1 y_2) = P_{A B | X Y}(a b | x_1 y_2), \\
	P_{A_2 B_1 | X_1 X_2 Y_1 Y_2}(a b | x_1 x_2 y_1 y_2) = P_{A B | X Y}(a b | x_2 y_1), \\
	P_{A_2 B_2 | X_1 X_2 Y_1 Y_2}(a b | x_1 x_2 y_1 y_2) = P_{A B | X Y}(a b | x_2 y_2).
\end{cases}\end{split}\end{align}
The existence of a joint distribution of all six variables---i.e.~the existence of a solution to the marginal problem---implies in particular
\begin{align}
	\forall{a b x_1 x_2 y_1 y_2}: \quad P_{A_1 B_1 | X_1 X_2 Y_1 Y_2}(a b | x_1 x_2 y_1 y_2) = \sum\nolimits_{a',b'} P_{A_1 A_2 B_1 B_2 | X_1 X_2 Y_1 Y_2}(a a' b b'|x_1 x_2 y_1 y_2),
\end{align}
and similarly for the other three conditional distributions under consideration. For compatibility with the Bell scenario, \cref{eq:bellfactor} therefore implies that the original distribution must satisfy in particular
\begin{align}\begin{split}\label{eq:finalBellstep}\forall{a b}:\; \begin{cases}
	P_{A B | X Y}(a b | 0 0) = \sum\nolimits_{a',b'} P_{A_1 A_2 B_1 B_2| X_1 X_2 Y_1 Y_2}(a a' b b'|0101) \\
	P_{A B | X Y}(a b | 1 0) = \sum\nolimits_{a',b'} P_{A_1 A_2 B_1 B_2| X_1 X_2 Y_1 Y_2}(a' a b b'|0101) \\
	P_{A B | X Y}(a b | 0 1) = \sum\nolimits_{a',b'} P_{A_1 A_2 B_1 B_2| X_1 X_2 Y_1 Y_2}(a a' b' b|0101) \\
	P_{A B | X Y}(a b | 1 1) = \sum\nolimits_{a',b'} P_{A_1 A_2 B_1 B_2| X_1 X_2 Y_1 Y_2}(a' a b' b|0101)
\end{cases}\end{split}\end{align}
The possibility to write the conditional probabilities in the Bell scenario in this form is equivalent to the existence of a latent variable model, as noted in Fine's theorem~\cite{FineTheorem}. Thus, the existence of a solution to our marginal problem implies the existence of a latent variable model for the original distribution; the converse follows from our \cref{mainlemma}. Hence the inflation of~\cref{fig:BellDagCopy1} provides necessary and sufficient conditions for the compatibility of the original distribution with the Bell scenario.

Moreover, it is possible to describe the marginal polytope over the ai-expressible sets of~\cref{eq:bellcontexts}, resulting in a concrete correspondence between tight Bell inequalities and the facets of our marginal polytope. This is based on the observation that the ``settings'' variables $X_1$ to $Y_2$ occur in all four contexts. The marginal polytope lives in $\oplus_{i=1}^4 \mathbb{R}^{2^6} = \oplus_{i=1}^4 (\mathbb{R}^2)^{\otimes 6}$, where each tensor factor has basis vectors corresponding to the two possible outcomes of each variable, and the direct summands enumerate the four contexts. The polytope is given as the convex hull of the points
\begin{align*}
	(e_{A_1} & \otimes e_{B_1} \otimes e_{X_1} \otimes e_{X_2} \otimes e_{Y_1} \otimes e_{Y_2}) \\[-3pt]
	\oplus\: (e_{A_1} & \otimes e_{B_2} \otimes e_{X_1} \otimes e_{X_2} \otimes e_{Y_1} \otimes e_{Y_2}) \\[-3pt]
	\oplus\: (e_{A_2} & \otimes e_{B_1} \otimes e_{X_1} \otimes e_{X_2} \otimes e_{Y_1} \otimes e_{Y_2}) \\[-3pt]
	\oplus\: (e_{A_2} & \otimes e_{B_2} \otimes e_{X_1} \otimes e_{X_2} \otimes e_{Y_1} \otimes e_{Y_2}),
\end{align*}
where all six variables range over their possible values. Since the last four tensor factors occur in every direct summand in exactly the same way, we can also write such a polytope vertex as
\[
	\left[ (e_{A_1} \otimes e_{B_1}) \oplus (e_{A_1} \otimes e_{B_2}) \oplus (e_{A_2} \otimes e_{B_1}) \oplus (e_{A_2} \otimes e_{B_2})\right] \otimes \left[ e_{X_1} \otimes e_{X_2} \otimes e_{Y_1} \otimes e_{Y_2}\right]
\]
in $\big(\oplus_{i=1}^4 \mathbb{R}^{2^2}\big)\otimes \mathbb{R}^{2^4}$. Now since the first four variables in the first tensor factor vary completely independently of the latter four variables in the second tensor factor, the resulting polytope will be precisely the tensor product~\cite{namioka_tensor_1969,bogart_hom_2013} of two polytopes: first, the convex hull of all points of the form
\[
	(e_{A_1} \otimes e_{B_1}) \oplus (e_{A_1} \otimes e_{B_2}) \oplus (e_{A_2} \otimes e_{B_1}) \oplus (e_{A_2} \otimes e_{B_2}),
\]
and second the convex hull of all $e_{X_1} \otimes e_{X_2} \otimes e_{Y_1} \otimes e_{Y_2}$. While the latter polytope is just the standard probability simplex in $\mathbb{R}^8$, the former polytope is precisely the ``local polytope'' or ``Bell polytope'' that is traditionally used in the context of Bell scenarios~\cite[Sec.~II.B]{Brunner2013Bell}. This implies that the facets of our marginal polytope are precisely the pairs consisting of a facet of the Bell polytope and a facet of the simplex, the latter of which are only the nonnegativity of probability inequalities like $P_{X_1X_2Y_1Y_2}(0101)\geq 0$. For example, in this way we obtain one version of the CHSH inequality~\cite{CHSHOriginal} as a facet of our marginal polytope,
\[
	\sum_{a,b,x,y} (-1)^{a + b + xy} P_{A_x B_y X_1 X_2 Y_1 Y_2}(a b 0 1 0 1) \leq 2 P_{X_1 X_2 Y_1 Y_2}(0101).
\]
This translates into the standard form of the CHSH inequality as follows. Upon using~\cref{eq:bellfactor}, the inequality becomes
\begin{align*}
	\sum_{a,b} (-1)^{a + b} \big( & P_{A B X Y}(ab00)P_X(1)P_Y(1) + P_{A B X Y}(ab01)P_X(1)P_Y(0) \\[-12pt]
	& + P_{A B X Y}(ab10)P_X(0)P_Y(1) - P_{A B X Y}(ab11)P_X(0)P_Y(0) \big) \leq P_X(0)P_X(1)P_Y(0)P_Y(1),
\end{align*}
so that dividing by the right-hand side results in one of the conventional forms of the CHSH inequality,
\[
	\sum_{a,b} (-1)^{a + b} \left( P_{AB|XY}(ab|00) + P_{AB|XY}(ab|01) + P_{AB|XY}(ab|10) - P_{AB|XY}(ab|11) \right) \leq 2.
\]
In conclusion, the inflation technique is powerful enough to get a precise characterization of all distributions compatible with the Bell causal structure, and our technique for generating polynomial inequalities through solving the marginal constraint problem recovers all Bell inequalities.

Some Bell inequalities may also be derived using the hypergraph transversals technique discussed in \cref{sec:TSEM}. For example, the inequality
\begin{align}\label{eq:preBell}\begin{split}
& \p[A_1 B_1 X_1 Y_1]{0000}\p[X_2]{1} \p[Y_2]{1} \\
&\leq
 \p[A_1 B_2 X_1 Y_2]{0001}\p[X_2]{1} \p[Y_1]{0} +\p[A_2 B_1 X_2 Y_1]{0010}\p[X_1]{0}\p[Y_2]{1}+ \p[A_2 B_2 X_2 Y_2]{1111}\p[X_1]{0} \p[Y_1]{0}
\end{split}\end{align}
is the inflationary precursor of the Bell inequality
\begin{align}\label{eq:aBell}
 \p[A B | X Y]{00|00} &\leq \p[A B | X Y]{00|01} +\p[A B | X Y]{00|10}+ \p[A B | X Y]{11|11},
\end{align}
as \cref{eq:aBell} is obtained from \cref{eq:preBell} by dividing both sides by $\p[X_1 Y_1 X_2 Y_2]{0011}=\p[X_1]{0} \p[Y_2]{0}\p[X_2]{1} \p[Y_2]{1}$ and then dropping copy indices. On the other hand, \cref{eq:preBell} follows directly from factorization relations on ai-expressible sets and the tautology
\begin{align}\begin{split}
	[ \mgreen{A_1 \eql 0},\, & \mgreen{B_1 \eql 0}, \mgreen{X_1 \eql 0}, \mgreen{Y_1\eql 0}, \mgreen{X_2 \eql 1}, \mgreen{Y_2 \eql 1}]
 \implies \begin{array}{r}
	[ \mgreen{A_1 \eql 0}, B_2 \eql 0, \mgreen{X_1 \eql 0}, \mgreen{Y_1\eql 0}, \mgreen{X_2 \eql 1}, \mgreen{Y_2 \eql 1}] \\
	\mathrel{\lor} [ A_2 \eql 0, \mgreen{B_1 \eql 0}, \mgreen{X_1 \eql 0}, \mgreen{Y_1\eql 0}, \mgreen{X_2 \eql 1}, \mgreen{Y_2 \eql 1}] \\
	\mathrel{\lor} [ A_2 \eql 1, B_2 \eql 1, \mgreen{X_1 \eql 0}, \mgreen{Y_1\eql 0}, \mgreen{X_2 \eql 1}, \mgreen{Y_2 \eql 1}]
	\end{array}
\end{split}\end{align}
which corresponds to the original ``Hardy paradox''~\cite{L.Hardy:PRL:1665} in our notation.

\let\cleardoublepage\clearpage
\newpage
\setlength{\bibsep}{2pt plus 2pt minus 1pt}
\bibliographystyle{apsrev4-1}
\nocite{apsrev41Control}
\bibliography{InflationTechniqueJCI}
\end{document}